\documentclass[12pt]{extarticle}
\usepackage[margin=1in]{geometry}
\usepackage{setspace}
\usepackage{natbib}
\setcitestyle{authoryear,open={(},close={)}}
\usepackage{url}
\usepackage{graphicx} 
\usepackage{amsmath}
\usepackage{amssymb}
\usepackage{mathtools}
\usepackage{amsthm}
\usepackage{booktabs}
\usepackage{multirow}
\usepackage{dcolumn}
\usepackage[graphicx]{realboxes}
\usepackage{adjustbox}
\usepackage{pdflscape}
\usepackage{nicematrix}
\usepackage{tablefootnote}
\usepackage{lscape}
\usepackage{tabularray}
\usepackage{float}
\usepackage{graphicx}
\usepackage{longtable}
\usepackage{array}
\usepackage{codehigh}
\usepackage{caption}
\usepackage[normalem]{ulem}
\UseTblrLibrary{booktabs}
\UseTblrLibrary{siunitx}

\usepackage{tikz}
\usepackage{xcolor}
\usetikzlibrary{positioning, shapes.geometric, arrows, calc}
\definecolor{colorOne}{HTML}{2E7D32}   
\definecolor{colorTwo}{HTML}{A5D6A7}   
\definecolor{colorThree}{HTML}{FFFFFF} 

\usepackage{times}

\newcommand{\bbR}[0]{\mathbb{R}}

\newcommand{\Cov}{\operatorname{Cov}}
\newcommand{\Var}{\operatorname{Var}}
\newcommand{\E}{\mathbb{E}}

\newtheorem{theorem}{Theorem}[section]

\newtheorem{definition}{Definition}[section] 
\newtheorem{proposition}{Proposition}[section]

\title{Early Investment Snowballs for Low-Skill Children: \\ \large Average Marginal Returns and the Dynamics of Skill Formation}
\author{Liam Puknys \\ University of Maryland, College Park}
\date{October 16, 2026}

\begin{document}

\begin{singlespace}
    \maketitle
\end{singlespace}

\begin{abstract}
    \begin{singlespace}
    Program evaluations routinely find larger effects of investment in low-skill children, while structural models of skill formation typically find that skill and investment are complements. I reconcile these findings by introducing the Average Marginal Return (AMR), which averages the marginal return to an input over the distribution of other inputs, providing a structural analogue to the average treatment effects estimated in program evaluations. The AMR reveals why input scales and distributions matter for the economic interpretation of structural estimates, why reduced-form interaction coefficients need not identify technological complementarity, and why dynamic complementarity can be difficult to detect in reduced-form work. I estimate AMRs to parental investment using a dynamic factor model of cognitive and noncognitive skill formation for ages 5–14 in the Children of the NLSY79. I accommodate binary and Likert measures using ordered probits, showing that treating them as continuous substantially biases estimates, and introduce effective factors to extend the scale-invariant normalization of \citet{freyberger_2025} to models with multiple skills per period. Investment in early childhood generates larger long-run gains in cognitive skill than investment at later ages. Although skill and investment are complements, the AMR to early investment is initially flat across skill levels but becomes sharply decreasing in initial skill over time. By ages 13–14, children at the 10th percentile of early cognitive skill gain roughly twice as much from early investment as children at the 90th.
    \end{singlespace}
\end{abstract}

\begin{singlespace}
    \textbf{Conferences and Presentations:} I would like to thank participants at the following conferences and presentations for feedback.
    \begin{itemize}
        \item 10/10/2026, Washington University in St. Louis, Economics Graduate Student Conference: Orazio Attanasio, Martin Garcia-Vazquez, Jaehong Lim, Hyun Soo Suh
        \item 8/8/2026, Bureau of Labor Statistics: Allison Aughinbaugh
        \item 6/23/2026, University of Chicago, Center for the Economics of Human Development: James Heckman, Tamara McGavock
        \item 6/23/2026, Federal Reserve Bank of Chicago: Mackenzie Alston, Kristin Butcher, Camilo Garcia-Jimeno, Kelli Marquardt, Derek Neal, Aryan Safi
        \item 4/10/2026, Washington Area Meeting of Economic Scientists: Elise Jang, Garrett Jones, Emel Feliz-Ozbay, Erkut Ozbay
    \end{itemize}
    \textbf{Additional Acknowledgments:} I would like to thank my advisors, Judith Hellerstein and Sergio Urzua, for their mentorship and helpful comments in writing this paper. I would also like to thank Diane Schanzenbach, Eric Chyn, Petra Todd, Andrew Foster for comments when they visited UMD. Finally I'd like to thank Guido Kuersteiner, Daniel Reck, Nolan Pope, Daniel Agness, Ethan Kaplan, John Shea and Mark Wilkerson for further helpful comments.  I used Claude and ChatGPT for research assistance.
\end{singlespace}

\section{Introduction}

There is a puzzle at the heart of child development. Reduced form studies of program evaluations overwhelmingly find that investment in children has larger effects for low-skill and disadvantaged students \citep{bitler_hoynes_domina2014, morris_et_al_2018}. However, structural studies frequently find that child skill and investment in children are complements \citep{attanasio_meghir_nix_2020, cunha_heckman_schennach2010, agostinelli_wiswall_2025}.\footnote{\cite{agostinelli_wiswall_2025}'s negative complementarity in log skills still corresponds to complementarity in levels of skill.} Notably, structural studies of program evaluations like \cite{attanasio_cattan_meghir_rubio_codina2020} find that skill and investment are complements as well, so the divide seems to be a structural-reduced form divide rather than a program analysis-observational study divide. 

In this paper, I reconcile structural and reduced form findings by introducing a new parameter, the Average Marginal Return (AMR). I show that the AMR is different from parameters reported in prior structural studies. It is the marginal counterpart to the average treatment effects reported in reduced form program analyses. 

I show how the AMR can be used to shed light on a number of puzzles facing the child development literature. It provides insight into seemingly innocuous changes in input scales can reverse the economic implications of structural models, into how interaction coefficients in linear models combine complementarity with changes in marginal returns, and into why dynamic complementarity, a leading paradigm in child development, can be difficult to detect in reduced-form studies. 

I then the use dynamic factor framework from \cite{cunha_heckman_schennach2010} to estimate the AMR to investment for children of different ages and skills. To enable this, I model binary and Likert proxies with ordered probit models, showing that treating discrete proxies as continuous is a major source of bias, and introduce the effective factor framework to extend \cite{freyberger_2025}'s scale-invariance normalization scheme to the case where multiple factor outputs are being produced each period. 

I show that the AMR to investment in early childhood is roughly flat in the short term, despite investment and skill being strong complements, and goes to strongly downward sloping in the long term, so that by high school, children and the 10th percentile of the early childhood cognitive skill distribution benefit around twice as much from early investment as children at the 90th percentile of the early childhood cognitive skill distribution do. 

I also show that early investment in children is highly efficient and has higher long-run effects on cognitive skill than investment in later ages. This contrasts with the effects of investment in noncognitive skill, which are highest in adolescence.

\section{General Framework}\label{sec:framework}

Marginal returns in models of skill formation are usually reported as properties of the production function alone. But the policy-relevant marginal return to an input depends on both the technology and the joint distribution of inputs, and the two objects can disagree even in sign. This section defines both objects in a simple model and shows how they come apart.

To illustrate the point of this paper, I will begin with a simple model of human capital production and define several concepts up front that will make understanding the subject of this paper far easier. Suppose that children's human capital (which I use interchangeably with the child's skill here) $\theta$ at time $t$ takes one of three values, $\theta_t\in \{0.5, 1, 2\}$. The child's human capital in period $t+1$ is given by the deterministic CES production function
\begin{equation}\label{eq:basic_CES}
    \theta_{t+1} = \left(\theta_t^{-1} + I_t^{-1}\right)^{-1},
\end{equation}
where $I_t$ is investment in the child by the child's parents at time $t$. Suppose that $I_t$ also takes three values, so that $I_t\in \{0.2, 1, 4.5\}$. A structural economist correctly estimating this model will derive both the functional form and the parameters of Equation (\ref{eq:basic_CES}), as well as what I will call in this paper the \textbf{Beckerian Return Schedule}, which tells you for a given child, what the elasticity of skill next period is with respect to investment this period. Formally, the Beckerian Return Schedule is a mapping $f:\Theta_t \times \mathcal{I}_t\to \mathbb{R}$ where $f(\bar{\theta}, i)=\left.\varepsilon_{I}\right|_{\theta_t = \bar{\theta}, I_t=i}$, where $\varepsilon_I$ is the elasticity of skill in period $t+1$ with respect to investment in time $t$. In this simple model, the Beckerian Return Schedule is given by Figure \ref{fig:BRS_simple}.
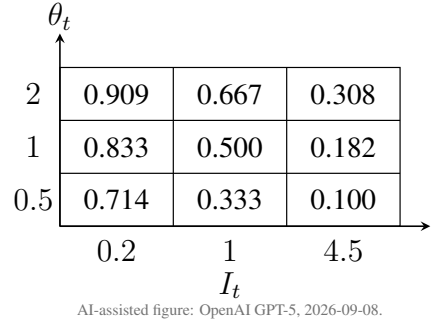
\begin{figure}
    \centering
\begin{adjustbox}{max width=\linewidth,center}
\begin{tikzpicture}[
  >=stealth,
  gridline/.style={draw=black,line width=0.45pt},
  axisline/.style={->,draw=black,line width=0.6pt},
  tablevalue/.style={font=\small}
]
  \draw[gridline] (0,0) rectangle (4.5,2.1);
  \draw[gridline] (1.5,0) -- (1.5,2.1);
  \draw[gridline] (3.0,0) -- (3.0,2.1);
  \draw[gridline] (0,0.7) -- (4.5,0.7);
  \draw[gridline] (0,1.4) -- (4.5,1.4);

  \draw[axisline] (0,0) -- (4.9,0);
  \draw[axisline] (0,0) -- (0,2.55);

  \node at (-0.35,1.75) {$2$};
  \node at (-0.35,1.05) {$1$};
  \node at (-0.35,0.35) {$0.5$};
  \node at (0.75,-0.32) {$0.2$};
  \node at (2.25,-0.32) {$1$};
  \node at (3.75,-0.32) {$4.5$};
  \node at (0,2.78) {$\theta_t$};
  \node at (2.25,-0.78) {$I_t$};

  \node[tablevalue] at (0.75,1.75) {0.909};
  \node[tablevalue] at (2.25,1.75) {0.667};
  \node[tablevalue] at (3.75,1.75) {0.308};
  \node[tablevalue] at (0.75,1.05) {0.833};
  \node[tablevalue] at (2.25,1.05) {0.500};
  \node[tablevalue] at (3.75,1.05) {0.182};
  \node[tablevalue] at (0.75,0.35) {0.714};
  \node[tablevalue] at (2.25,0.35) {0.333};
  \node[tablevalue] at (3.75,0.35) {0.100};

  \node[font=\tiny,text=black!55] at (2.25,-1.12)
    {AI-assisted figure: OpenAI GPT-5, 2026-09-08.};
\end{tikzpicture}
\end{adjustbox}
    \caption{The Beckerian Return Schedule}
    \label{fig:BRS_simple}
\end{figure}
Notably, the Beckerian Return Schedule is defined independently of the joint distribution of $\theta_t$ and $I_t$. Let that joint distribution be given by Figure \ref{fig:joint_dist_simple}.
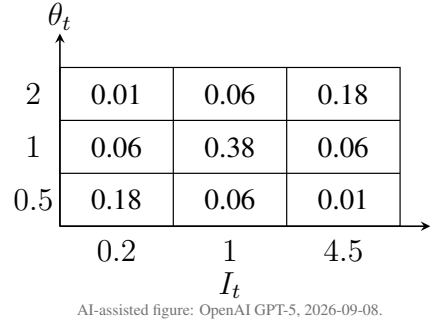
\begin{figure}
    \centering
\begin{adjustbox}{max width=\linewidth,center}
\begin{tikzpicture}[
  >=stealth,
  gridline/.style={draw=black,line width=0.45pt},
  axisline/.style={->,draw=black,line width=0.6pt},
  tablevalue/.style={font=\small}
]
  \draw[gridline] (0,0) rectangle (4.5,2.1);
  \draw[gridline] (1.5,0) -- (1.5,2.1);
  \draw[gridline] (3.0,0) -- (3.0,2.1);
  \draw[gridline] (0,0.7) -- (4.5,0.7);
  \draw[gridline] (0,1.4) -- (4.5,1.4);

  \draw[axisline] (0,0) -- (4.9,0);
  \draw[axisline] (0,0) -- (0,2.55);

  \node at (-0.35,1.75) {$2$};
  \node at (-0.35,1.05) {$1$};
  \node at (-0.35,0.35) {$0.5$};
  \node at (0.75,-0.32) {$0.2$};
  \node at (2.25,-0.32) {$1$};
  \node at (3.75,-0.32) {$4.5$};
  \node at (0,2.78) {$\theta_t$};
  \node at (2.25,-0.78) {$I_t$};

  \node[tablevalue] at (0.75,1.75) {0.01};
  \node[tablevalue] at (2.25,1.75) {0.06};
  \node[tablevalue] at (3.75,1.75) {0.18};
  \node[tablevalue] at (0.75,1.05) {0.06};
  \node[tablevalue] at (2.25,1.05) {0.38};
  \node[tablevalue] at (3.75,1.05) {0.06};
  \node[tablevalue] at (0.75,0.35) {0.18};
  \node[tablevalue] at (2.25,0.35) {0.06};
  \node[tablevalue] at (3.75,0.35) {0.01};

  \node[font=\tiny,text=black!55] at (2.25,-1.12)
    {AI-assisted figure: OpenAI GPT-5, 2026-09-08.};
\end{tikzpicture}
\end{adjustbox}
    \caption{Joint Distribution of $\theta_t$ and $I_t$}
    \label{fig:joint_dist_simple}
\end{figure}

The Beckerian Return Schedule is vital to understand if you are a parent, parenting a specific child.\footnote{For example, the parents I know who read Emily Oster's books on parenting in order do so to get more information on how to parent their own particular child \citep{oster2013, oster2019, oster2021}. They are trying to get information on the Beckerian Return Schedule, in particular, the cells that pertain to their child.} It is also important if you are an econometrician trying to model parental choices, as Gary Becker did in much of his seminal work \citep{becker_tomes1976, becker1991}.

However, the Beckerian Return Schedule alone is not very useful for policy makers. Policy makers frequently only see a few of the relevant variables that go into producing student human capital, and necessarily need to target groups of students, not individuals. For example, they may see test scores, which are proxies for $\theta_t$, but have very limited information on a child's home environment, $I_t$. 

For the sake of this simple model, suppose that policy makers can marginally increase investment in children and want to know which students will see the largest test score gains/human capital increases from the additional investment. However, the policy maker only sees test scores ($\theta_t$) and has no other information on the students. A real-world policy analog to this hypothetical is ``do I give my next dollar to the gifted and talented program or to the remedial program?'' Then the policy maker wants to know the \textbf{Average Marginal Return} (AMR) of investment for each value of $\theta_t$. The average marginal return at $\theta_t=\bar{\theta}$ is the integral of the elasticity of skill with respect to investment over the conditional distribution of other inputs given $\theta_t=\bar{\theta}$. In this problem the AMR is $\sum_{i} P(I_t=i|\theta_t=\bar{\theta}) \left.\varepsilon_I\right|_{\theta_t=\bar{\theta}, I_t=i}.$ Calculating it requires both the Beckerian Return Schedule in Figure \ref{fig:BRS_simple} and the joint distribution in Figure \ref{fig:joint_dist_simple}. 

Frequently, structural economists reporting marginal returns will calculate a slice of the Beckerian Return Schedule (a Beckerian Return Slice), holding all inputs other than $\theta_t$ at fixed level, like their means or medians, and plot the resulting slice of the Beckerian Return Schedule \citep{attanasio_cattan_meghir_rubio_codina2020, freyberger_2025, attanasio_raquel_giannola_nores2026}. This has caused some confusion: many structural studies find that skill and investment are complements \citep{cunha_heckman_schennach2010, agostinelli_wiswall_2025}, so plotted Beckerian Return Slices frequently show higher returns for higher skilled children even though reduced form studies overwhelmingly find higher returns for lower skilled children \citep{bitler_hoynes_domina2014, morris_et_al_2018}. Intertemporal dynamics like dynamic complementarity \citep{List2026, heckman_mosso2014} certainly play a role in bridging the gap between structural and reduced form studies here, and I will discuss them extensively in this piece. However, reduced form studies typically investigate heterogeneity in treatment effects via conditional average treatment effects, while structural studies usually plot slices of the Beckerian Return Schedule. These are not the same objects. Structural economists should report Average Marginal Returns instead of Beckerian Return Slices if they want to back out reduced-formed objects from their models.

A Beckerian Return Slice along $\theta_t$ can be quite different than the Average Marginal Returns along $\theta_t$, as Figure \ref{fig:bsr_v_amr_simple} shows.
\begin{figure}
    \centering
\begin{adjustbox}{max width=\linewidth,center}
\begin{tikzpicture}[
  >=stealth,
  gridline/.style={draw=black,line width=0.45pt},
  axisline/.style={->,draw=black,line width=0.6pt},
  tablevalue/.style={font=\small}
]
  \fill[yellow!35] (1.5,0) rectangle (3.0,2.1);
  \fill[red!25] (5.15,0) rectangle (6.65,2.1);

  \draw[gridline] (0,0) rectangle (4.5,2.1);
  \draw[gridline] (1.5,0) -- (1.5,2.1);
  \draw[gridline] (3.0,0) -- (3.0,2.1);
  \draw[gridline] (0,0.7) -- (4.5,0.7);
  \draw[gridline] (0,1.4) -- (4.5,1.4);

  \draw[axisline] (0,0) -- (4.9,0);
  \draw[axisline] (0,0) -- (0,2.55);

  \node at (-0.35,1.75) {$2$};
  \node at (-0.35,1.05) {$1$};
  \node at (-0.35,0.35) {$0.5$};
  \node at (0.75,-0.32) {$0.2$};
  \node at (2.25,-0.32) {$1$};
  \node at (3.75,-0.32) {$4.5$};
  \node at (0,2.78) {$\theta_t$};
  \node at (2.25,-0.78) {$I_t$};

  \node[tablevalue] at (0.75,1.75) {0.909};
  \node[tablevalue] at (2.25,1.75) {0.667};
  \node[tablevalue] at (3.75,1.75) {0.308};
  \node[tablevalue] at (0.75,1.05) {0.833};
  \node[tablevalue] at (2.25,1.05) {0.500};
  \node[tablevalue] at (3.75,1.05) {0.182};
  \node[tablevalue] at (0.75,0.35) {0.714};
  \node[tablevalue] at (2.25,0.35) {0.333};
  \node[tablevalue] at (3.75,0.35) {0.100};

  \draw[gridline] (5.15,0.7) -- (6.65,0.7);
  \draw[gridline] (5.15,1.4) -- (6.65,1.4);

  \node[tablevalue] at (5.90,1.75) {0.4179};
  \node[tablevalue] at (5.90,1.05) {0.5018};
  \node[tablevalue] at (5.90,0.35) {0.5983};
  \node[font=\bfseries] at (5.90,2.48) {AMR};

  \node[align=center,font=\small] (slice-label) at (2.25,3.35)
    {Beckerian Return Slice at Median $I_t$};
  \draw[axisline] (slice-label.south) -- (2.25,2.13);

  \node[font=\tiny,text=black!55] at (3.325,-1.12)
    {AI-assisted figure: OpenAI GPT-5, 2026-09-08.};
\end{tikzpicture}
\end{adjustbox}
    \caption{Beckerian Return Slice vs Average Marginal Returns}
    \label{fig:bsr_v_amr_simple}
\end{figure}
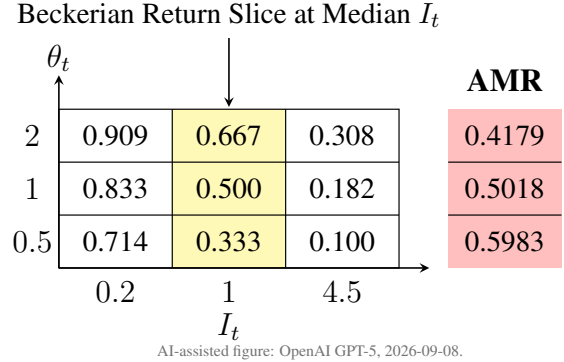
The Beckerian Return Slice at Median $I_t$, $1$, reflects the properties of the human capital production function from Equation (\ref{eq:basic_CES}). Skill stock and investment are complements, and so holding investment constant, more skilled students have higher marginal returns to investment (ranging from $\frac{1}{3}$ for the least skilled children to $\frac{2}{3}$ for the most skilled). However, the AMR by $\theta_t$ tells a very different story. The average return on investment is higher for low skilled children than it is for high skilled children, ranging from $0.60$ for the lowest skilled to $0.42$ for the most skilled. 

A structural econometrician who hands policy makers the Beckerian Return Slice and says ``marginal returns to investment are highest for high skilled kids'' will have deeply misled them. Even though all else equal low skilled children benefit less from investment than high skilled children do, \textit{not all else is equal}. The AMR takes the joint inequality in inputs into account. A Beckerian Return Slice does not. In this toy example, low skilled children receive so much less investment than high skilled children that on average, low skilled children actually benefit more from investment even though skill and investment are strong complements. Note that this is not a general result: in theory, the AMR can have the opposite slope as a BRS slice, the same slope, a steeper slope, or a flat slope.

\subsection{Definitions}\label{sec:framework:definitions}

The toy model above has one input, one output, two periods, and no shocks. The objects it illustrates are general, and this subsection formally defines them in general terms. Fix an output $\theta_{i,t+h}$ from a vector of outputs $\theta_{t+h}$ produced over a horizon of $h\geq 1$ periods. Let $Y=(Y_t,\dots,Y_{t+h-1})\in\mathcal{Y}$ collect every input over the horizon (in the toy model $Y=(Y_t)=(\ln\theta_t, \ln I_t)$). Let $\eta=(\eta_t,\dots,\eta_{t+h-1})\in\mathcal{E}$ collect a vector of shocks in each period. Intermediate outputs $\theta_{t+1},\dots,\theta_{t+h-1}$ are produced along the way and substituted out, so that
\begin{equation}\label{eq:def_g}
    \ln\theta_{i,t+h}=g_i(\ln Y,\eta),
\end{equation}
where $g_i$ is the $i$th output component of $g$, the composition of the period-by-period technologies.\footnote{Whether later inputs are held fixed or respond to earlier ones is a choice embedded in $g$. For example, if $g$ holds later investment fixed, the derivatives below are properties of the technology alone; if $g$ accounts for investment responses, they are policy elasticities that include behavioral responses. Section \ref{sec:early_childhood_cog:constr_invest_elast} states the choice used in the empirical work.} For a specific input $y_j$ (investment in period $t$, say), let
\begin{equation}\label{eq:def_elasticity}
    \varepsilon_j\equiv\frac{\partial\ln\theta_{t+h}}{\partial\ln y_j}=\frac{\partial g(\ln Y,\eta)}{\partial\ln y_j}
\end{equation}
denote the elasticity of $\theta_{t+h}$ with respect to $Y_j$. Because $Y$ and $\eta$ are random, $\varepsilon_j$ is a random variable; the definitions below differ only in what is done with that randomness. Throughout, $F$ denotes the joint distribution of $(Y,\eta)$ and a breve marks a value at which a variable is held fixed.

\begin{definition}[Beckerian Return at a Point]\label{def:br_point}
The Beckerian Return to input $Y_j$ at the point $(\breve{y},\breve{\eta})\in\mathcal{Y}\times\mathcal{E}$ is the elasticity evaluated there,
\[
    \left.\varepsilon_j\right|_{Y=\breve{y},\,\eta=\breve{\eta}}
    =\left.\frac{\partial g(\ln Y,\eta)}{\partial\ln Y_j}\right|_{Y=\breve{y},\,\eta=\breve{\eta}}.
\]
\end{definition}

\begin{definition}[Beckerian Return Schedule]\label{def:brs}
The Beckerian Return Schedule to input $Y_j$ is the mapping from points to Beckerian Returns,
\[
    \mathrm{BRS}_j:\mathcal{Y}\times\mathcal{E}\to\bbR,
    \qquad
    (\breve{y},\breve{\eta})\mapsto\left.\varepsilon_j\right|_{Y=\breve{y},\,\eta=\breve{\eta}}.
\]
\end{definition}

The schedule is a property of $g$ alone. It is defined without reference to $F$, and two populations with the same technology have the same schedule however differently their inputs are distributed. Because the schedule has as many arguments as there are inputs and shocks, it is rarely reported in full. What is frequently reported is a slice of the BRS.

\begin{definition}[Beckerian Return Slice]\label{def:br_slice}
Fix a reference value $\bar{y}_{-k}$ for every input other than $y_k$ and fix a reference value $\bar{\eta}$ for the shocks (in practice medians, means, or zero). The Beckerian Return Slice to input $y_j$ along input $y_k$ is the one-dimensional restriction of the schedule
\[
    \mathrm{BRS}_{j\mid k}(\breve{y}_k)
    =\left.\varepsilon_j\right|_{Y_k=\breve{y}_k,\,Y_{-k}=\bar{y}_{-k},\,\eta=\bar{\eta}},
\]
viewed as a function of $\breve{y}_k$.
\end{definition}

In Figure \ref{fig:bsr_v_amr_simple}, $j$ is investment in period $t$, $k$ is skill in period $t$, and $\bar{y}_{-k}$ is the median of investment. A slice fixes every other input at a single value and so, like the full schedule, ignores $F$.

\begin{definition}[Total Average Marginal Return]\label{def:amr_total}
The total Average Marginal Return to input $y_j$ is the elasticity averaged over the joint distribution of all inputs and shocks,
\[
    \mathrm{AMR}_j=\E[\varepsilon_j]
    =\int_{\mathcal{Y}\times\mathcal{E}}\frac{\partial g(\ln y,\eta)}{\partial\ln y_j}\,dF(y,\eta).
\]
\end{definition}

\begin{definition}[Conditional Average Marginal Return]\label{def:amr_cond}
The Average Marginal Return to input $y_j$ conditional on input $y_k$ taking the value $\breve{y}_k$ is the elasticity averaged over the conditional distribution of all other inputs and shocks given $y_k=\breve{y}_k$,
\[
    \mathrm{AMR}_{j\mid k}(\breve{y}_k)=\E\!\left[\varepsilon_j\,\middle|\,y_k=\breve{y}_k\right]
    =\int\frac{\partial g(\ln y,\eta)}{\partial\ln y_j}\bigg|_{y_k=\breve{y}_k}\,dF\!\left(y_{-k},\eta\,\middle|\,y_k=\breve{y}_k\right),
\]
viewed as a function of $\breve{y}_k$. The same definition applies with $Y_k=\breve{y}_k$ replaced by $Y_k\in B$ for a set $B$, which is how the empirical sections report AMRs within percentile bins of a conditioning input.
\end{definition}

\subsection{The AMR as a Treatment Effect}\label{sec:framework:treatment_effects}

The AMR can be interpreted as the marginal limit of an average treatment effect. Returning to the toy example, I depart from the notation used in the rest of the paper and instead adopt the standard notation of the treatment effects literature: let $Y=\ln\theta_{t+1}$, $D=\ln I_t$, and $Y(d)$ denote the potential outcome under $D=d$, holding the child's other inputs at their realized values, so that $\partial Y(d)/\partial d=\varepsilon_I(d)$. Consider a policy that raises every child's log investment by $\delta$. Treating the policy as a binary treatment, its conditional average treatment effect among children with $\theta_t=\bar{\theta}$ is $\mathrm{CATE}(\delta\mid\bar{\theta})=\E[Y(D+\delta)-Y(D)\mid\theta_t=\bar{\theta}]$. Taking the limit as $\delta$ approaches zero gives
\begin{equation}\label{eq:amr_cate}
    \mathrm{AMR}(\bar{\theta})
    =\E\!\left[\varepsilon_I\mid\theta_t=\bar{\theta}\right]
    =\lim_{\delta\to 0}\frac{\mathrm{CATE}(\delta\mid\bar{\theta})}{\delta}.
\end{equation}
The AMR is thus the conditional average treatment effect of a marginal, proportional increase in investment, per unit of that increase. Here, ``marginal'' refers to a small increase in investment for every child, not to the marginal individual of the marginal treatment effect \citep{heckman_vytlacil2005, heckman_urzua_vytlacil2006, carneiro_heckman_vytlacil2010}.

\section{Theoretical Applications of the AMR}\label{sec:applications}

To illustrate the usefulness of the AMR, I show how it can be used to help explain three puzzles in empirical work: why changing the scale of inputs can change the policy implications of a model \citep{freyberger_2025}, why adding higher order terms to a linear regression can change interaction coefficients, and why dynamic complementarity can be difficult to detect in reduced form work \citep{attanasio_2026}.

\subsection{Importance of Scale}\label{sec:applications:scale}

Scaling issues have plagued studies of child skill formation. For example, \cite{freyberger_2025} finds that the age five (Beckerian) marginal products and estimated impacts of income transfers reported in \cite{attanasio_meghir_nix_2020} flip sign after converting investment from a standardized scale to 100 INR. While dramatic, Freyberger's finding is not unique: the production function estimated in \cite{attanasio_cattan_meghir_rubio_codina2020} changed from perfect substitutes to Cobb-Douglas after they rescaled their factors between early analysis and publication. The distinction between the AMR and BRS helps explain why seemingly innocuous changes in input normalizations can have dramatic impacts on the economic implications of structural models.\footnote{This section helps build intuition for the technical conditions laid out in \cite{freyberger_2025}.}

Assume that skill and investment in period $t$ are jointly log-normally distributed and that the child skill production function is CES:
\begin{equation}
    \theta_{t+1}=\left(\gamma_\theta \theta_t^\rho + \gamma_I I^\rho \right)^{1/\rho}\eta_t,
\end{equation}
where $\rho$ is the substitution parameter, $\eta$ is an idiosyncratic error term that is iid and with $E[\ln \eta_t] = 0$, and the $\gamma$'s are positive constants. This is the setting of \cite{cunha_heckman_schennach2010}, the seminal paper that introduced dynamic factor analysis. Then, the derivative of the AMR with respect to $\ln\theta_t$ is
\begin{equation}
    \frac{d}{d\ln\theta_t}\E[\varepsilon_I|\theta_t]= -\rho (1-\beta_t)\E[\varepsilon_I\varepsilon_\theta|\theta_t],
\end{equation}
where $\beta_t=\frac{\operatorname{Cov}(\ln\theta_t, \ln I_t)}{\operatorname{Var}(\ln\theta_t)}$, and the derivative of the BRS with respect to $\ln\theta_t$ is
\begin{equation}
    \frac{d^2 \ln\theta_{t+1}}{d\ln\theta_t d\ln I_t} = -\rho \varepsilon_I\varepsilon_\theta.
\end{equation}
Note that the slope of the AMR and the BRS have different signs if and only if $1-\beta\leq 0$, that is, if
\begin{equation}\label{eq:key_condition}
    \operatorname{Cov}(\ln\theta_t, \ln I_t) \geq \operatorname{Var}(\ln\theta_t).
\end{equation}

Condition \ref{eq:key_condition} is not scale invariant. It can never be met if the variance of both log factors is normalized to 1. It can also be forced to hold by rescaling either log skill or log investment. If the factors are improperly scaled, the model cannot properly fit both the AMR and BRS. It will be torn between the two, potentially distorting $\rho$ and the $\gamma$. \cite{freyberger_2025} shows how dramatic an impact this improper scaling can have on real data.

Condition \ref{eq:key_condition} along with the distinction between the AMR and the BRS helps explain why seemingly innocuous changes to input scales in nonlinear models can cause dramatic changes in the estimates and economic implications of those models. I discuss in Section \ref{sec:factor_scales} what the correct scale to use in my empirical application is, drawing heavily on \cite{freyberger_2025}.

\subsection{Bias in Linear Model Interaction Terms}\label{sec:applications:translog}

In this section, I show that distinction between the AMR and BRS helps clarify why adding higher order terms can change estimates of complementarity in a linear models with interaction effects \citep{deming2017, agostinelli_wiswall_2025, agostinelli_saharkhir_wiswall2025}. 

Applied researchers frequently use a linear model with interaction effects to estimate complementarity in an underlying production function they are approximating. In this section, I use the notation and true production function from Equation (\ref{eq:basic_CES}). It is common to approximate a production function $g$ using an expression of the form
\begin{equation}\label{eq:translog_like}
    \ln\theta_{t+1} = \beta_0^L + \beta_\theta^L \ln\theta_t + \beta_I^L \ln I_t + \beta_{\theta, I}^L\ln\theta_t\ln I_t.
\end{equation}
This log-linear interaction effects specification (LLI) is similar to a true translog specification (given by equation \ref{eq:translog_true}), except that it drops the quadratic terms that are vital for capturing changes returns to each input and make translog a second-order approximation of any well-behaved production function.
\begin{equation}\label{eq:translog_true}
    \ln\theta_{t+1} = \beta_0^T + \beta_\theta^T \ln\theta_t + \beta_I^T \ln I_t + \beta_{\theta, I}^T\ln\theta_t\ln I_t + \beta_{II}^T\ln I_t^2 + \beta_{\theta \theta}^T \ln\theta_t^2
\end{equation}

Under the LLI approximation given by equation (\ref{eq:translog_like}), the Beckerian Return at a point and the AMR calculated at the same value of $\theta_t$ are both given by $$\varepsilon_I|_{\theta_t=\bar\theta, I_t=i}=\beta_I+\beta_{\theta,I}\ln\bar{\theta}=\E[\varepsilon_I| \theta_t=\bar{\theta}].$$ The LLI approximation conflates the AMR and the BRS, which leads to a predictably signable bias in $\beta_{\theta I}^L$ relative to the true translog $\beta_{\theta I}^T$:
\begin{equation}\label{eq:LLI_bias}
    \beta_{\theta I}^{\text{L}} = \beta_{\theta I}^{\text{T}} + \frac{\E\left[\widetilde{(\ln\theta_t\ln I_t)}\,\widetilde{(\ln\theta_t)^2}\right]}{\E\left[\widetilde{(\ln\theta_t \ln I_t)}^{\,2}\right]} \beta^{\text{T}}_{\theta\theta} + \frac{\E\left[\widetilde{(\ln\theta_t\ln I_t)}\,\widetilde{(\ln I_t)^2}\right]}{\E\left[\widetilde{(\ln\theta_t \ln I_t)}^{\,2}\right]} \beta^{\text{T}}_{II},
\end{equation}
where $\widetilde{x}$ denotes the residuals of projecting $\ln\theta_t \ln I_t$, $(\ln\theta_t)^2$, and $(\ln I_t)^2$ onto $(1, \ln \theta_t, \ln I_t)$. A formal proof and statement of Equation (\ref{eq:LLI_bias}) can be found in the Appendix in Proposition \ref{prop:translog_bias}.

Equation (\ref{eq:LLI_bias}) illustrates how misspecification can cause the BRS and AMR to bias each other. The joint distribution and changes in marginal returns interact with each other and together contaminate the interaction coefficient.

In practice, reduced-form researchers can often increase both the size and statistical significance of their coefficients while more accurately capturing complementarities just by adding quadratic terms to their regressions. For example, \cite{deming2017} investigates the interactions between social and cognitive skills in determining wages on the NLSY79. Adding quadratic terms to Column 6 of his Table I causes his interaction effect between Cognitive and Social skills to increase by 40\% from 0.010 to 0.014 and to go from a p-value of .08 to a p-value of .03. 

In contrast, under the true translog function given by equation (\ref{eq:translog_true}), the Beckerian schedule of returns for an individual with skill stock $\theta_t$ and investment $I_t$ are given by $$\varepsilon_I|_{\theta_t, I_t} = \beta_I+\beta_{\theta,I}\ln\theta + 2\beta_{II}\ln I_t,$$ which is generally not equal to the expected marginal returns to investment for an individual with skill stock $\theta_t$, $$\E[\varepsilon_I|\theta_t]=\beta_I+\beta_{\theta,I}\ln\theta + 2\beta_{II}\left[\E[\ln I_t] + \frac{\Cov(\ln I_t, \ln \theta_t)}{\Var(\ln\theta_t)}(\ln\theta_t-\E[\ln\theta_t]) \right].$$ As a result, the derivative of the expected marginal returns to investment with respect to skill stock is given by $$\frac{d}{d\ln\theta_t}\E\left[\varepsilon_I|\theta_t\right] = \beta_{\theta, I} + 2\beta_{II}\frac{\Cov(\ln I_t,\ln \theta_t)}{\Var \ln \theta_t}.$$ This produces a situation very similar to the CES case, where a complementary translog function ($\beta_{\theta I}>0$) can still have $\E[\varepsilon_I|\theta_t]$ be decreasing in $\theta_t$ if
\begin{equation}\label{eq:translog_cond}
    -2\beta_{II}\frac{\Cov(\ln I_t, \ln \theta_t)}{\Var(\ln\theta_t)}>\beta_{\theta, I}.
\end{equation}
Condition \ref{eq:translog_cond} is the counterpart of Condition \ref{eq:key_condition} for the true translog case. In a translog expansion of a CES production function, $\beta_{\theta,I}=-2\beta_{II}$, so the two conditions are equivalent if the true data generating process is CES \citep{Kmenta1967}. As discussed earlier, no such condition exists in LLI specifications, and their estimated interaction effects conflate curvature with interaction as a result. In short, LLI specifications conflate AMR and the BRS, and so fail to approximate either when they are substantially different. Reduced-form researchers can often both reduce bias and increase their coefficients' size and significance by switching to full translog approximations.

\subsection{A Static Numerical Example}\label{sec:applications:num_ex}

In this section I walk through a numerical example to help illustrate the theoretical arguments I make in Sections \ref{sec:applications:scale} and \ref{sec:applications:translog}. Suppose that $\theta_{t+1}$ is linked deterministically to $\theta_t$ and $I_t$ through equation (\ref{eq:basic_CES}), so that the true data generating process is CES with strong complementarity. Further suppose that 
\begin{equation}\label{eq:num_ex_dist}
    \begin{pmatrix}
            \ln \theta_t \\ \ln I_t
        \end{pmatrix} \sim \mathcal{N}\left(0, \begin{bmatrix}
            1 & 2 \\ 2 & 5
        \end{bmatrix}\right).
\end{equation}
Then, if $\theta_t = 1$, then $\E[I_t|\theta_t=1] = \sqrt{e}$ and $\E[\varepsilon_I|\theta_t=2] = \E\!\left[\frac{I_t^{-1}}{I_t^{-1}+1^{-1}} \;\middle|\; \theta_t=1\right]=\frac{1}{2}$. If $\theta = 2$, then $\E[I_t|\theta_t=1]=4\sqrt{e}$ and $\E[\varepsilon_I|\theta_t=2] = \E\!\left[\frac{I_t^{-1}}{I_t^{-1}+2^{-1}} \;\middle|\; \theta_t=2\right]\approx 0.36$, which is less than the elasticity when $\theta = 1$. A doubling in skill is associated with a quadrupling in investment, so each unit of investment has less skill to complement it as skill increases. The complementarity causes the expected marginal returns to investment to fall as skill increases. 

\begin{figure}
    \centering
    \includegraphics[width=0.5\linewidth]{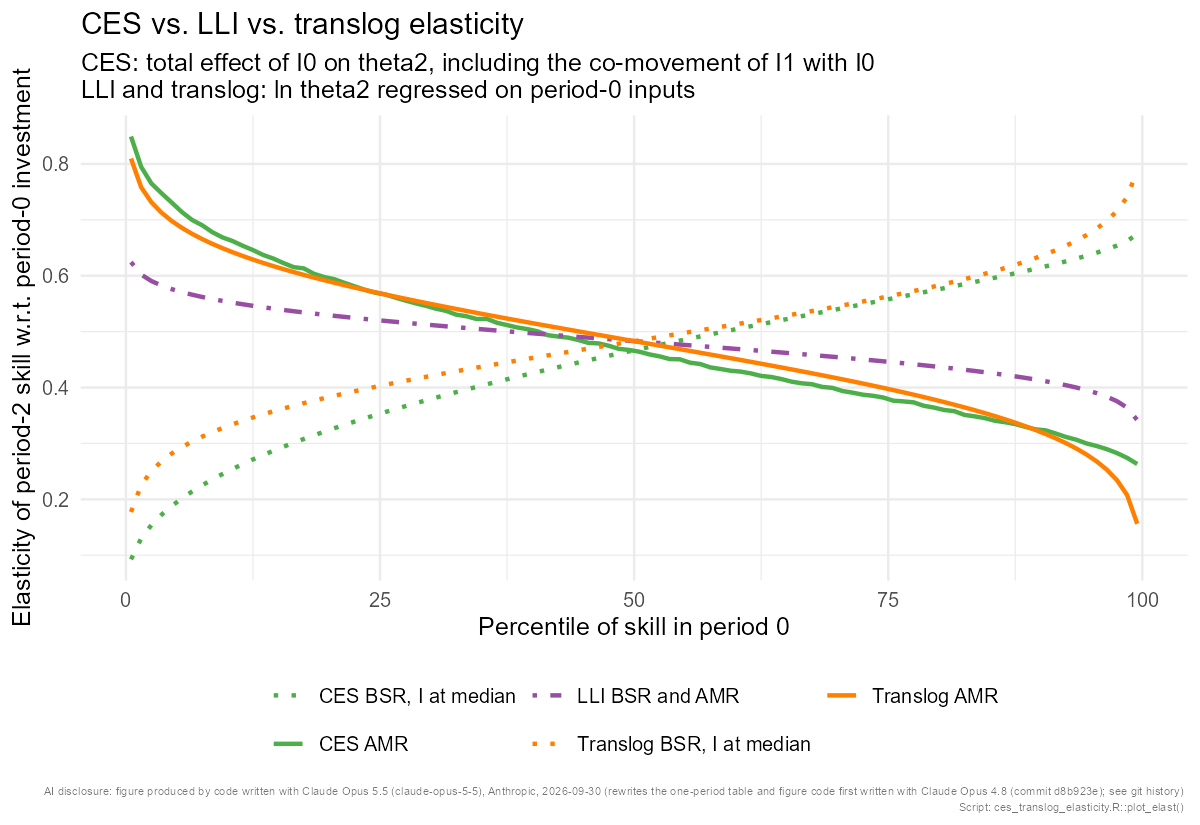}
    \caption{Plotting $\varepsilon_I$ under different modeling assumptions.}
    \label{fig:translog_comparison}
\end{figure}

It is also instructive to see how well the LLI and translog models discussed in \ref{sec:applications:translog} perform under this data generating process. To test this, I generate one million draws from the distribution in Equation (\ref{eq:num_ex_dist}) and then plot both the AMR, $\E[\varepsilon_I|\theta_t=\bar{\theta}]$, and the BRS given by $\varepsilon_I|_{I_t=\operatorname{Median}{I}=1,\theta_t=\bar{\theta}}$ for the correctly specified CES, the LLI approximation, and the translog approximation in Figure \ref{fig:translog_comparison}. Note that while the translog BRS and AMR match the true CES BRS and AMR quite closely, the LLI approximates neither. This is despite the LLI approximation predicting $\ln\theta_{t+1}$ quite closely: Table \ref{tab:translog_coefs} shows that the LLI approximation has an $R^2$ of 0.987. Just because a LLI can closely approximate the output of a production function does not mean that it recovers either the AMR or the BRS well.  

\begin{table}[htbp]
\centering
\caption{Translog Approximation of a CES Production Function}
\label{tab:translog_coefs}
\begin{tabular}{lcc}
\toprule
 & LLI & Translog \\
\midrule
\quad Intercept ($\beta_0$) & -0.781 & -0.721 \\
\quad $\ln\theta_t$ ($\beta_\theta$) & 0.499 & 0.500 \\
\quad $\ln I_t$ ($\beta_I$) & 0.500 & 0.500 \\
\quad $\ln\theta_t \times \ln I_t$ ($\beta_{\theta I}$) & -0.060 & 0.182 \\
\quad $(\ln\theta_t)^2$ ($\beta_{\theta\theta}$) & --- & -0.091 \\
\quad $(\ln I_t)^2$ ($\beta_{II}$) & --- & -0.091 \\
\midrule
$R^2$ & 0.987 & 1.000 \\
$N$ & \multicolumn{2}{c}{1{,}000{,}000} \\
\bottomrule
\end{tabular}
\begin{minipage}{0.95\textwidth}
\vspace{0.5em}\footnotesize
\textit{Notes:} Coefficients from OLS translog regressions on 1{,}000{,}000 simulated draws from a known two-input CES technology ($\rho=-1$), where $\theta_{t+1} = (\theta_t^{-1} + I_t^{-1})^{-1}$. The Log-Linear Interaction (LLI) omits the quadratic terms $(\ln\theta_t)^2$ and $(\ln I_t)^2$; the translog includes all five second-order terms. Because the data-generating process is deterministic, standard errors are negligible ($<0.001$) and are omitted: the residual reflects functional-form approximation error, not sampling noise. --- indicates a term absent from that specification.
\end{minipage}
\end{table}

Notably, $\beta_{\theta, I}$ is positive in the true translog case, accurately reflecting the complementary nature of the CES data generating process, while $\beta_{\theta, I}$ is negative in the LLI specification, just as the bias from Equation (\ref{eq:LLI_bias}) predicts. An economist who interpreted the negative interaction term in the LLI case to mean that skill and investment are substitutes would be deeply mistaken.

\subsection{Why Dynamic Complementarity is Hard to Detect}\label{sec:applications:dc}

I now turn to how the AMR plays out over multiple periods, focusing on dynamic complementarity, a leading theoretical framework in child development \citep{heckman_mosso2014, List2026}. Dynamic complementarity is the idea that early investment raises the return to later investment: if skill and investment are complements, an early intervention raises the productivity of subsequent investments by increasing skill. Prior studies have tested for dynamic complementarity by comparing the effects of later interventions on children who did and did not receive earlier ones \citep{johnson_jackson2019, rossin-slater_wust2020, carneiro_et_al2025, meghir_et_al2023}. In this section, I use the AMR to show that such tests can fail to detect even strong dynamic complementarities. The reason is that dynamic complementarity affects the returns to the entire schedule of later investments, while these tests capture only the marginal return to an additional intervention. Because that marginal return depends on the joint distribution of skill and investment, an early intervention can substantially raise the returns to inframarginal investments while leaving the marginal return unchanged or even reducing it.

This is an important contribution because while dynamic complementarity is a compelling theory, empirical evidence on it is mixed. \cite{johnson_jackson2019} show that the benefits of the Head Start preschool program were large when followed by access to better funded schools and \cite{heckman_tian_zhang_zhou2026DC} shows that children who have been in an early childhood program for longer conditional on age have higher first time passing rates for new tests, both of which are consistent with dynamic complementarity. However, \cite{rossin-slater_wust2020} finds evidence of strong dynamic substitutability between and infant health and a preschool program and \cite{carneiro_et_al2025} finds that the impact of classroom quality at different ages appears to be additive. Some studies are inconclusive. For example, \cite{meghir_et_al2023} finds positive but statistically insignificant interactions between an early stimulation and an enhanced preschool program.\footnote{See \cite{attanasio_2026} for a more in-depth discussion of evidence of dynamic complementarity in applied research.}

The AMR can help explain why evidence for dynamic complementarity is mixed. To see this, we can extend the model in Section \ref{sec:applications:num_ex} to three periods, with each period being governed by the CES transition from Equation (\ref{eq:basic_CES}) and making the joint distribution of skill and investment
\begin{equation}\label{eq:cs_ex_dist}
    \begin{pmatrix}
            \ln \theta_t \\ \ln I_t \\ \ln I_{t+1}
        \end{pmatrix} \sim \mathcal{N}\left(\begin{bmatrix}
            0 \\ 0 \\ \mu_3
        \end{bmatrix}, \begin{bmatrix}
            1 & 2 & .8 \\ 2 & 5 & 2 \\ .8 & 2 & 5
        \end{bmatrix}\right),.
\end{equation}
Where $\mu_3\in \{0,2\}$. This investment environment is based off a stylized version of the US school system. In early childhood ($t$), children are entirely dependent on their parents for investment, so skill and investments are highly correlated. In later childhood ($t+1$) children are in school. School has two effects. First, it raises investment for all children (captured when $\mu_3=2$). Secondly, it decreases the correlation between investment and initial child skill. Even if school quality is correlated with parental involvement, that correlation is less than the direct effect of parental involvement (this is captured by the covariance matrix, which makes $\operatorname{Cov}(\ln \theta_t, \ln I_{t+1}|\ln I_t)=0$).

\begin{figure}
    \centering
    \includegraphics[width=\linewidth]{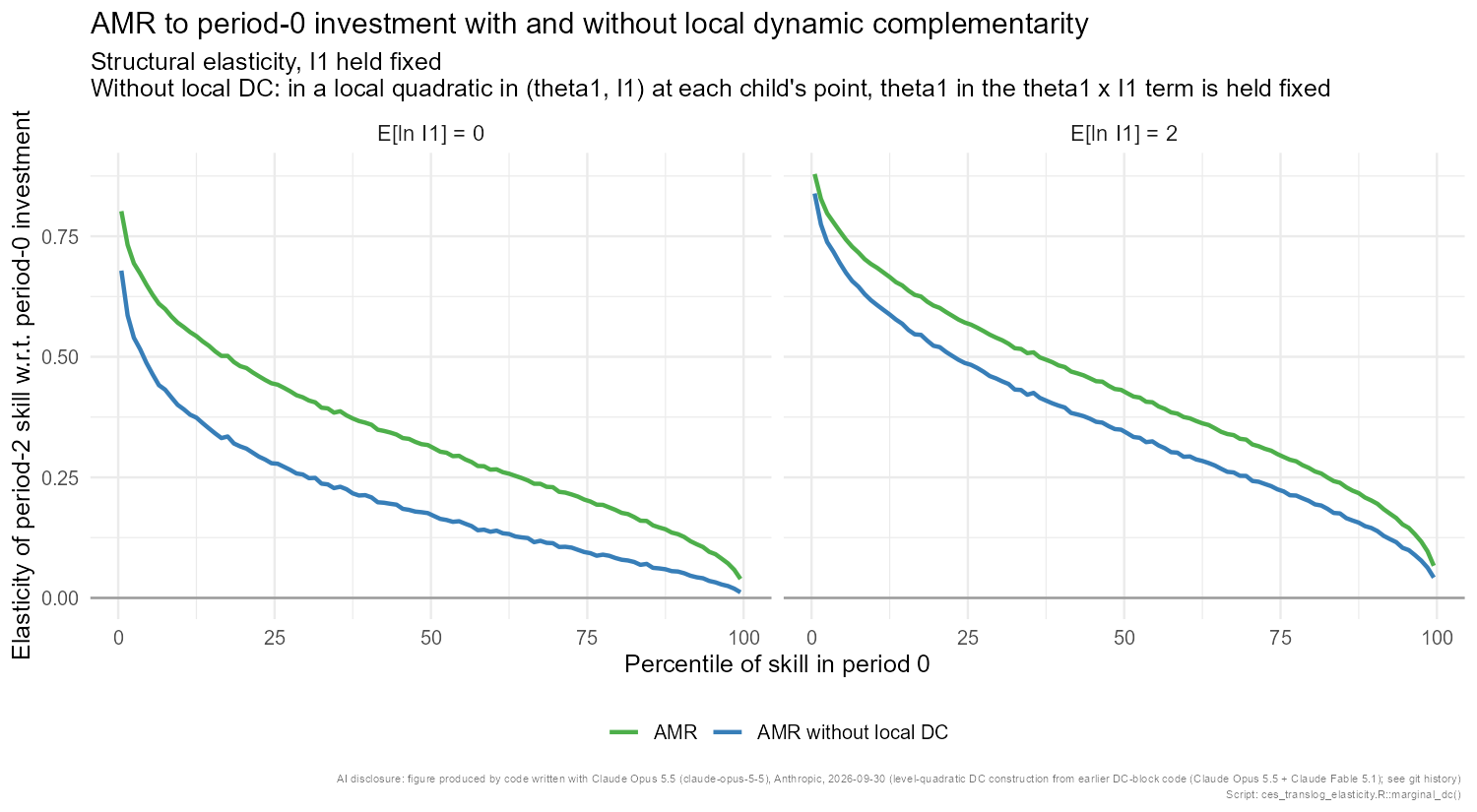}
    \caption{Local Effect of Dynamic Complementarity in an Investment Poor and Investment Rich Environment}
    \label{fig:dc_marginal}
\end{figure}

Throughout this section the AMR is the two-period object: the elasticity of $\theta_{t+2}$ with respect to $I_t$, averaged over the conditional distribution of the other inputs given $\theta_t$, with $I_{t+1}$ held fixed so that parents do not respond. Part of that elasticity runs through dynamic complementarity: early investment raises $\theta_{t+1}$, which raises the marginal return to $I_{t+1}$. To isolate this channel, I approximate the period-$(t+1)$ technology by a quadratic in at each child's own point and recompute the elasticity holding the skill argument of the $\ln\theta_{t+1}\times\ln I_{t+1}$ interaction at its pre-(marginal) intervention value, so that the early investment's effect on $\theta_{t+1}$ no longer feeds through to the return on $I_{t+1}$. I call the difference between the two elasticities the local contribution of dynamic complementarity. Figure \ref{fig:dc_marginal} plots the AMR with and without this contribution, with left-hand panel corresponding to an investment-scarce later childhood ($\mu_3=0$) and the right-hand panel corresponding to an investment-rich later childhood ($\mu_3=2$).

A naive reader might expect dynamic complementarity to play a larger role when investment is more plentiful at school age: the increase in skill from the early childhood program has more investment to enhance. However, Figure \ref{fig:dc_marginal} shows the opposite: local dynamic complementarity matters less when later investment is more abundant. 

\begin{figure}
    \centering
%
\begingroup
\definecolor{dcBlue}{HTML}{2166AC}
\definecolor{dcRed}{HTML}{B52835}

\def\dcDose{0.7}
\def\dcBlueSkill{2.1}
\def\dcRedSkill{4.2}
\def\dcAxisMax{6}

\newcommand{\dcIntervention}[3]{%
  \draw[dc intervention,draw=#3]
    (#1,{#2+0.09}) -- (#1,{#2+\dcDose-0.10});
  \fill[fill=#3] (#1,#2) circle[radius=2.5pt];
  \draw[draw=#3,fill=white,line width=1.15pt]
    (#1,{#2+\dcDose}) circle[radius=2.5pt];
  \node[anchor=west,text=#3,font=\footnotesize,inner sep=1.5pt]
    at ({#1+0.16},{#2+0.5*\dcDose}) {$+\Delta I$};
}

\newcommand{\dcPanel}[7]{%
  \begin{scope}[shift={(#1,#2)}]
    \node[dc panel title] at (-0.25,7.02) {#3};
    \node[dc panel subtitle] at (-0.25,6.62) {#4};

    \foreach \q in {0.7,1.4,2.1,2.8,3.5,4.2,4.9,5.6}{
      \draw[dc isoquant] (\q,\dcAxisMax) -- (\q,\q)
        -- (\dcAxisMax,\q);
    }
    \draw[dc diagonal] (0,0) -- (\dcAxisMax,\dcAxisMax);
    \node[rotate=45,fill=white,inner sep=1.5pt,
          text=black!70,font=\footnotesize]
      at (5.45,5.75) {$I_t=\theta_t$};

    \draw[dc axis] (0,0) -- (6.25,0);
    \draw[dc axis] (0,0) -- (0,6.25);
    \node[rotate=90,font=\small] at (-0.55,3.12)
      {School-age investment, $I_t$};
    \foreach \x/\lab in {\dcBlueSkill/B,\dcRedSkill/R}{
      \draw[black!75,line width=0.45pt] (\x,0) -- (\x,-0.08);
      \node[anchor=north,inner sep=2pt,font=\footnotesize]
        at (\x,-0.10) {$\theta_t^{\lab}$};
    }
    \node[font=\small] at (3.1,-0.78)
      {Skill entering school age, $\theta_t$};

    \dcIntervention{\dcBlueSkill}{#5}{dcBlue}
    \dcIntervention{\dcRedSkill}{#6}{dcRed}

    \node[font=\small] at (3.0,-1.28) {#7};

    \pgfmathsetmacro{\dcBlueGain}{min(\dcDose,max(\dcBlueSkill-(#5),0))}
    \pgfmathsetmacro{\dcRedGain}{min(\dcDose,max(\dcRedSkill-(#6),0))}
    \typeout{DC-CHECK #3: blue=\dcBlueGain; red=\dcRedGain; dose=\dcDose}
  \end{scope}%
}

\begin{tikzpicture}[
  x=1cm,y=1cm,>=stealth,
  font=\small,
  dc axis/.style={->,draw=black!85,line width=0.7pt},
  dc isoquant/.style={draw=black!27,line width=0.55pt},
  dc diagonal/.style={draw=black!70,line width=0.8pt},
  dc intervention/.style={->,line width=1.25pt},
  dc panel title/.style={anchor=west,font=\bfseries\small,inner sep=0pt},
  dc panel subtitle/.style={anchor=west,font=\footnotesize,
                            text=black!70,inner sep=0pt}
]
  \node[font=\large\bfseries] at (7.0,17.17)
    {School-age returns after a preschool intervention};
  \node[font=\small,text=black!80] at (7.0,16.62)
    {Perfect complements: $\theta_{t+1}=\min\{\theta_t,I_t\}$};

  \dcPanel{0}{8.85}
    {A. Classic dynamic complementarity}
    {Preschool moves skill across the $45^{\circ}$ line}
    {2.8}{2.8}
    {School-age gain: \textcolor{dcBlue}{$G_B=0$}\quad
                      \textcolor{dcRed}{$G_R=\Delta I$}}

  \dcPanel{8}{8.85}
    {B. High investment environment}
    {Skill constrains both children}
    {4.9}{4.9}
    {School-age gain: \textcolor{dcBlue}{$G_B=0$}\quad
                      \textcolor{dcRed}{$G_R=0$}}

  \dcPanel{0}{0}
    {C. Low investment environment}
    {Investment constrains both children}
    {0.7}{0.7}
    {School-age gain: \textcolor{dcBlue}{$G_B=\Delta I$}\quad
                      \textcolor{dcRed}{$G_R=\Delta I$}}

  \dcPanel{8}{0}
    {D. Investment response}
    {Investment response partially restores skill constraint}
    {1.75}{3.85}
    {School-age gain: \textcolor{dcBlue}{$G_B=\frac{1}{2}\Delta I$}\quad
                      \textcolor{dcRed}{$G_R=\frac{1}{2}\Delta I$}}

  \begin{scope}[yshift=-0.15cm]
  \fill[dcBlue] (1.3,-1.97) circle[radius=2.5pt];
  \node[anchor=west] at (1.48,-1.97) {Blue: no preschool intervention};
  \fill[dcRed] (8.1,-1.97) circle[radius=2.5pt];
  \node[anchor=west] at (8.28,-1.97) {Red: received preschool intervention};

  \fill[black!75] (0.8,-2.55) circle[radius=2.5pt];
  \node[anchor=west,font=\footnotesize] at (0.98,-2.55) {Before school-age program};
  \draw[black!75,fill=white,line width=1.1pt]
    (6.0,-2.55) circle[radius=2.5pt];
  \node[anchor=west,font=\footnotesize] at (6.18,-2.55) {After school-age program};
  \draw[->,black!75,line width=1.1pt] (11.15,-2.75) -- (11.15,-2.35);
  \node[anchor=west,font=\footnotesize] at (11.33,-2.55) {Same dose $\Delta I$};

  \node[align=center,text width=14.7cm,font=\footnotesize,text=black!75]
    at (7.0,-3.27)
    {Gray lines are isoquants. In panel D, the preschool benefit remains, but additional exogenous school-age investment yields no further gain.};

  \node[font=\tiny,text=black!45] at (7.0,-3.93)
    {TikZ code assisted by OpenAI GPT-6 (model ID unavailable), OpenAI, 2026-09-25.};
  \end{scope}
\end{tikzpicture}
\endgroup
    \caption{Diagram illustrating the relationship between dynamic complementarity and the joint distribution of skill and investment.}
    \label{fig:tikz_dynamic_complementarity}
\end{figure}

The mechanism is visible in the CES itself. With $\rho=-1$ the response of the investment elasticity to skill is $\partial\varepsilon_I/\partial\ln\theta=\varepsilon_I\varepsilon_\theta$, which is largest when the two inputs are balanced and shrinks toward zero when either one dominates. When $\mu_3=2$, investment at $t+1$ swamps skill for most children, skill is the binding input, $\varepsilon_{I,t+1}$ is small, and raising $\theta_{t+1}$ does little to the return on $I_{t+1}$. Figure \ref{fig:tikz_dynamic_complementarity} makes the same point in the limiting case $\rho\to-\infty$, where the CES becomes perfect complements. It shows the effect of a school-age intervention on a hypothetical child. In each panel, the blue dot is the child without a preschool intervention and the red dot is the same child with one; the grey lines are isoquants of the school-age technology. The only thing that changes between panels is the joint distribution of skill and investment in the school age period. The impact of the preschool program on school age skills is the same, as is the increase in school age investment caused by the second hypothetical intervention.

Panel A shows the classic case of dynamic complementarity: the preschool intervention moves the child across the 45 degree line and so later investments have a larger effect. Panel B shows the high investment environment, which corresponds to the second panel of Figure \ref{fig:dc_marginal}. When investment is plentiful, the school age intervention has no impact on either child. Dynamic complementarity mattered a lot for infra-marginal investment, but not for marginal investment. This explains why the local effects of dynamic complementarity are lower in the right-hand panel of Figure \ref{fig:dc_marginal}. Panel C of Figure \ref{fig:tikz_dynamic_complementarity} shows that when investment is sufficiently scarce at school age, dynamic complementarity plays no role. Children who are sufficiently investment starved get the same benefit from the second intervention regardless of whether they received the first.

Panel D shows what happens when parents respond endogenously to investment. In the plotted case, parents increase school age investment in response to the preschool intervention so that the marginal impact of school age investment is the same regardless of whether the child participates in a preschool program or not. If the elasticity of parental investment is greater than the elasticity of school age skill (both with respect to the preschool intervention), perhaps because the preschool intervention lowers the cost of future investment in addition to making it more valuable, than the school age intervention could have a smaller impact on the child who received the preschool intervention. Note that in this case as in Panel B, dynamic complementarity has a large effect on infra-marginal investment, just not on marginal investment. Many studies \citep{attanasio_cattan_meghir_rubio_codina2020, bassi_et_al2026, jha_list_royer_samek2026} find evidence that parents increase later investment in response to earlier interventions, so Panel D has real-world significance. \cite{cotton_et_al2026}, \cite{lim2026Korea} and \cite{lim2026US} find that children with higher cognitive skills complete more practice problems and spend more time on academic self investment. Since parents often respond to early childhood interventions by increasing investment in the child and children may respond similarly, dynamic complementarity would be difficult to detect by looking at the impacts of repeated exogenous investment shocks, even if it plays a large role in enhancing infra-marginal investment.

\section{Data}\label{sec:data}

I use the CNLSY to track children's cognitive and socio-emotional (noncognitive) development from ages 5 to 14. I divide childhood into five two-year periods, with $t=0$ being ages 4–5, $t=1$ ages 6-7, and so on through $t=4$ (ages 13–14). I ignore the earlier ages used in \cite{cunha_heckman_schennach2010} because I lack consistent and reliable proxies for child cognitive skill in before the age of 5. 

Following \cite{cunha_heckman_schennach2010}, I have five factor inputs. $\theta_t$ from Section \ref{sec:framework} becomes a vector with two components: one for child cognitive skill, $\theta_{C,t}$, and one for child noncognitive skill, $\theta_{N,t}$. Investment $I_t$ becomes maternal/parental investment $M_{I,t}$ and maternal cognitive and noncognitive skill, $M_C$ and $M_N$ respectively, enter into the production function as time-invariant factors to capture heterogeneity. In this section, I describe my data and my proxy selection process.

My primary data source is the Children of the National Longitudinal Survey of Youth 1979 (CNLSY). The CNLSY follows the biological children of women in the National Longitudinal Survey of Youth 1979 (NLSY79), a nationally representative sample of individuals who were ages 14--22 in 1979. Beginning in 1986, children in the CNLSY were surveyed biennially through 2020. The survey contains detailed information on children’s cognitive test scores, parental reports of child behavior, and multiple measures of the home environment and parental inputs, making it well suited for studying early childhood skill formation. 

The cognitive and behavioral assessments administered in the CNLSY are consistent from ages 5-14. One advantage of this is that these tests are arguably measuring the same skills in each period, and helps this study avoid the illusion of fadeout that administering different tests at different ages can cause \citep{heckman_zhou2026}.

To measure cognitive skill, I use the Peabody Individual Achievement Test (PIAT), which includes subtests in mathematics, vocabulary, and reading comprehension. To measure child noncognitive skill, I use parental reports from the Behavioral Problems Index (BPI). I use subscales capturing headstrong behavior, antisocial tendencies, peer conflict, hyperactivity, and anxiety. Both of these have been widely used by prior studies as measures of cognitive and noncognitive skills \citep{cunha_heckman_schennach2010, agostinelli_wiswall_2025}.

To measure Parental investment, I use the Home Observation for Measurement of the Environment (HOME) inventory. The HOME inventory records both features of the child’s home environment and parent–child interactions, with available items varying by child age. For example, for children ages three to four, I use the number of books in the household, how often the mother reads to the child, whether the household has a CD player, how often the child goes on a tript to a museam or an outing, and the number of magazines in the home. For children ages five and six, I use the same measures, minus museum trips and outings, and additionally include how often the mother praises the child directly or to another person, whether there is a musical instrument in the home, and whether the child participates in special lessons or extracurricular activities. A full list of investment measures for all ages can be found in Table \ref{tab:signal_table_it}.

In addition to the CNLSY, I use data from the NLSY79 to measure maternal characteristics. These include mother’s education, marital status, age, number of children, and total net household income. Maternal cognitive skill is proxied using subscores from the Armed Services Vocational Aptitude Battery (ASVAB), administered for the vast majority of respondents in 1980, while maternal noncognitive skill is proxied using subscores of the Rosenberg Self-Esteem Scale, administered in 1980. Consistent with prior work \citep{cunha_heckman_schennach2010, agostinelli_wiswall_2025}, I treat these maternal skill measures as time-invariant. 

Many of the factor measures I use are binary or Likert, motivating the discrete measurement system described in Section~\ref{sec:measurement}. I model all investment and noncognitive skill proxies as discrete, and all cognitive skill proxies as continuous, since early testing indicated that there was little gain to treating the PIAT or ASVAB tests as discrete. 

To select proxies, I ran through the proxies used in \cite{cunha_heckman_schennach2010}. For each candidate proxy, I compute a simple diagnostic measure of signal: the fraction of the proxy’s variance that can be attributed to the corresponding latent factor in a static factor model using basic or polychoric factor analysis. I keep proxies only if they exhibited at least 10\% signal in this diagnostic stage. Proxies with very low signal contribute little to identification but add computational cost, so I drop them. For example, I exclude the Rotter Locus of Control subscales commonly used as proxies for maternal noncognitive skill because they exhibit low signal even when modeled as discrete.\footnote{\cite{cunha_heckman_schennach2010} finds they have very low signal in their estimation as well.} The screen was applied to a pre-analysis diagnostic; signal in the estimated dynamic model can be different than it was in the diagnostic tests, and one proxy (reading to the child at 9–10) falls below the 10\% threshold in the final analysis.

To avoid complications arising from within-household spillovers and dynamic sibling interactions, I restrict the analysis to first-born children. For a child to be included in the model I require that I observe at least one proxy for each child skill and parental investment at ages 5-6, as well as all proxies for that child's maternal cognitive and noncognitive skills in the NLSY79. The final sample size is 2,402 child-mother pairs, not all of whom are observed at every age.

Finally, I construct a measure of maternal permanent income using observed household income from ages 22 to 57, which spans the relevant working-life window for all mothers in the sample. For the purposes of constructing permanent income alone, I linearly interpolate missing income observations between survey years. After the final observed income report, I assume household income remains constant at its last observed value. This interpolation is only done for the purposes of calculating household permanent income. The measures of annual household income I use directly in the model are not interpolated in this manner (see Appendix \ref{sec:app:it_likelihood} for details on how I model missing annual covariates). To accommodate log transformations, I add one dollar to all income measures. 

I calculate permanent income using a discount factor of 0.95.\footnote{Permanent Income $=\sum_{a=22}^{57} 0.95^{a-22} \omega_a$, where $\omega_a$ is net household income in 2022 dollars.} Because the NLSY79 was administered biennially after 1994, income in odd years from 1995 to 2021 is linearly interpolated using adjacent even-year observations. This extended income history allows me to distinguish short-run income fluctuations from longer-run resources when modeling parental investment decisions.

\section{An Empirical Model of Skill Formation}\label{sec:model_skill_form}

To trace out the trajectory of child development on real data, I revisit the analysis of \cite{cunha_heckman_schennach2010}, incorporating the elegant control function from \cite{attanasio_meghir_nix_2020}, a two-factor output extension of the normalization scheme of \cite{freyberger_2025}, and a threshold crossing model to measure factors from discrete proxies used in \cite{Montoya2024}. The model is jointly estimated across all periods from ages 5-14. Each period is a two-year block since the CNLSY was administered biennially. 

In this section, I describe the empirical model I use along with the econometric innovations of this paper: effective factors and using an ordered probit to model discrete proxies to measure continuous factors in a dynamic factor model.\footnote{While \cite{almuzara_garcia-vazquez_mullins2026} also use an ordinal measurement system, they also discretize their factors.} Section \ref{sec:measurement} describes how I recover latent factors from discrete and continuous proxies. Section \ref{sec:skill_prod} specifies the generalized CES production function I use. Section \ref{sec:factor_scales} introduces effective-factors, a framework that simplifies discussions of factor normalization and scale, and which allows my production functions to be scale invariant and semi-ordinal. This is important because my latent factors have no natural scale, and the ordinal discrete factors lack the natural distribution that continuous factors inherit from their proxies. 

In section \ref{sec:prod_est_cnlsy:invest_endog}, I discuss how I attempt to address endogeneity of parental investment, as well as how I model non-factor covariates. Prior work using the CNLSY has relied on short-run income as an exogenous shifter of investment, motivated by credit constraints. The longer income histories now available in the NLSY79 allow me to separate short-run income from permanent income and test this assumption. Lastly, I explain how to I calculate AMRs in Section \ref{sec:early_childhood_cog:constr_invest_elast}.

\subsection{Measurement system}\label{sec:measurement}

Unfortunately, I cannot directly observe latent child skills $\theta$ and maternal inputs $M$. Instead, the CNLSY contains proxies for these latent factors, like test scores and behavioral reports. Figure \ref{fig:proxy_visual} shows the relationship between child cognitive skill, $\theta_C$, and its proxies: PIAT math, reading comprehension, and vocabulary tests.
\begin{figure}
    \centering
\begin{adjustbox}{max width=\linewidth, center}
\begin{tikzpicture}[
  factor/.style={ellipse, draw=colorOne, fill=colorTwo, very thick,
                 minimum width=3.6cm, minimum height=1.8cm,
                 align=center, font=\Large},
  disc/.style={rectangle, rounded corners=4pt, draw=colorOne, fill=white,
               thick, dashed, minimum width=4.4cm, minimum height=1.0cm,
               font=\large},
  arr/.style={->, draw=colorOne, line width=1pt},
  eqlbl/.style={font=\normalsize, align=left, anchor=west}
]

\node[disc] (d1) at (0,  1.0) {Math};
\node[disc] (d2) at (-2.5, -0.2) {Reading Comprehension};
\node[disc] (d3) at (0, -1.4) {Vocabulary};

\node[factor] (fN) at (9, -0.2)
  {$\theta_{C,t}$\\[-0.1em]{\footnotesize\itshape child cognitive}};

\draw[arr] (d1.east) -- (fN.west);
\draw[arr] (d2.east) -- (fN.west);
\draw[arr] (d3.east) -- (fN.west);
\end{tikzpicture}
\end{adjustbox}
    \caption{Proxies for Child Cognitive Skill}
    \label{fig:proxy_visual}
\end{figure}
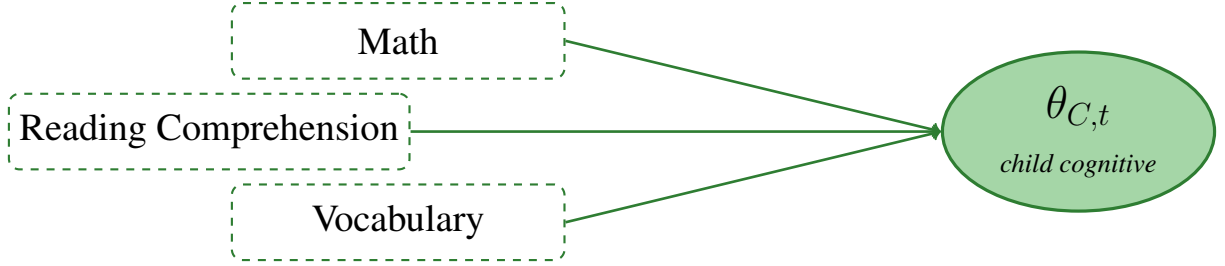

Both maternal and child proxies for cognitive skill are approximately continuous, and for those I use the standard log-linear measurement equation
\begin{equation}
        z_{ijt}=a_{ijt}+\lambda_{ijt}\,\ln\theta_{it}+X'_t\beta_{ijt}+\varepsilon_{ijt}, \label{eq:basic_cont_meas}
\end{equation}
where $j$ indexes proxies, $i$ indexes the latent factor being measured, $t$ indexes time, and $\varepsilon_{ijt}$ is a mean-zero measurement error. $X$ contains a vector of controls, which in practice is age at time of measurement.  Note that the distribution of factors that have continuous proxies is identified via the distribution of those proxies under this linear measurement model. 

While most previous models of dynamic skill formation all have treated proxies as though they are continuous \citep{cunha_heckman_schennach2010, agostinelli_wiswall_2025, attanasio_raquel_giannola_nores2026}, many CNLSY proxies are discrete (binary or Likert scale). Departing from previous studies, I model these discrete proxies using an ordered probit model, replacing continuous measurement equation (\ref{eq:basic_cont_meas}) with
\begin{equation} \label{eq:discrete_meas}
    \dot z_{ijt}=k\;\;\text{if}\;\;\nabla_{i,j,k-1,t}<\lambda_{ijt}\,\ln\theta_{it}+X'_t\beta_{ijt}+\varepsilon_{ijt}\le\nabla_{ijkt},
\end{equation}
where $K_{ij}$ is the number of categories and $\{\nabla_{ijkt}\}_{k=0}^{K_{ij}}$ are unknown cutoffs to be estimated, with $\nabla_{i0t}=-\infty$ and $\nabla_{iK_{ij}t}=\infty$. Note that the intercept $a_{ijt}$ from equation \ref{eq:basic_cont_meas} is absorbed into the cutoffs. 

To ensure my measurement system is identified, for each factor $i$ and each period $t$ I designate the first proxy ($j=1$) to be the anchor proxy for that factor in that period, which implicitly sets the scale and mean of the factor. To identify the variance of the factors, I set $\lambda_{i1t}=1$ for all $i$ and $t$. To identify the means of the factors, I assume that $a_{i1t}=0$ for continuous proxies and that $\nabla_{i,1,K_{i1}-1,t}=0$ for discrete proxies. 

My anchor proxies are as follows. Child cognitive skill is anchored by the PIAT mathematics subtest (continuous). Child noncognitive skill is anchored by headstrong subscale of the Behavioral Problems Index (discrete). Parental investment is anchored by the number of children's books in the child's home (discrete). The time-invariant maternal factors are anchored by the arithmetic reasoning subtest of the ASVAB (continuous) for maternal cognitive skill and by the Rosenberg self-esteem item ``I feel I do not have much to be proud of'' (discrete) for maternal noncognitive skill.

I also assume that $\varepsilon_{ijt}\sim N(0,\sigma^2_{ijt})$ for all continuous proxies and that $\varepsilon_{ijt}\sim N(0,1)$ for all discrete proxies, and that all $\varepsilon$'s are jointly independent across proxies, across time, and with the latent factors. Lastly, I model the log factors at ages 5-6 (except for investment, which is included at ages 3-4, not 5-6) as being distributed according to a mixture of two multivariate normal distributions, which I view as a flexible approximation of their true distribution. The choice of anchor proxy determines the cardinal scale of the factor. I discuss in Section \ref{sec:factor_scales} how to transform these estimated factors so that the elasticities recovered from estimated production functions are scale-invariant.

I will show later that the decision to explicitly model discrete proxies as discrete rather than continuous matters. Treating discrete proxies as continuous leads to substantial bias in Monte Carlo simulations. 
\begin{landscape}
\begin{footnotesize}
\begin{longtable}{l l r l c c c c}
\caption{Percent Signal by Proxy}
\label{tab:signal_table_it}\\
\toprule
Proxy & Age & $n$ & Type & $\widehat\lambda$ & $\widehat{\operatorname{Var}}(\ln f)$ & $\operatorname{Var}(\varepsilon)$ & Pct signal \\
\midrule
\endfirsthead
\multicolumn{8}{l}{\textit{(continued from previous page)}}\\
\toprule
Proxy & Age & $n$ & Type & $\widehat\lambda$ & $\widehat{\operatorname{Var}}(\ln f)$ & $\operatorname{Var}(\varepsilon)$ & Pct signal \\
\midrule
\endhead
\midrule
\multicolumn{8}{r}{\textit{continued on next page \ldots}}\\
\endfoot
\bottomrule
\endlastfoot
\multicolumn{8}{l}{\textbf{Cognitive child skill ($\theta_C$)}} \\
\midrule
PIAT Reading Comp. & 5--6 & 2,232 & Cts. & 1.644 & 0.302 & 0.183 & 0.817 \\
PIAT Math$^{a}$ & 5--6 & 2,386 & Cts. & 1.000 & 0.302 & 0.700 & 0.301 \\
PIAT Reading Recog. & 5--6 & 2,339 & Cts. & 1.807 & 0.302 & 0.021 & 0.980 \\
PIAT Reading Comp. & 7--8 & 1,912 & Cts. & 1.222 & 0.480 & 0.303 & 0.703 \\
PIAT Math$^{a}$ & 7--8 & 1,967 & Cts. & 1.000 & 0.480 & 0.532 & 0.474 \\
PIAT Reading Recog. & 7--8 & 1,968 & Cts. & 1.329 & 0.480 & 0.166 & 0.836 \\
PIAT Reading Comp. & 9--10 & 1,901 & Cts. & 1.103 & 0.548 & 0.377 & 0.638 \\
PIAT Math$^{a}$ & 9--10 & 1,919 & Cts. & 1.000 & 0.548 & 0.478 & 0.534 \\
PIAT Reading Recog. & 9--10 & 1,917 & Cts. & 1.217 & 0.548 & 0.223 & 0.784 \\
PIAT Reading Comp. & 11--12 & 1,785 & Cts. & 1.072 & 0.590 & 0.357 & 0.655 \\
PIAT Math$^{a}$ & 11--12 & 1,800 & Cts. & 1.000 & 0.590 & 0.442 & 0.571 \\
PIAT Reading Recog. & 11--12 & 1,802 & Cts. & 1.145 & 0.590 & 0.263 & 0.746 \\
PIAT Reading Comp. & 13--14 & 1,348 & Cts. & 1.012 & 0.625 & 0.409 & 0.610 \\
PIAT Math$^{a}$ & 13--14 & 1,355 & Cts. & 1.000 & 0.625 & 0.419 & 0.599 \\
PIAT Reading Recog. & 13--14 & 1,359 & Cts. & 1.118 & 0.625 & 0.275 & 0.740 \\
\midrule
\multicolumn{8}{l}{\textbf{Noncognitive child skill ($\theta_N$)}} \\
\midrule
BPI Antisocial & 5--6 & 2,357 & Disc. & 0.889 & 1.479 & 1.000 & 0.539 \\
BPI Anxious & 5--6 & 2,390 & Disc. & 0.754 & 1.479 & 1.000 & 0.457 \\
BPI Headstrong$^{a}$ & 5--6 & 2,393 & Disc. & 1.000 & 1.479 & 1.000 & 0.597 \\
BPI Hyperactive & 5--6 & 2,392 & Disc. & 0.923 & 1.479 & 1.000 & 0.558 \\
BPI Peer Conflict & 5--6 & 2,392 & Disc. & 0.713 & 1.479 & 1.000 & 0.429 \\
BPI Antisocial & 7--8 & 1,988 & Disc. & 0.842 & 1.674 & 1.000 & 0.543 \\
BPI Anxious & 7--8 & 2,024 & Disc. & 0.791 & 1.674 & 1.000 & 0.512 \\
BPI Headstrong$^{a}$ & 7--8 & 2,021 & Disc. & 1.000 & 1.674 & 1.000 & 0.626 \\
BPI Hyperactive & 7--8 & 2,028 & Disc. & 0.862 & 1.674 & 1.000 & 0.554 \\
BPI Peer Conflict & 7--8 & 2,033 & Disc. & 0.739 & 1.674 & 1.000 & 0.478 \\
BPI Antisocial & 9--10 & 1,927 & Disc. & 0.868 & 1.752 & 1.000 & 0.569 \\
BPI Anxious & 9--10 & 1,979 & Disc. & 0.814 & 1.752 & 1.000 & 0.537 \\
BPI Headstrong$^{a}$ & 9--10 & 1,984 & Disc. & 1.000 & 1.752 & 1.000 & 0.637 \\
BPI Hyperactive & 9--10 & 1,978 & Disc. & 0.821 & 1.752 & 1.000 & 0.542 \\
BPI Peer Conflict & 9--10 & 1,984 & Disc. & 0.793 & 1.752 & 1.000 & 0.524 \\
BPI Antisocial & 11--12 & 1,824 & Disc. & 0.833 & 2.117 & 1.000 & 0.595 \\
BPI Anxious & 11--12 & 1,873 & Disc. & 0.730 & 2.117 & 1.000 & 0.530 \\
BPI Headstrong$^{a}$ & 11--12 & 1,873 & Disc. & 1.000 & 2.117 & 1.000 & 0.679 \\
BPI Hyperactive & 11--12 & 1,873 & Disc. & 0.809 & 2.117 & 1.000 & 0.581 \\
BPI Peer Conflict & 11--12 & 1,872 & Disc. & 0.708 & 2.117 & 1.000 & 0.515 \\
BPI Antisocial & 13--14 & 1,420 & Disc. & 0.718 & 2.946 & 1.000 & 0.603 \\
BPI Anxious & 13--14 & 1,457 & Disc. & 0.713 & 2.946 & 1.000 & 0.600 \\
BPI Headstrong$^{a}$ & 13--14 & 1,455 & Disc. & 1.000 & 2.946 & 1.000 & 0.747 \\
BPI Hyperactive & 13--14 & 1,450 & Disc. & 0.687 & 2.946 & 1.000 & 0.582 \\
BPI Peer Conflict & 13--14 & 1,457 & Disc. & 0.644 & 2.946 & 1.000 & 0.550 \\
\midrule
\multicolumn{8}{l}{\textbf{Parental investment ($M_I$)}} \\
\midrule
CD player & 3--4 & 854 & Disc. & 0.530 & 2.524 & 1.000 & 0.415 \\
Mom reads to child & 3--4 & 1,656 & Disc. & 0.665 & 2.524 & 1.000 & 0.527 \\
Museum trips & 3--4 & 1,677 & Disc. & 0.276 & 2.524 & 1.000 & 0.161 \\
Num.\ books$^{a}$ & 3--4 & 1,656 & Disc. & 1.000 & 2.524 & 1.000 & 0.716 \\
Num.\ magazines & 3--4 & 1,647 & Disc. & 0.353 & 2.524 & 1.000 & 0.240 \\
Outings & 3--4 & 1,677 & Disc. & 0.245 & 2.524 & 1.000 & 0.131 \\
CD player & 5--6 & 597 & Disc. & 0.395 & 2.460 & 1.000 & 0.278 \\
Compliments to another & 5--6 & 946 & Disc. & 0.348 & 2.460 & 1.000 & 0.229 \\
Mom reads to child & 5--6 & 2,356 & Disc. & 0.525 & 2.460 & 1.000 & 0.404 \\
Musical instrument & 5--6 & 1,234 & Disc. & 0.321 & 2.460 & 1.000 & 0.202 \\
Num.\ books$^{a}$ & 5--6 & 2,400 & Disc. & 1.000 & 2.460 & 1.000 & 0.711 \\
Num.\ magazines & 5--6 & 1,207 & Disc. & 0.325 & 2.460 & 1.000 & 0.206 \\
Special lessons & 5--6 & 1,234 & Disc. & 0.513 & 2.460 & 1.000 & 0.393 \\
Times praised & 5--6 & 953 & Disc. & 0.415 & 2.460 & 1.000 & 0.298 \\
Compliments to another & 7--8 & 1,967 & Disc. & 0.441 & 1.472 & 1.000 & 0.222 \\
Mom reads to child & 7--8 & 2,037 & Disc. & 0.477 & 1.472 & 1.000 & 0.251 \\
Musical instrument & 7--8 & 2,030 & Disc. & 0.429 & 1.472 & 1.000 & 0.213 \\
Museum trips & 7--8 & 2,034 & Disc. & 0.409 & 1.472 & 1.000 & 0.198 \\
Newspaper & 7--8 & 2,027 & Disc. & 0.278 & 1.472 & 1.000 & 0.102 \\
Num.\ books$^{a}$ & 7--8 & 2,037 & Disc. & 1.000 & 1.472 & 1.000 & 0.596 \\
Special lessons & 7--8 & 2,032 & Disc. & 0.639 & 1.472 & 1.000 & 0.376 \\
Theater trips & 7--8 & 2,033 & Disc. & 0.449 & 1.472 & 1.000 & 0.229 \\
Times praised & 7--8 & 1,972 & Disc. & 0.477 & 1.472 & 1.000 & 0.251 \\
Compliments to another & 9--10 & 1,815 & Disc. & 0.543 & 1.137 & 1.000 & 0.251 \\
Mom reads to child & 9--10 & 971 & Disc. & 0.244 & 1.137 & 1.000 & 0.063 \\
Musical instrument & 9--10 & 1,947 & Disc. & 0.644 & 1.137 & 1.000 & 0.320 \\
Museum trips & 9--10 & 1,948 & Disc. & 0.489 & 1.137 & 1.000 & 0.214 \\
Newspaper & 9--10 & 1,946 & Disc. & 0.337 & 1.137 & 1.000 & 0.114 \\
Num.\ books$^{a}$ & 9--10 & 1,953 & Disc. & 1.000 & 1.137 & 1.000 & 0.532 \\
Special lessons & 9--10 & 1,947 & Disc. & 0.700 & 1.137 & 1.000 & 0.358 \\
Theater trips & 9--10 & 1,949 & Disc. & 0.546 & 1.137 & 1.000 & 0.253 \\
Times praised & 9--10 & 1,557 & Disc. & 0.561 & 1.137 & 1.000 & 0.264 \\
Compliments to another & 11--12 & 1,629 & Disc. & 0.497 & 1.028 & 1.000 & 0.202 \\
Musical instrument & 11--12 & 1,875 & Disc. & 0.653 & 1.028 & 1.000 & 0.304 \\
Museum trips & 11--12 & 1,877 & Disc. & 0.465 & 1.028 & 1.000 & 0.182 \\
Newspaper & 11--12 & 1,872 & Disc. & 0.351 & 1.028 & 1.000 & 0.112 \\
Num.\ books$^{a}$ & 11--12 & 1,874 & Disc. & 1.000 & 1.028 & 1.000 & 0.507 \\
Special lessons & 11--12 & 1,878 & Disc. & 0.689 & 1.028 & 1.000 & 0.328 \\
Theater trips & 11--12 & 1,880 & Disc. & 0.560 & 1.028 & 1.000 & 0.244 \\
Times praised & 11--12 & 1,477 & Disc. & 0.634 & 1.028 & 1.000 & 0.292 \\
Compliments to another & 13--14 & 1,227 & Disc. & 0.557 & 0.784 & 1.000 & 0.195 \\
Musical instrument & 13--14 & 1,460 & Disc. & 0.757 & 0.784 & 1.000 & 0.310 \\
Museum trips & 13--14 & 1,458 & Disc. & 0.522 & 0.784 & 1.000 & 0.176 \\
Newspaper & 13--14 & 1,462 & Disc. & 0.427 & 0.784 & 1.000 & 0.125 \\
Num.\ books$^{a}$ & 13--14 & 1,460 & Disc. & 1.000 & 0.784 & 1.000 & 0.439 \\
Special lessons & 13--14 & 1,462 & Disc. & 0.698 & 0.784 & 1.000 & 0.276 \\
Theater trips & 13--14 & 1,462 & Disc. & 0.625 & 0.784 & 1.000 & 0.235 \\
Times praised & 13--14 & 1,373 & Disc. & 0.653 & 0.784 & 1.000 & 0.251 \\
\midrule
\multicolumn{8}{l}{\textbf{Maternal cognitive skill ($M_C$)}} \\
\midrule
ASVAB Arithmetic$^{a}$ & Maternal & 2,402 & Cts. & 1.000 & 0.713 & 0.285 & 0.715 \\
ASVAB Math Knowl. & Maternal & 2,402 & Cts. & 1.007 & 0.713 & 0.275 & 0.725 \\
ASVAB Numerical & Maternal & 2,402 & Cts. & 0.855 & 0.713 & 0.477 & 0.522 \\
ASVAB Paragraph & Maternal & 2,402 & Cts. & 0.981 & 0.713 & 0.312 & 0.688 \\
ASVAB Speed & Maternal & 2,402 & Cts. & 0.778 & 0.713 & 0.568 & 0.432 \\
ASVAB Word & Maternal & 2,402 & Cts. & 1.010 & 0.713 & 0.271 & 0.728 \\
\midrule
\multicolumn{8}{l}{\textbf{Maternal noncognitive skill ($M_N$)}} \\
\midrule
Self-esteem: Failure (rev.) & Maternal & 2,402 & Disc. & 1.001 & 1.483 & 1.000 & 0.598 \\
Self-esteem: Good qualities & Maternal & 2,402 & Disc. & 0.781 & 1.483 & 1.000 & 0.475 \\
Self-esteem: No good (rev.) & Maternal & 2,402 & Disc. & 0.661 & 1.483 & 1.000 & 0.394 \\
Self-esteem: Not proud (rev.)$^{a}$ & Maternal & 2,402 & Disc. & 1.000 & 1.483 & 1.000 & 0.597 \\
Self-esteem: Positive att. & Maternal & 2,402 & Disc. & 0.773 & 1.483 & 1.000 & 0.470 \\
Self-esteem: Useless (rev.) & Maternal & 2,402 & Disc. & 0.519 & 1.483 & 1.000 & 0.286 \\
Self-esteem: Wish more respect (rev.) & Maternal & 2,402 & Disc. & 0.594 & 1.483 & 1.000 & 0.343 \\
Self-esteem: Worth & Maternal & 2,402 & Disc. & 0.708 & 1.483 & 1.000 & 0.426 \\
\end{longtable}
\end{footnotesize}
\noindent
\begin{minipage}{\linewidth}
\footnotesize
\textit{Notes:} Percent signal $= \widehat\lambda^2\operatorname{Var}(\ln f) / (\widehat\lambda^2\operatorname{Var}(\ln f) + \operatorname{Var}(\varepsilon))$, with $\operatorname{Var}(\varepsilon)=\widehat\sigma_\varepsilon^2$ for continuous proxies and $\operatorname{Var}(\varepsilon)=1$ for ordered-probit discrete proxies; $\operatorname{Var}(\ln f)$ is the variance of the latent factor at the proxy's age in the simulated cohort. $^{a}$ anchor proxy, $\widehat\lambda=1$ by normalization. $n$ is the number of children with a nonmissing value of the proxy ($N = 2,402$). The investment rows at ages 3--4 measure $\ln I_{-1}$, a coordinate of the period-0 mixture. \textit{AI disclosure:} table produced by code written with Claude Opus 5 (claude-opus-5) and Claude Sonnet 5 (claude-sonnet-5) and Claude Opus 5.5 (claude-opus-5-5), Anthropic, 2026-10-08 (CES-share kernel reused from MLE/estimator/plot\_elasticities.R, mixed model lineage; see git history); MLE/estimator/intertemporal/make\_it\_paper\_outputs.R::write\_it\_signal\_table().
\end{minipage}
\end{landscape}

\subsection{Skill production}\label{sec:skill_prod}

I parameterize the production function using a factor-scale-invariant constant-elasticity-of-substitution (CES) technology. Let $\theta_{i,t+1}$ denote latent child skill $i$ in period $t+1$, and let $y_{i',t}$ denote inputs indexed by $i'$, where $i'$ ranges over child skills and parental inputs at age $t$. The law of motion is
\begin{equation} \label{eq:CES_prod}
    \theta_{i,t+1}=\Bigl(\textstyle\sum_{i'}\gamma_{i,i',t}\,y_{i',t}^{\rho_{it}/\phi_{i,i',t}}\Bigr)^{\!\psi_{it}/\rho_{it}}\!\!\exp(Q'_t\kappa_{it}+\nu_{it})
\end{equation}
where $Q$ is a vector of observable controls (\ref{eq:CES_prod}) (an intercept, the number of children in the household, the child's gender, mother's marital status, mother's age, the interaction of number of children with mother's marital status and maternal education), and $\nu_{i,t}$ is an idiosyncratic shock with mean zero. Figure \ref{fig:prod_visual} provides a visual representation of equation \ref{eq:CES_prod}. 
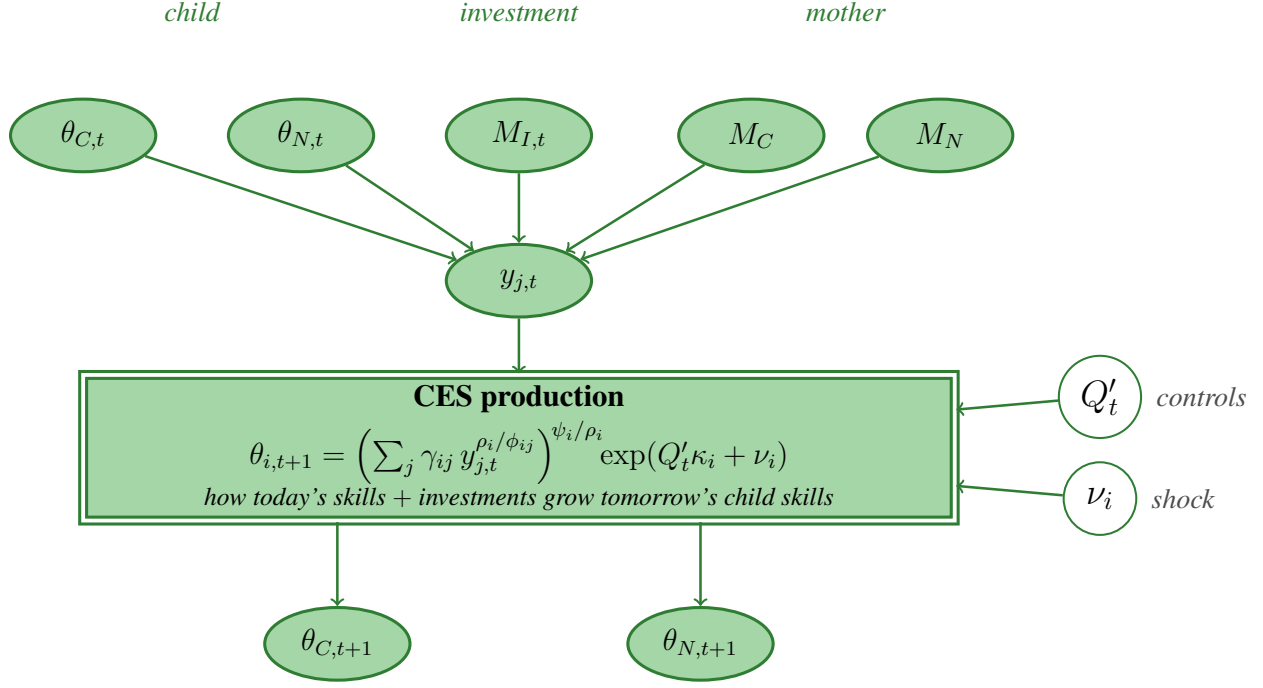
\begin{figure}
    \centering
\begin{adjustbox}{max width=\linewidth, center}
\begin{tikzpicture}[
  factor/.style={ellipse, draw=colorOne, fill=colorTwo, very thick,
                 minimum width=2.0cm, minimum height=1.0cm,
                 align=center, font=\normalsize},
  prodbox/.style={rectangle, draw=colorOne, fill=colorTwo, very thick,
                  double, double distance=1.5pt,
                  minimum width=12cm, minimum height=2.0cm,
                  align=center, font=\normalsize},
  shock/.style={circle, draw=colorOne, fill=white, line width=0.8pt,
                minimum size=1.0cm, font=\large},
  arr/.style={->, draw=colorOne, line width=1pt},
  caption/.style={font=\small\itshape, text=black!70},
  group/.style={font=\small\itshape, text=colorOne}
]
\node[group] at (-4.5, 5.7) {child};
\node[group] at ( 0.0, 5.7) {investment};
\node[group] at ( 4.5, 5.7) {mother};

\node[factor] (i1) at (-6.0, 4) {$\theta_{C,t}$};
\node[factor] (i2) at (-3.0, 4) {$\theta_{N,t}$};
\node[factor] (i3) at ( 0.0, 4) {$M_{I,t}$};
\node[factor] (i4) at ( 3.2, 4) {$M_C$};
\node[factor] (i5) at ( 5.8, 4) {$M_N$};

\node[factor] (y) at (0, 2) {$y_{j,t}$};

\node[prodbox] (ces) at (0, -0.3) {\textbf{CES production}\\[0.2em]
  $\theta_{i,t+1}=\Bigl(\textstyle\sum_j\gamma_{ij}\,y_{j,t}^{\rho_i/\phi_{ij}}\Bigr)^{\!\psi_i/\rho_i}\!\!\exp(Q'_t\kappa_i+\nu_i)$\\[0.05em]
  \itshape\footnotesize how today's skills $+$ investments grow tomorrow's child skills};

\node[shock] (nu) at (8, -1) {$\nu_i$};
\node[caption, right=0.05cm of nu] {shock};
\draw[arr] (nu) -- (ces);

\node[shock] (controls) at (8, .4) {$Q_t'$};
\node[caption, right=0.05cm of controls] {controls};
\draw[arr] (controls) -- (ces);

\node[factor] (oC) at (-2.5, -3) {$\theta_{C,t+1}$};
\node[factor] (oN) at ( 2.5, -3) {$\theta_{N,t+1}$};

\foreach \i in {i1,i2,i3,i4,i5}
  \draw[arr] (\i) -- (y);

\draw[arr] (y.south) -- (y |- ces.north);

\draw[arr] (oC |- ces.south) -- (oC.north);
\draw[arr] (oN |- ces.south) -- (oN.north);
\end{tikzpicture}
\end{adjustbox}
    \caption{Production Process}
    \label{fig:prod_visual}
\end{figure}

The inclusion of the intercept and the $\phi$ and $\psi$ parameters is essential for estimating factor-scale-invariant elasticities. When anchor proxies set factor means, input and output factors are generally on different scales; \cite{freyberger_2025} shows that an intercept in the production function absorbs mean differences between the input and output factors. Similarly, when anchor proxies set factor scales, the $\phi$ and $\psi$ parameters are needed to ensure that estimated elasticities are invariant to those scale choices.

One extension of \cite{freyberger_2025}'s framework in this paper is that the $\phi$'s are allowed to vary between the production functions for cognitive and noncognitive skill within periods. I discuss the motivation and implications of this extension in Section \ref{sec:factor_scales}. I normalize $\phi_{i,i,t}=1$ in all equations, which pins $\rho$ realtive to the acnhor scale of the produced skill.\footnote{It is worth noting that the substitution parameter $\rho$ is not and cannot be identified in a scale-invariant way in this model.}

\subsection{Factor Scale}\label{sec:factor_scales}

\cite{freyberger_2025} shows that the factor elasticities are identified for each production function separately, which could cause a problem here. However, it may be that factors behave as though they have different skew (technically variance of the log factor) in a cognitive skill production function than the do in a noncognitive skill production function. For example, parental investment might behave as though it were approximately normally distributed when producing noncognitive skill but behave as though it were heavily right-skewed when producing cognitive skill. The same underlying ordinal ranking of investment levels might require different cardinal scalings depending on the production context. \cite{freyberger_2025}'s normalization scheme implicitly enforces that investment acts like it has the same skew in both cognitive and noncognitive skill production. This implicit restriction could cause jointly estimated production functions to have different elasticities than separately estimated production functions, which is undesirable. 

Allowing the $\phi$'s and $\psi$'s to differ across production functions within a period allows the skew of input and output factors to differ by the skill being produced. I refer to the transformation of a factor induced by the $\phi$'s and $\psi$'s in Equation (\ref{eq:CES_prod}) as an \textit{effective factor}.

In the context of this paper, the effective factor for an input $i'$ to skill $i$ in time $t$ is $\gamma_{i,i',t}^{\operatorname{Sign}(\rho)}y_{i',t}^\frac{|\rho_{it}|}{\phi_{i,i',t}}$ and the effective factor of output $i$ in time $t+1$ is $\theta_{i,t+1}^\frac{|\rho_{it}|}{\psi_{it}}$. Note that I use $|\rho_{it}|$ because $\rho_{it}<0$ under gross complementarity; the sign is handled separately in equation (\ref{eq:effective_factor_elasts}). 

The identified economic parameters of interest are the elasticities of the effective factors
\begin{equation}\label{eq:effective_factor_elasts}
    \Tilde{\varepsilon}_{ijt}=\frac{\partial \Tilde{\theta}_{ijt+1}}{\partial \Tilde{y}_{ijt}}\frac{\Tilde{y}_{ijt}}{\Tilde{\theta}_{ijt+1}}=\frac{\Tilde{y}_{ijt}^{\operatorname{Sign}(\rho_{it})}}{\sum_{j'}\Tilde{y}_{ij't}^{\operatorname{Sign}(\rho_{it})}},
\end{equation}
which are invariant to choices of anchor proxies or measurement parameters under the measurement system described in Section \ref{sec:measurement}. Following the logic of \cite{freyberger_2025}, the $\phi$'s, $\gamma$'s, and $\psi$'s compensate for any changes to the estimated factor normalization, rendering the effective factor elasticities invariant to the normalizations used to set the log factors' scales. It is worth noting that a traditional CES production function handles differences in effective factor (not log factor) variance through the share parameters $\gamma$. The $\phi$'s and $\psi$'s are capturing differences in effective factor higher moments. Ordinal scales can enter into the child production function with different effective skews and higher moments in different periods of child development. 

Note that we can rewrite the CES production function from Equation (\ref{eq:CES_prod}) using effective factors as
\begin{equation}\label{eq:CES_eff_factor}
    \tilde{\theta}_{i,t+1}=\left(\sum_j\tilde{y}^{{\operatorname{Sign}(\rho_{it})}}_{j,t}\right)^{\frac{1}{\operatorname{Sign}(\rho_{it})}}\exp\left(\frac{|\rho_{it}|}{\psi_{it}}(Q'_t\kappa_i+\nu_i)\right).
\end{equation}
Equation (\ref{eq:CES_eff_factor}) shows that while the magnitude of the substitution parameter $\rho$ is fundamentally unidentified, its sign is. It also shows that the log inputs of any CES function can be rescaled to make the production function one of Cobb-Douglas, perfect substitutes, or gross complements with $\rho\in\{-1, 0, 1\}$. This highlights a subtle undercurrent of both this paper and \cite{freyberger_2025}: scale, complementarity, and curvature are not separately identified. Fortunately, \cite{freyberger_2025} shows that elasticities can be identified when the production function is CES, even if scale, complementarity, and convexity are not.

\subsection{Accounting for Endogeneity of Parental Investment}\label{sec:prod_est_cnlsy:invest_endog}

A central concern in estimating skill production functions is that parental investment may be endogenous. Parents may adjust investment in response to shocks to child development or to unobserved determinants of future skill growth, which would generally bias estimates of the causal effect of investment. Consistent with prior work, I attempt to address this concern using a control function approach that allows investment to be correlated with unobserved productivity shocks in the skill production function.

Following \cite{attanasio_meghir_nix_2020}, I model parental investment as a function of contemporaneous child skills, parental skills, and observed family characteristics. Specifically, I estimate the investment equation
\begin{equation} \label{eq:par_invest}
    \ln M_{I,t}=\beta_{M_I, t}\ln M_{I,t-1} + \ln \theta_t'\beta_{\theta,t} + \beta_{M_C,t}\ln M_{C}+\beta_{M_N,t}\ln M_{N}+Q_t'\beta_{Q,t}+\ln W_t'\beta_{W,t}+\eta_t,
\end{equation}
where $Q_t$ includes the same controls as used in the production function (\ref{eq:CES_prod}) (an intercept, the number of children in the household, the child's gender, mother's marital status, mother's age, the interaction of number of children with mother's marital status and maternal education) and $W_t$ contains log net household income and log maternal permanent household income. The error term $\eta_t$ captures unobserved determinants of parental investment and is assumed to be normally distributed with mean zero.

To account for the endogeneity of investment in the skill production function, I include the residual from equation (\ref{eq:par_invest}), denoted $\hat{\eta}_t$, as an additional control in $X_t$ in the CES production function (\ref{eq:CES_prod}). See Figure \ref{fig:invest_eq_visual} for a visual representation of this control function approach. This approach allows unobserved shocks that influence parental investment to be correlated with productivity shocks in skill formation. Formally, the identifying assumption is that
\[
E[\nu_{i,t} \mid \theta_t, M_{t-1}, Q_t, W_t, \eta_t] = \kappa_{i,\eta}\eta_t
\quad \text{for } i \in \{C,NC\},
\]
so that conditioning on $\hat{\eta}_t$ renders investment orthogonal to the remaining production shock. This strategy uses maternal permanent and temporary household income, as well as lagged investment, as a nonlinear exogenous shifter (like an instrument) for parental investment in period $t$. The use of current income as an exogenous shifter has historically been motivated by parents facing short-run credit constraints \citep{cunha_heckman_schennach2010, attanasio_meghir_nix_2020}. 
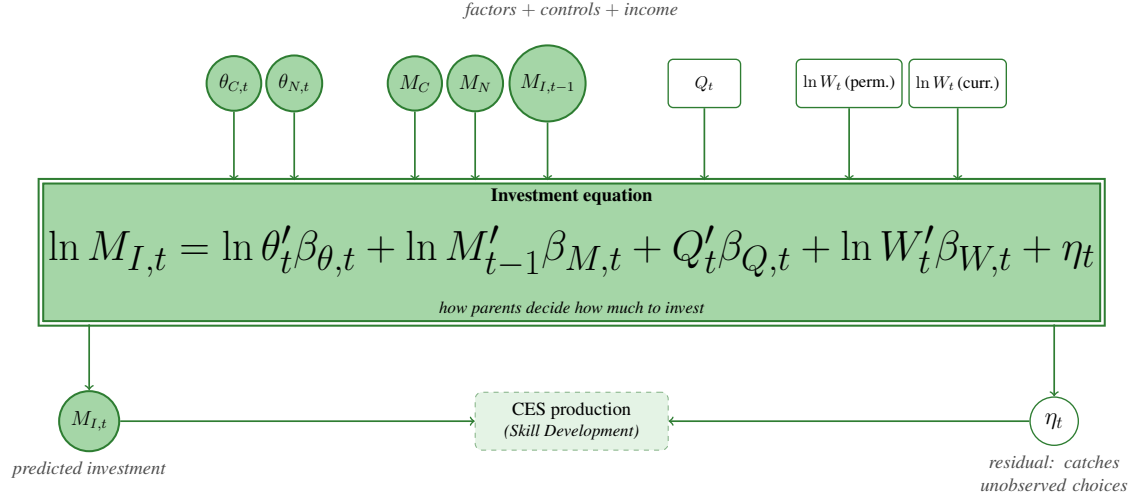
\begin{figure}
    \centering
\begin{adjustbox}{max width=\linewidth, center}
\begin{tikzpicture}[
  factor/.style={circle, draw=colorOne, fill=colorTwo, very thick,
                 minimum width=1.0cm, minimum height=1.0cm,
                 align=center, font=\normalsize},
  inputbox/.style={rectangle, rounded corners=3pt, draw=colorOne, fill=white,
                   thick, minimum width=1.5cm, minimum height=1.0cm,
                   font=\footnotesize},
  invbox/.style={rectangle, draw=colorOne, fill=colorTwo, very thick,
                 double, double distance=1.5pt,
                 minimum width=22cm, minimum height=2.0cm,
                 align=center, font=\huge},
  shock/.style={circle, draw=colorOne, fill=white, line width=0.8pt,
                minimum size=1.0cm, font=\large},
  arr/.style={->, draw=colorOne, line width=1pt},
  fbarr/.style={->, draw=colorOne, line width=1pt, dashed},
  caption/.style={font=\small\itshape, text=black!70},
  cesghost/.style={rectangle, draw=colorOne!60, dashed, rounded corners=4pt,
                   fill=colorTwo!30, minimum width=4.0cm, minimum height=1.2cm,
                   align=center, font=\small}
]
\node[factor]   (f1) at (-7, 3.0) {$\theta_{C,t}$};
\node[factor]   (f2) at (-5.75, 3.0) {$\theta_{N,t}$};
\node[factor]   (f3) at (-3.25, 3.0) {$M_C$};
\node[factor]   (f4) at (-2, 3.0) {$M_N$};
\node[factor]   (f5) at (-0.5, 3.0) {$M_{I,t-1}$};
\node[inputbox]  (q1) at (2.75, 3.0) {$Q_t$};
\node[inputbox] (y1) at (5.75, 3.0) {$\ln W_t$\,(perm.)};
\node[inputbox] (y2) at (8, 3.0) {$\ln W_t$\,(curr.)};

\node[caption] at (0, 4.5) {factors $+$ controls $+$ income};

\node[invbox] (ie) at (0, -0.5) {\small\textbf{Investment equation}\\[0.2em]
  $\ln M_{I,t}= \ln \theta_t'\beta_{\theta,t} + \ln M_{t-1}' \beta_{M, t} +Q_t'\beta_{Q,t}+\ln W_t'\beta_{W,t}+\eta_t$\\[0.05em]
  \itshape\footnotesize how parents decide how much to invest};

\node[factor] (II) at (-10, -4) {$M_{I,t}$};
\node[caption, below=0.05cm of II] {predicted investment};

\node[shock] (eta) at (10, -4) {$\eta_t$};
\node[caption, below=0.05cm of eta, align=center, text width=8cm]
  {residual: catches\\unobserved choices};

\node[cesghost] (ces) at (0, -4)
  {CES production\\[-0.1em]\itshape\footnotesize (Skill Development)};

\foreach \v in {f1,f2,f3,f4,f5,q1,y1,y2}
  \draw[arr] (\v.south) -- (\v |- ie.north);

\draw[arr] (II |- ie.south) -- (II.north);
\draw[arr] (eta |- ie.south) -- (eta.north);

\draw[arr] (II.east) -- (ces.west);
\draw[arr] (eta.west) -- (ces.east);
\end{tikzpicture}
\end{adjustbox}
    \caption{Investment Equation}
    \label{fig:invest_eq_visual}
\end{figure}

One departure from prior work is the inclusion of maternal permanent income in the investment equation. Prior work using the CNLSY was limited to short-run income measures because longer income histories were not yet available. I observe roughly 20 additional years of maternal income beyond what \cite{cunha_heckman_schennach2010} could. This extended history allows me to construct a permanent income measure, distinguish short-run fluctuations from longer-run resources, and directly test the credit-constraint motivation of the control function. If short-run income does not predict investment, it cannot serve as an exogenous shifter regardless of the underlying mechanism, and the standard control-function strategy does not address the endogeneity of parental investment. In my empirical analysis, lagged investment is by the strongest predictor of investment, which lends credence to this strategy even if the traditional household income instrument fails. 

In Section \ref{sec:results:investment} I show that short-run income does not consistently predict parental investment once I control for permanent income. This finding has two implications. First, short-run income is likely not a valid exogenous shifter of parental investment in this sample, regardless of the mechanism invoked to justify it. Prior work using the CNLSY that relies on short-run income for identification therefore requires an alternative identification strategy to support causal interpretation of investment elasticities.\footnote{It is important to note that \cite{cunha_heckman_schennach2010} includes such alternative identification strategies.} Second, my own production function estimates inherit this concern, which complicates attempts to interpret my results as causal effects.

\subsection{Initial Mixture and Covariates}\label{sec:model_skill_form:covariates}

Sections \ref{sec:skill_prod} and \ref{sec:prod_est_cnlsy:invest_endog} describe the laws of motion for child skill and investment. In this section, I describe how factors in period 0 and the covariates are modeled. Before estimation, all continuous factor proxies and covariates in the raw data are normalized to have mean 0 and variance 1 within each period. 

I model log child and maternal skills at ages 5-6, log investment at ages 3-4, log maternal household permanent income, log household income in period 0, mother's age at the NLSY interview corresponding to when her firstborn child was aged 5-6, and mother's education at age 30 as a mixture of two joint normal distributions. I model whether the child is male using a binary probit model using all the mixture factors and covariates as inputs. I model whether the mother is married at the NLSY79 interview when her child is 5-6 years old using a binary probit model that uses all the mixture variables and male as inputs. Lastly, I model the number of children in the household when the eldest is aged 5-6 using two probits. The first probit models having a second child in the house, and the second probit models having a third child conditional on having a second one. Both take all mixture variables, male, and married in period 0 as inputs. Less than two-and-a-half percent of households have more than 3 children, so I count all households that have more than three children has having three children. 

In addition to the factors, annual log household income, mather's marriage status, and number of children in the house are allowed to vary over time. Post period 0, these covariates are modeled as follows: mother's marital status is modeled with a binary probit that take in log current period child skills, log maternal skills, log lagged investment, log permanent income, lagged log annual household income, maternal education, maternal age at time of birth, lagged marriage, lagged number of kids, and the interaction of lagged marriage and lagged number of kids. Log current income is modeled via a linear equation using the same variables, except using current marital status rather than lagged marital status and excluding the lagged marriage-lagged number of kids interaction. 

The change in number of children in the household is modeled as an ordered probit that can take values of $-1$, $0$, and $1$. Its inputs are current period child skills, maternal skills, lagged investment, lagged log annual income, log permanent income, maternal age at time of birth, maternal education at age 30, both current and lagged marital status, lagged number of kids, and lagged number of kids interacted with lagged marital status. In each period, the order of factors and covariates produced is: child skills (simultaneously), mother's marital status, log period household income, change in number of kids, then log investment.

I view the models linking covariates to factors each other as entirely correlational, I do not view any of them as causal (other than the CES skill production functions and the investment equation). As a result, in all AMR calculations, I freeze the covariates and assume they do not respond to shocks to factors.

\subsection{Calculating AMRs}\label{sec:early_childhood_cog:constr_invest_elast}

To calculate AMRs in the structural model, I first estimate the model. Then, I simulate a dataset of one million child-mother pairs. I then calculate the elasticity of an outcome factor with respect to an input factor analytically. In this calculation, I freeze non-factor covariates so that, for example, a marginal increase in child investment at ages 5-6 does not affect the number of children in the house at ages 13-14. In my default specifications, I also shut off the immediate effects of the marginal increase in one factor on investment in the same period. Future investment may respond to a marginal increase in skill this period through the skill production channels. For example, an increase in child cognitive skill at ages 5-6 will not affect investment at ages 5-6, though it may affect investment at ages 7-8 through its effect on child skills at ages 7-8. For investment, I shut down the immediate investment lag in my default specification, so that a marginal increase in investment in period $t-1$ only affects investment in period $t$ through its effect on skills in period $t$, not through the investment lag. The investment lag is unaltered for future periods. Note that many studies \citep{attanasio_cattan_meghir_rubio_codina2020, jha_list_royer_samek2026, bassi_et_al2026} find that early investment in children increases later investment in children, so this assumption may be wrong. I present alternative results that allow full pass through of investment through the lag in Appendix Figures \ref{fig:amr_cog_long_propI}-\ref{fig:amr_noncog_invage_propI}. 

After calculating the elasticity for each simulated observation, I average elasticities either across the whole data set or within percentile bins of some target factor (e.g. cognitive skill at ages 5-6) to calculate the relevant AMR. This numerically integrates elasticities across the joint distribution of inputs without forcing the computation of an analytically complicated integral.

\subsection{The AMR of Investment at Ages 5--6 on Cognitive Skill at Ages 13--14}\label{sec:model_skill_form:amr_inv0_cog4}
 
In this subsection, I walk through how to calculate the long run AMR on investment at ages 5-6 on cognitive skill at ages 13-14. While normally this kind of in-depth walk-through is best saved for an appendix, the equation for the long run effective factor elasticity is highly economically significant and is deeply tied to Dynamic Complementarity, so I am including this in the paper's main body.  

I index periods from the first period of the intertemporal model, so that $t=0$ is ages 5--6 and $t=4$ is ages 13--14. Let $y_t=(\theta_{C,t},\theta_{N,t},M_{I,t},M_C,M_N)$ collect the inputs to production in period $t$, with $y_{C,t}=\theta_{C,t}$, $y_{N,t}=\theta_{N,t}$, and $y_{I,t}=M_{I,t}$, and let $\beta_{M_I,t}$ denote the coefficient on $\ln M_{I,t-1}$. The one-period elasticity of skill $i$ with respect to input $j$ implied by the production function (\ref{eq:CES_prod}) is
\begin{equation}\label{eq:amr_inv0_zeta}
    \zeta_{i,j,t}\equiv\frac{\partial\ln\theta_{i,t+1}}{\partial\ln y_{j,t}}
    =\underbrace{\frac{\psi_{it}}{\phi_{i,j,t}}}_{\substack{\text{raw-unit}\\ \text{rescaling}}}\;
    \underbrace{\frac{\gamma_{i,j,t}\,y_{j,t}^{\rho_{it}/\phi_{i,j,t}}}{\sum_{j'}\gamma_{i,j',t}\,y_{j',t}^{\rho_{it}/\phi_{i,j',t}}}}_{\tilde{\varepsilon}_{i,j,t}\text{: effective-factor elasticity}},
    \qquad i\in\{C,N\},\quad j\in\{C,N,I\},
\end{equation}
where the sum in the denominator runs over all five inputs and $\phi_{i,i,t}=1$. Note that $\psi_{it}$ and $\rho_{i,t,t}$ absorb changes in the log factor variances here. Multiplying both sides by $\frac{\phi_{i,j,t}}{\psi_{it}}$ shows how the extra CES coefficients in Equation (\ref{eq:CES_prod}) cause the effective factor elastiticies to be scale invariant. 

I collect these elasticities and the skill coefficients of the investment equation into
\begin{equation}\label{eq:amr_inv0_blocks}
    \renewcommand{\arraystretch}{0.6}
    \mathbf{S}_t=\underbrace{\begin{pmatrix}\zeta_{C,C,t}&\zeta_{C,N,t}\\ \zeta_{N,C,t}&\zeta_{N,N,t}\end{pmatrix}}_{\substack{\text{skills at }t\text{ to skills at }t+1\text{:}\\ \text{persistence and cross-skill effects}}},
    \qquad
    \mathbf{m}_t=\underbrace{\begin{pmatrix}\zeta_{C,I,t}\\ \zeta_{N,I,t}\end{pmatrix}}_{\substack{\text{investment at }t\\ \text{to skills at }t+1}},
    \qquad
    \boldsymbol\beta_{\theta,t}=\underbrace{\begin{pmatrix}\beta_{\theta_C,t}\\ \beta_{\theta_N,t}\end{pmatrix}}_{\substack{\text{skills at }t\\ \text{to investment at }t}},
    \qquad
    \mathbf{e}_C=\begin{pmatrix}1\\ 0\end{pmatrix}.
\end{equation}
I marginally increase $\ln M_{I,0}$ while holding the investment shock $\eta_0$, the production shocks $\nu_{i,t}$, the later investment shocks $\eta_t$, and the covariates $Q_t$ and $W_t$ at their realized values. All of these enter additively in logs, so they drop out of every derivative and matter only through the point $y_t$ at which each $\zeta_{i,j,t}$ is evaluated. Because maternal skills are fixed, only $\theta_{C,t}$, $\theta_{N,t}$, and $M_{I,t}$ move. The AMR of effective investment at ages 5--6 on effective cognitive skill at ages 13--14 is then
\begin{equation}\label{eq:amr_inv0_cog4}
\begin{aligned}
    \E\!\left[\frac{d\ln\tilde{\theta}_{C,4}}{d\ln\tilde{M}_{I,0}}\right]
    ={}&\underbrace{\frac{|\rho_{C,3}|/\psi_{C,3}}{|\rho_{C,0}|/\phi_{C,I,0}}}_{\substack{\text{raw log units to}\\ \text{effective units}}}
    \;\E\Biggl[\;\underbrace{\mathbf{e}_C'\,\mathbf{S}_3\mathbf{S}_2\mathbf{S}_1\,\mathbf{m}_0}_{\substack{\text{ages 5--6 investment raises skills at ages 7--8, which}\\ \text{persistence and cross-skill effects carry to ages 13--14}}}
    \\[4pt]
    &+\underbrace{\mathbf{e}_C'\,\mathbf{S}_3\mathbf{S}_2\,\mathbf{m}_1}_{\substack{\text{return to ages 7--8}\\ \text{investment through skills}}}
    \;\underbrace{\Bigl(\underbrace{\beta_{M_I,1}}_{\text{persistence}}+\underbrace{\boldsymbol\beta_{\theta,1}'\mathbf{m}_0}_{\substack{\text{response to the skills}\\ \text{it just produced}}}\Bigr)}_{\text{induced change in ages 7--8 investment}}
    \\[4pt]
    &+\underbrace{\mathbf{e}_C'\,\mathbf{S}_3\,\mathbf{m}_2}_{\substack{\text{return to ages 9--10}\\ \text{investment through skills}}}
    \;\underbrace{\Bigl\{\underbrace{\prod_{\tau=1}^{2}\bigl(\beta_{M_I,\tau}+\boldsymbol\beta_{\theta,\tau}'\mathbf{m}_{\tau-1}\bigr)}_{\text{through ages 7--8 investment}}
    +\underbrace{\boldsymbol\beta_{\theta,2}'\mathbf{S}_1\,\mathbf{m}_0}_{\substack{\text{carried by skills}\\ \text{past ages 7--8}}}\Bigr\}}_{\text{induced change in ages 9--10 investment}}
    \\[4pt]
    &+\underbrace{\mathbf{e}_C'\,\mathbf{m}_3}_{\substack{\text{return to ages 11--12}\\ \text{investment}}}
    \;\Bigl\{\underbrace{\prod_{\tau=1}^{3}\bigl(\beta_{M_I,\tau}+\boldsymbol\beta_{\theta,\tau}'\mathbf{m}_{\tau-1}\bigr)}_{\text{through ages 7--8 and 9--10 investment}}
    +\underbrace{\bigl(\beta_{M_I,1}+\boldsymbol\beta_{\theta,1}'\mathbf{m}_0\bigr)\,\boldsymbol\beta_{\theta,3}'\mathbf{S}_2\,\mathbf{m}_1}_{\substack{\text{through ages 7--8 investment,}\\ \text{carried by skills past ages 9--10}}}
    \\[4pt]
    &\hphantom{+\underbrace{\mathbf{e}_C'\,\mathbf{m}_3}_{\substack{\text{return to ages 11--12}\\ \text{investment}}}\;\Bigl\{}
    +\underbrace{\boldsymbol\beta_{\theta,2}'\mathbf{S}_1\,\mathbf{m}_0\,\bigl(\beta_{M_I,3}+\boldsymbol\beta_{\theta,3}'\mathbf{m}_2\bigr)}_{\substack{\text{carried by skills past ages 7--8,}\\ \text{through ages 9--10 investment}}}
    +\underbrace{\boldsymbol\beta_{\theta,3}'\mathbf{S}_2\mathbf{S}_1\,\mathbf{m}_0}_{\substack{\text{carried by skills past}\\ \text{ages 7--8 and 9--10}}}\Bigr\}
    \;\Biggr].
\end{aligned}
\end{equation}
Every $\zeta_{i,j,t}$, and so every $\mathbf{S}_t$ and $\mathbf{m}_t$, is evaluated at the child's own inputs $y_t$, so the Total AMR is equal to this expectation, which is taken over the joint distribution of $(y_0,y_1,y_2,y_3)$ implied by the estimated model. Setting $\beta_{M_I,1}=0$ gives the default AMRs I present in the main body of the paper in which ages 7--8 investment does not inherit the ages 5--6 increase directly. To calculate AMRs by cognitive skill percentile at age 5-6, condition the expectations on cognitive skill percentile at ages 5-6.

It is worth taking a moment to remark on the effective factor conversion scalar, $\frac{|\rho_{C,3}|/\psi_{C,3}}{|\rho_{C,0}|/\phi_{C,I,0}}$. It is a clear manifestation of the theory of Dynamic Complementarity laid out in \cite{heckman_mosso2014} and hinted at in \cite{cunha_heckman_schennach2010}. If later periods have a more negative subsitution parameter ($\rho_{C,3}$) than earlier periods do, than early investment will snowball into later periods, and the AMR will increase over time. 

Unfortunately, it is impossible to separate the effects of Dynamic Complementarity from changes in returns to scale and log factor scales, as the $\psi$ and $\phi$ in the effective unit conversion scalar remind us. If the AMR of investment is increasing over time, this model cannot tell whether that is caused by increasing complementarity, decreasing returns to scale, or increasing log outcome factor variance. Nonetheless, Equation (\ref{eq:amr_inv0_cog4}) shows how increasing complementarity would manifest in the empirically large returns to early childhood programs \citep{hendren_sprung-keyser2020}, even if we cannot separate complementarity from changes in scale with the tools we currently have. This also highlights the economic value of using effective factors: they clarify the relationship between certain theories of child development and empirical models that use them.

\section{Monte Carlo Simulations}\label{sec:mc}

This section uses Monte Carlo simulations to both show that the AMRs generated by the model are unbiased, recovers heterogeneity in the AMR well, and to quantify the consequences of treating discrete proxies as continuous in a dynamic latent factor model.

\subsection{Design}\label{sec:mc_design}

I simulate data from the model in Section \ref{sec:model_skill_form} with known parameter values. I try to replicate the enviroment of the real data as closely as possible. Each period's simulated factors are designed to have identical marginal distributions as the observed factors. My sample size is 2,402, and I account for the missing nature of my data. To do that, I record in each real observation which proxies, covariates, and periods it is missing. I port over that missing data structure to my simulated data, and match the missingness structure of each simulated child to a real one at random, so that each simulated child has the same missing data pattern as one real child. 

I have two variants of my Monte Carlo design, M1 and M2. The M1 estimator corresponds to the estimator used in this paper. The M2 estimator treats all proxies as if they were continuous and demonstrates that treating discrete proxies as continuous biases both the levels and the shapes of AMR curves.

\subsection{Monte Carlo Results}\label{sec:mc_results}

Tables \ref{tab:mc_it_period_amr} and \ref{tab:mc_it_total_amr} report absolute bias in the expected effective factor AMRS for both estimators across 50\footnote{The 50 replications are still in progress. This section will eventually show results for 50 replications, but I do not have them all yet.} Monte Carlo replications. AMRs are calculated as per Section \ref{sec:early_childhood_cog:constr_invest_elast}. Table \ref{tab:mc_it_period_amr} shows the period-by-period short term AMR recovery, while Table \ref{tab:mc_it_total_amr} shows the long run AMR recovery of factors at ages 5-6 for skills at ages 13-14. For each replication $r$, I compute the paired difference between the estimated and true AMR, where the true AMR is evaluated at the data-generating parameters. Both tables show the median fo that difference is taken across replications to limit the influence of outliers from occasional poorly-behaved optimization runs. Statistical significance is assessed using percentile-based confidence intervals constructed from the empirical distribution of paired differences across replications: $\dagger$, $\dagger\dagger$, and $\dagger\dagger\dagger$ indicate that the confidence interval excludes zero at the 10\%, 5\%, and 1\% levels, respectively. Because significance here indicates a \textit{failure} to recover the true parameter, daggers rather than stars are used to distinguish these markers from conventional hypothesis tests, where rejection is typically the desired outcome. I also include the 90\% confidence interval of paired differences across replications in brackets under the median difference.

\begin{landscape}
\begingroup\renewcommand{\baselinestretch}{1}%
\footnotesize\setlength{\tabcolsep}{2pt}\setlength{\LTpost}{0pt}\renewcommand{\arraystretch}{0.9}
\begin{longtable}{lcccccccccccc}
\caption{Per-Transition Average Marginal Returns (AMR) of Inputs: Bias in the Population-Average Effective Elasticity ($\tilde{\varepsilon}$), Median Paired Difference from True}
\label{tab:mc_it_period_amr} \\
\toprule
 & \multicolumn{3}{c}{Ages 5--6 to 7--8} & \multicolumn{3}{c}{Ages 7--8 to 9--10} & \multicolumn{3}{c}{Ages 9--10 to 11--12} & \multicolumn{3}{c}{Ages 11--12 to 13--14} \\
\cmidrule(lr){2-4} \cmidrule(lr){5-7} \cmidrule(lr){8-10} \cmidrule(lr){11-13}
Input & True & M1 & M2 (Cont) & True & M1 & M2 (Cont) & True & M1 & M2 (Cont) & True & M1 & M2 (Cont) \\
\midrule
\endfirsthead
\multicolumn{13}{l}{\textit{Table~\thetable{} (continued)}} \\
\toprule
 & \multicolumn{3}{c}{Ages 5--6 to 7--8} & \multicolumn{3}{c}{Ages 7--8 to 9--10} & \multicolumn{3}{c}{Ages 9--10 to 11--12} & \multicolumn{3}{c}{Ages 11--12 to 13--14} \\
\cmidrule(lr){2-4} \cmidrule(lr){5-7} \cmidrule(lr){8-10} \cmidrule(lr){11-13}
Input & True & M1 & M2 (Cont) & True & M1 & M2 (Cont) & True & M1 & M2 (Cont) & True & M1 & M2 (Cont) \\
\midrule
\endhead
\multicolumn{13}{r}{\textit{Continued on next page}} \\
\endfoot
\bottomrule
\endlastfoot
\multicolumn{13}{l}{\textit{Panel A: Cognitive skill next period ($\theta_{C,t+1}$)}} \\*[0.5ex]
\quad $\theta_C$ & 0.495 & -0.002 & -0.030 & 0.452 & -0.009 & -0.006 & 0.711 & -0.004 & -0.008 & 0.851 & -0.020 & -0.009 \\*
\quad  &  & [-0.06, 0.08] & [-0.09, 0.03] &  & [-0.05, 0.04] & [-0.05, 0.02] &  & [-0.07, 0.05] & [-0.08, 0.04] &  & [-0.10, 0.01] & [-0.06, 0.02] \\[0.3ex]
\quad $\theta_N$ & 0.033 & +0.000 & +0.006 & 0.030 & -0.000 & +0.003 & 0.010 & +0.000 & -0.000 & 0.024 & +0.005 & -0.006$^{\dagger}$ \\*
\quad  &  & [-0.01, 0.02] & [0.00, 0.01] &  & [0.00, 0.01] & [0.00, 0.01] &  & [0.00, 0.00] & [0.00, 0.00] &  & [0.00, 0.02] & [-0.01, 0.00] \\[0.3ex]
\quad $M_I$ & 0.108 & +0.005 & +0.084$^{\dagger\dagger\dagger}$ & 0.055 & +0.001 & +0.054$^{\dagger\dagger\dagger}$ & 0.020 & -0.001 & +0.020$^{\dagger\dagger\dagger}$ & 0.021 & +0.001 & +0.012$^{\dagger\dagger\dagger}$ \\*
\quad  &  & [-0.04, 0.04] & [0.05, 0.12] &  & [-0.01, 0.02] & [0.04, 0.09] &  & [0.00, 0.01] & [0.01, 0.03] &  & [0.00, 0.01] & [0.00, 0.02] \\[0.3ex]
\quad $M_C$ & 0.301 & -0.017 & -0.077$^{\dagger\dagger\dagger}$ & 0.370 & -0.008 & -0.055$^{\dagger}$ & 0.208 & +0.001 & -0.018 & 0.054 & +0.012 & -0.000 \\*
\quad  &  & [-0.04, 0.04] & [-0.11, -0.02] &  & [-0.04, 0.05] & [-0.11, -0.02] &  & [-0.05, 0.05] & [-0.06, 0.04] &  & [-0.01, 0.04] & [-0.01, 0.02] \\[0.3ex]
\quad $M_N$ & 0.063 & +0.001 & +0.007$^{\dagger\dagger\dagger}$ & 0.093 & +0.001 & +0.003 & 0.052 & +0.001 & +0.005 & 0.049 & +0.004 & +0.001 \\*
\quad  &  & [0.00, 0.01] & [0.00, 0.01] &  & [0.00, 0.02] & [0.00, 0.01] &  & [0.00, 0.01] & [0.00, 0.01] &  & [0.00, 0.02] & [-0.01, 0.01] \\[0.3ex]
\midrule
\multicolumn{13}{l}{\textit{Panel B: Noncognitive skill next period ($\theta_{N,t+1}$)}} \\*[0.5ex]
\quad $\theta_C$ & 0.009 & +0.000 & +0.032$^{\dagger\dagger\dagger}$ & 0.012 & -0.000 & +0.009$^{\dagger\dagger\dagger}$ & 0.008 & -0.000 & +0.003$^{\dagger\dagger\dagger}$ & 0.019 & -0.000 & +0.034$^{\dagger\dagger\dagger}$ \\*
\quad  &  & [0.00, 0.01] & [0.02, 0.05] &  & [0.00, 0.00] & [0.00, 0.01] &  & [0.00, 0.00] & [0.00, 0.01] &  & [-0.01, 0.01] & [0.03, 0.05] \\[0.3ex]
\quad $\theta_N$ & 0.949 & -0.010 & -0.153$^{\dagger\dagger\dagger}$ & 0.960 & -0.000 & -0.079$^{\dagger\dagger\dagger}$ & 0.967 & -0.001 & -0.038$^{\dagger\dagger\dagger}$ & 0.900 & -0.004 & -0.178$^{\dagger\dagger\dagger}$ \\*
\quad  &  & [-0.08, 0.01] & [-0.20, -0.12] &  & [-0.01, 0.01] & [-0.11, -0.05] &  & [-0.01, 0.00] & [-0.06, -0.03] &  & [-0.04, 0.03] & [-0.22, -0.14] \\[0.3ex]
\quad $M_I$ & 0.009 & +0.003 & +0.075$^{\dagger\dagger\dagger}$ & 0.004 & +0.000 & +0.030$^{\dagger\dagger\dagger}$ & 0.003 & +0.000 & +0.015$^{\dagger\dagger\dagger}$ & 0.021 & +0.000 & +0.052$^{\dagger\dagger\dagger}$ \\*
\quad  &  & [0.00, 0.04] & [0.05, 0.09] &  & [0.00, 0.00] & [0.02, 0.04] &  & [0.00, 0.00] & [0.01, 0.02] &  & [-0.01, 0.01] & [0.04, 0.07] \\[0.3ex]
\quad $M_C$ & 0.005 & +0.001 & +0.012$^{\dagger\dagger\dagger}$ & 0.007 & -0.000 & +0.009$^{\dagger\dagger\dagger}$ & 0.006 & +0.001 & +0.005$^{\dagger\dagger\dagger}$ & 0.016 & -0.000 & +0.027$^{\dagger\dagger\dagger}$ \\*
\quad  &  & [0.00, 0.00] & [0.01, 0.02] &  & [0.00, 0.00] & [0.01, 0.01] &  & [0.00, 0.00] & [0.00, 0.01] &  & [0.00, 0.01] & [0.02, 0.04] \\[0.3ex]
\quad $M_N$ & 0.028 & +0.004 & +0.040$^{\dagger\dagger\dagger}$ & 0.017 & -0.000 & +0.030$^{\dagger\dagger\dagger}$ & 0.016 & +0.001 & +0.016$^{\dagger\dagger\dagger}$ & 0.045 & +0.000 & +0.066$^{\dagger\dagger\dagger}$ \\*
\quad  &  & [0.00, 0.02] & [0.03, 0.05] &  & [0.00, 0.00] & [0.02, 0.04] &  & [0.00, 0.00] & [0.01, 0.02] &  & [-0.01, 0.01] & [0.05, 0.08] \\[0.3ex]
\end{longtable}
\begingroup\footnotesize\setlength{\parindent}{0pt}\setlength{\parskip}{0.4em}%
\setlength{\leftskip}{0.025\textwidth}\setlength{\rightskip}{0.025\textwidth}%
\textit{Notes:} Cells give the median over 20 MC replications of the paired difference (estimated $-$ true) in the population-average effective elasticity $\tilde{\varepsilon}$ of next-age skill with respect to the input in that transition (the input's CES share), with its 90\% interval (5th and 95th percentiles of the paired differences across replications) in brackets; True is its value at the data-generating parameters. Daggers, M1 and M2 are as in Table~\ref{tab:mc_it_total_amr}.

{\scriptsize\setlength{\rightskip}{0.025\textwidth plus 1fil}%
\textit{AI disclosure:} table produced by code written with Claude Sonnet 5.5 (claude-sonnet-5-5), Anthropic, 2026-10-08 (revised 2026-10-05 to 2026-10-07; layout of paper/Tables/avg\_elasticities\_all\_periods.tex; helpers shared with the total-AMR writer, mixed model lineage, see git history); MLE/estimator/intertemporal/mc/it\_mc\_analysis.R::it\_mc\_write\_period\_amr\_table().
\par}
\par\endgroup
\endgroup
\end{landscape}

The results from table \ref{tab:mc_it_period_amr} are extremely promising for my discrete proxy estimator, M1. In every period and input factor, the median difference is small, and the 90\% confidence interval is tight relative to the size of the true AMR. The estimator I use in my empirical application recovers short-run AMRs quite well. 

The M2 rows of Table \ref{tab:mc_it_period_amr} show that treating discrete proxies as continuous causes substantial bias. For cognitive outcomes, this bias is particularly consistent in investment: treating discrete investment proxies as continuous causes the model to significantly overstate the importance of investment in every period. For noncognitve outcomes, the continuous proxy estimator is significantly biased for every input factor in every period. 

Table \ref{tab:mc_it_total_amr} shows that the discrete M1 estimator's long run AMR on cognitive skill recovery is unbiased, but imprecise. This imprecision is due to the effective factor conversion from period to period, which for factor $i$ is $\frac{|\rho_{C,3}|\phi_{C,i,0}}{|\rho_{C,0}|\psi_{C,3}}$, where $t=0$ is period $0$ (ages 5-6) and $t=3$ is ages 11-12 (not ages 13-14). CES substitution parameters are notoriously poorly identified, and the model under realistic sample size and missing data conditions cannot consistently pin it down. If $\rho_{C,0}$ is estimated to be close to 0, it causes wild swings in the level of the AMR across periods. 

\begin{landscape}
\begingroup\renewcommand{\baselinestretch}{1}%
\footnotesize\setlength{\tabcolsep}{3pt}
\begin{longtable}{lcccccccccc}
\caption{Total Average Marginal Returns (AMR) of Inputs at Ages 5--6 on Skill at Ages 13--14: Bias in the Population-Average Effective Elasticity ($\tilde{\varepsilon}$), Median Paired Difference from True, with the Ages 5--6 Curvature Estimated ($\hat\rho_0$) and Held Fixed ($\bar\rho_0$)}
\label{tab:mc_it_total_amr} \\
\toprule
 & \multicolumn{10}{c}{Skill at ages 13--14} \\
\cmidrule(lr){2-11}
 & \multicolumn{5}{c}{Cognitive ($\theta_C$)} & \multicolumn{5}{c}{Noncognitive ($\theta_N$)} \\
\cmidrule(lr){2-6}\cmidrule(lr){7-11}
 & True AMR & \multicolumn{2}{c}{M1} & \multicolumn{2}{c}{M2 (Cont)} & True AMR & \multicolumn{2}{c}{M1} & \multicolumn{2}{c}{M2 (Cont)} \\
\cmidrule(lr){3-4}\cmidrule(lr){5-6}\cmidrule(lr){8-9}\cmidrule(lr){10-11}
Input at ages 5--6 &  & $\hat\rho_0$ & $\bar\rho_0$ & $\hat\rho_0$ & $\bar\rho_0$ &  & $\hat\rho_0$ & $\bar\rho_0$ & $\hat\rho_0$ & $\bar\rho_0$ \\
\midrule
\endfirsthead
\multicolumn{11}{l}{\textit{Table~\thetable{} (continued)}} \\
\toprule
 & \multicolumn{10}{c}{Skill at ages 13--14} \\
\cmidrule(lr){2-11}
 & \multicolumn{5}{c}{Cognitive ($\theta_C$)} & \multicolumn{5}{c}{Noncognitive ($\theta_N$)} \\
\cmidrule(lr){2-6}\cmidrule(lr){7-11}
 & True AMR & \multicolumn{2}{c}{M1} & \multicolumn{2}{c}{M2 (Cont)} & True AMR & \multicolumn{2}{c}{M1} & \multicolumn{2}{c}{M2 (Cont)} \\
\cmidrule(lr){3-4}\cmidrule(lr){5-6}\cmidrule(lr){8-9}\cmidrule(lr){10-11}
Input at ages 5--6 &  & $\hat\rho_0$ & $\bar\rho_0$ & $\hat\rho_0$ & $\bar\rho_0$ &  & $\hat\rho_0$ & $\bar\rho_0$ & $\hat\rho_0$ & $\bar\rho_0$ \\
\midrule
\endhead
\multicolumn{11}{r}{\textit{Continued on next page}} \\
\endfoot
\bottomrule
\endlastfoot
\quad $\theta_C$ & 1.838 & -0.069 & -0.169 & +0.014 & -0.079 & 0.015 & +0.001 & +0.002 & +0.023$^{\dagger\dagger\dagger}$ & +0.023$^{\dagger\dagger\dagger}$ \\*
\quad  &  & [-0.63, 1.66] & [-0.49, 0.08] & [-0.69, 2.04] & [-0.24, 0.22] &  & [-0.01, 0.19] & [-0.01, 0.01] & [0.02, 0.03] & [0.01, 0.03] \\[0.3ex]
\quad $\theta_N$ & 0.161 & +0.012 & -0.001 & +0.036 & +0.025 & 0.407 & +0.127 & +0.069 & +0.131$^{\dagger\dagger\dagger}$ & +0.143$^{\dagger\dagger\dagger}$ \\*
\quad  &  & [-0.05, 0.22] & [-0.04, 0.05] & [-0.05, 0.29] & [-0.01, 0.06] &  & [-0.29, 3.94] & [-0.25, 0.58] & [0.08, 0.21] & [0.10, 0.18] \\[0.3ex]
\quad $M_I$ & 0.392 & +0.025 & -0.049 & +0.356$^{\dagger}$ & +0.316$^{\dagger\dagger\dagger}$ & 0.005 & +0.001 & +0.003 & +0.044$^{\dagger\dagger\dagger}$ & +0.043$^{\dagger\dagger\dagger}$ \\*
\quad  &  & [-0.24, 0.77] & [-0.17, 0.15] & [0.00, 1.38] & [0.19, 0.50] &  & [0.00, 0.24] & [0.00, 0.03] & [0.03, 0.05] & [0.03, 0.05] \\[0.3ex]
\quad $M_C$ & 2.397 & +0.099 & -0.112 & -0.131 & -0.273 & 0.032 & +0.006 & +0.007 & +0.037$^{\dagger\dagger\dagger}$ & +0.039$^{\dagger\dagger\dagger}$ \\*
\quad  &  & [-0.89, 3.89] & [-0.49, 0.33] & [-1.08, 2.09] & [-0.47, 0.15] &  & [-0.02, 0.32] & [-0.02, 0.04] & [0.02, 0.05] & [0.03, 0.05] \\[0.3ex]
\quad $M_N$ & 0.893 & -0.009 & -0.022 & +0.134 & +0.071 & 0.044 & +0.008 & +0.008 & +0.082$^{\dagger\dagger\dagger}$ & +0.086$^{\dagger\dagger\dagger}$ \\*
\quad  &  & [-0.35, 1.10] & [-0.24, 0.24] & [-0.32, 1.59] & [-0.11, 0.34] &  & [-0.03, 0.44] & [-0.02, 0.04] & [0.07, 0.11] & [0.08, 0.12] \\[0.3ex]
\end{longtable}
\begingroup\footnotesize\setlength{\parindent}{0pt}\setlength{\parskip}{0.4em}%
\setlength{\leftskip}{0.025\textwidth}\setlength{\rightskip}{0.025\textwidth}%
\textit{Notes:} M1 models the discrete proxies as ordered probits (correctly specified); M2 treats all proxies as continuous (misspecified). Cells give the median over 20 MC replications of the paired difference (estimated $-$ true) in the total average marginal return (AMR), with its 90\% interval (5th and 95th percentiles of the paired differences across replications) in brackets; $^{\dagger}$, $^{\dagger\dagger}$, $^{\dagger\dagger\dagger}$ indicate that the 90\%, 95\% and 99\% percentile interval of the paired differences excludes zero. True AMR is the total AMR at the data-generating parameters: the population-average elasticity of skill at ages 13--14 with respect to the input at ages 5--6 in effective units, through the four transitions (investment at ages 7--12 responds through the investment equations; covariates and factors at ages 5--6 held fixed for the purpose of investment at ages 5--6. Lagged investment held fixed when determining investment at ages 7--8).

Total AMRs are imprecise: converting between the effective units of the ages 5--6 and ages 11--12 technologies depends on the curvature parameters ($\rho$), which are weakly identified at this sample size, and replications with $\rho$ near zero give extreme values (hence the wide intervals). In the $\bar\rho_0$ columns the ages 5--6 curvature parameter $\rho_{i,4}$ in the conversion to effective units is held at its median across that model's replications (each estimate times $|\hat\rho_{i,4}|/|\bar\rho_{i,4}|$); $\rho_{i,4}$ is common to all inputs of skill $i$, so this removes the imprecision a column shares.

{\scriptsize\setlength{\rightskip}{0.025\textwidth plus 1fil}%
\textit{AI disclosure:} table produced by code written with Claude Sonnet 5.5 (claude-sonnet-5-5) and Claude Opus 5.5 (claude-opus-5-5), Anthropic, 2026-10-08 (revised 2026-10-01 to 2026-10-07; layout adapted from MLE/estimator/mc\_analysis.R; mixed model lineage, see git history); MLE/estimator/intertemporal/mc/it\_mc\_analysis.R::it\_mc\_write\_total\_amr\_table().
\par}
\par\endgroup
\endgroup
\end{landscape}

To address the instability caused by this one parameter, I also present AMRs and 90\% confidence intervals in which I hold $\rho_{i,0}$ at its median value across replications for output skill $i$. These columns, denoted with at $\bar{\rho}_0$, show that fixing $\rho_{i,0}$ across replications removes most of the uncertainty in the estimator. Since $\rho_{i,0}$ is shared in all AMR calculations involving skill $i$, AMR levels are comparable across all periods and across all input factors, though not across output factors. Within output skills, the relative AMR levels are well defined, even though the absolute levels are not. 
\begin{figure}
    \centering
    \includegraphics[width=\linewidth]{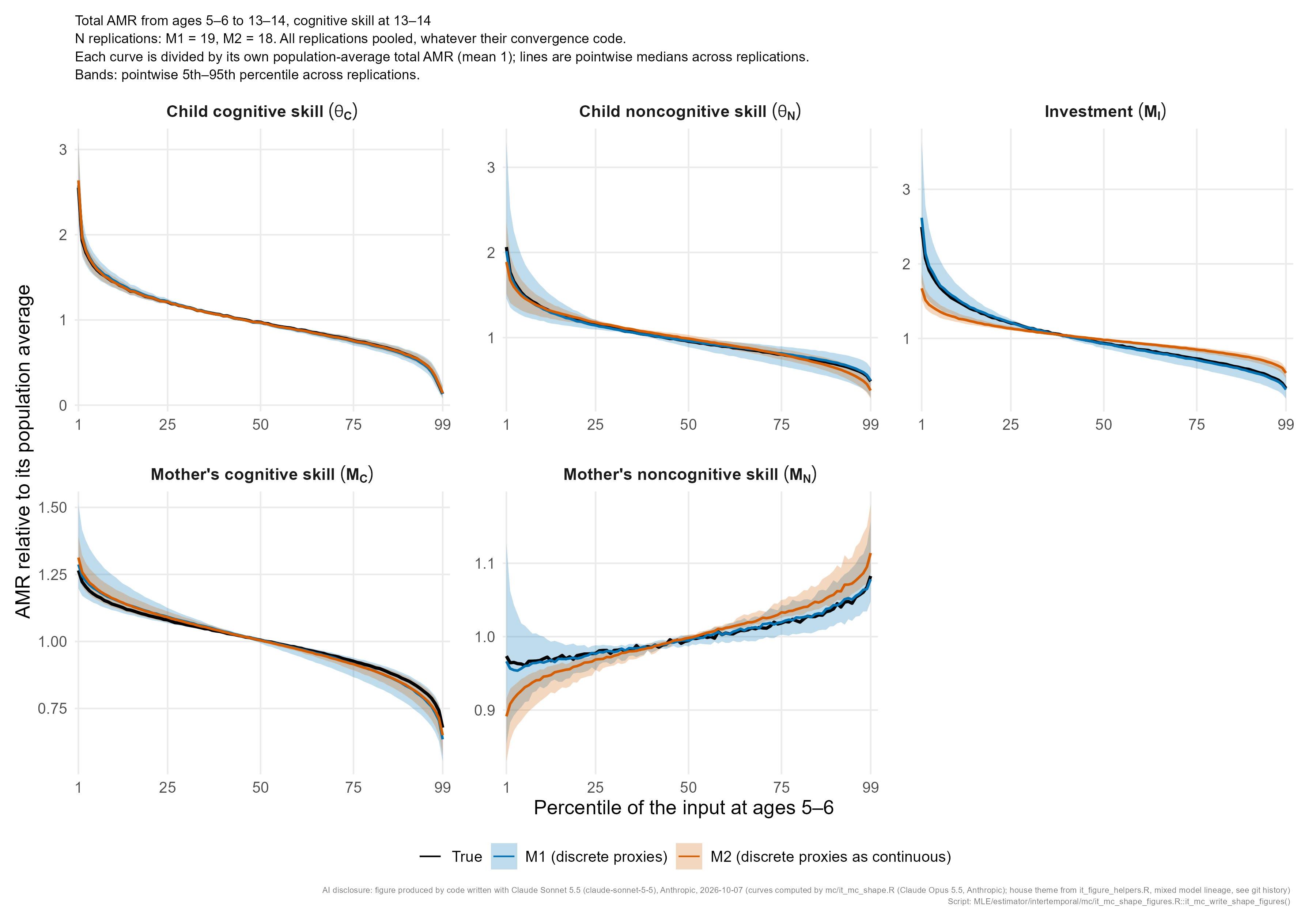}
    \caption{MC Elasticity Shape Recovery, Cognitive, Ages 5-6 to 13-14}
    \label{fig:mc_shape_cog}
\end{figure}
Figures \ref{fig:mc_shape_cog} and \ref{fig:mc_shape_noncog} show that the discrete model recovers the shape of AMR curves remarkably well (though the continuous model performs quite poorly when either the output's or the examined input's proxies are actually discrete). The figures plot the long run median AMR for all factors by age 5-6 skill percentile, cognitive skill for Figure \ref{fig:mc_shape_cog} and noncognitive skill for \ref{fig:mc_shape_noncog}. All MC runs had their AMR curves divided by the population average total AMR, so that all have the same overall mean, one. 
\begin{figure}
    \centering
    \includegraphics[width=\linewidth]{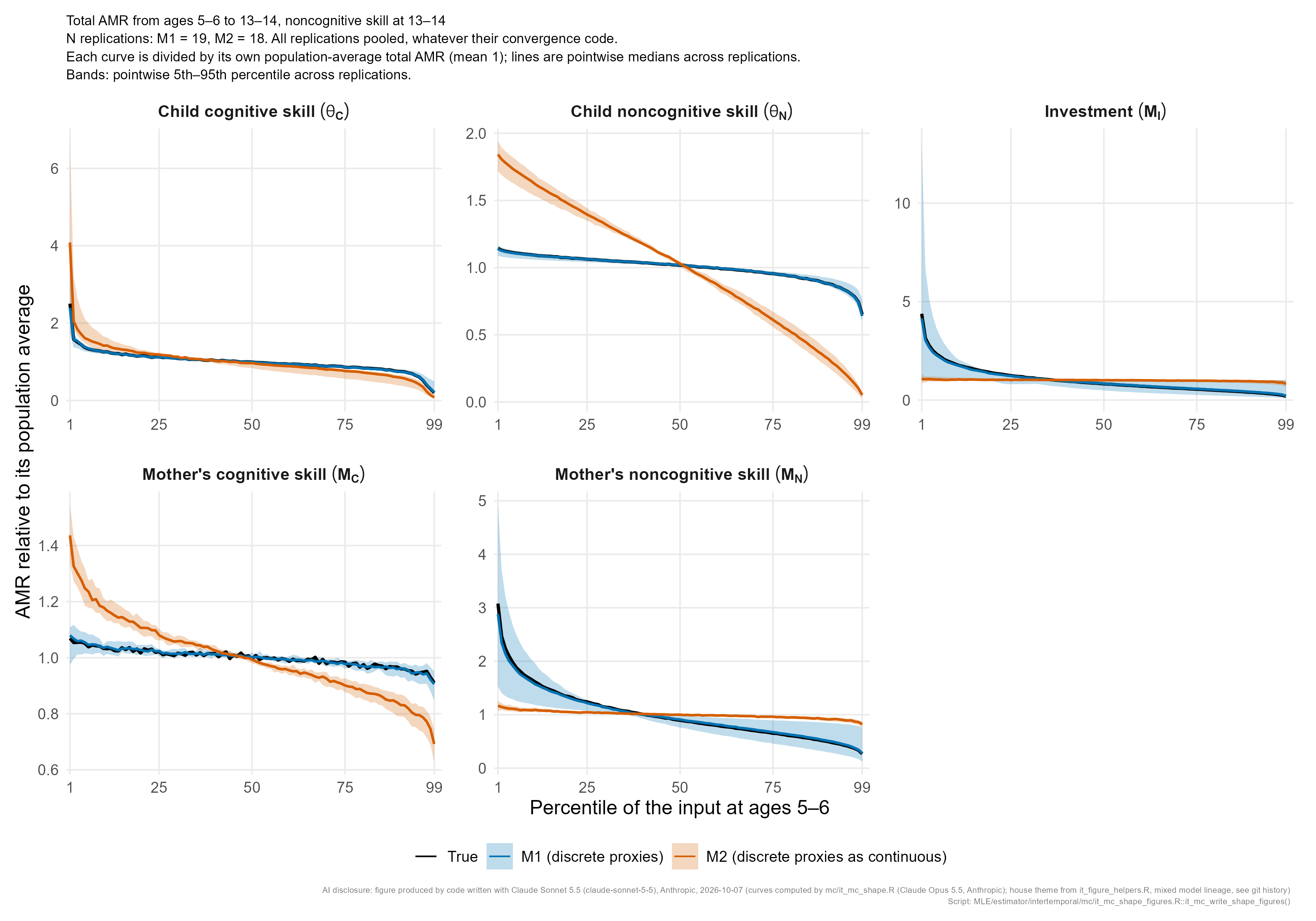}
    \caption{MC Elasticity Shape Recovery, Noncognitive, Ages 5-6 to 13-14}
    \label{fig:mc_shape_noncog}
\end{figure}
Appendix Table \ref{tab:mc_it_investment} shows that the investment equation coefficients are recovered without bias, if somewhat noisily, for the discrete model, while the continuous model has statistically significant biases on lagged investment and the standard deviation of the investment residual. The dominant role of lagged investment in the investment equations makes all other coefficients small and difficult to detect by comparison. 

The Monte Carlo simulations show that my preferred discrete proxy specification (M1) accurately recovers heterogeneity in AMR curves, short run AMR levels, and relative AMR levels even over long time horizons. It also shows that long run AMR levels are unbiased, but very noisy under the real data constraints that I face. However, given the role the subistution parameter plays in that noise, relative levels of AMRs are likely comparable across both factors and periods.

The Monte Carlo simulations are much less kind to the all continuous specification. A continuous proxy specification causes the estimator to overstate the effect of parental investment in every period, and causes large and statistically significant biases across all inputs in both the short and long run for noncognitive outcomes. It also cannot accurately capture AMR heterogeneity. These results demonstrate the importance of treating binary and Likert proxies as discrete. Treating them as continuous is a first order source of bias.

\section{Empirical Results}\label{sec:results}

In this section, I calculate AMRs from the model in Section \ref{sec:model_skill_form} to answer three questions. In Section \ref{sec:results:investment} I review the coefficients from the model's investment equations and evaluate the validity of using current household income as an exogenous investment shifter. Section \ref{sec:results:shortrun} presents one period, short run, AMRs, Section \ref{sec:results:longrun} presents the AMRs of factor inputs at ages 5-6 on skills at different outcome ages, allowing me to track the dynamic properities of early investment. Section \ref{sec:results:invage} presnts the AMR on skills at ages 13-14 to different input periods respectively. I use it to show that early childhood is unique in both the long run level of cognitive payoff to investment and that it is the only period in which marginal investments benefit low-skilled children more than high-skilled children. 

Appendix Tables \ref{tab:controls_all_it}, \ref{tab:ces_params_all_it}, \ref{tab:mixture_it}, and \ref{tab:covariate_transitions_it} contain the estimated parameters for the controls in the skill production Equation (\ref{eq:CES_prod}), the raw parameters inside the CES block of Equation (\ref{eq:CES_prod}), the mixture estimates, and the coefficients on the transition equations.

\subsection{Investment}\label{sec:results:investment}

Table \ref{tab:investment_all_it} shows the dominant roll lagged investment plays in predicting investment this period. Continuous control variables are standardized to have mean 0 and variance 1. A panel shows standardized factor coefficients  Very few coefficients other than lagged investment are statistically significant, and no coefficient other than lagged investment is consistently signed and statistically significant. Nonetheless, there are several features of the table worth discussing. The first is that investment becomes easier to predict as children age. The standard deviation of the investment shock falls throughout childhood, shrinking by nearly a factor of four between ages 5-6 and 13-14. This is likely the result of the variance of the investment factor falling over time (see Table \ref{tab:signal_table_it}). The shrinking coefficient on lagged log investment should be interpreted with caution for the same reason. 

\begin{table}[htbp]
\centering
\caption{Investment Equation Estimates by Period}
\label{tab:investment_all_it}
\begin{tabular}{lccccc}
\toprule
 & 5--6 & 7--8 & 9--10 & 11--12 & 13--14 \\
\midrule
$\beta_0$ & 0.991$^{***}$ & 0.702$^{***}$ & 0.029 & -0.370$^{***}$ & -0.174 \\
 & (0.184) & (0.152) & (0.153) & (0.129) & (0.125) \\
\midrule
\multicolumn{6}{l}{\textit{Factor loadings (per standard deviation of the factor)}} \\[0.5ex]
\quad $\theta_C$ & -0.033 & -0.093$^{**}$ & 0.078$^{**}$ & 0.010 & -0.018 \\
 & (0.040) & (0.037) & (0.032) & (0.033) & (0.034) \\
\quad $\theta_N$ & 0.072 & 0.060$^{*}$ & -0.009 & 0.012 & 0.055$^{*}$ \\
 & (0.044) & (0.034) & (0.029) & (0.027) & (0.029) \\
\quad $M_C$ & 0.042 & 0.034 & -0.027 & 0.042 & 0.025 \\
 & (0.068) & (0.050) & (0.043) & (0.042) & (0.042) \\
\quad $M_N$ & 0.075 & 0.003 & 0.056$^{*}$ & -0.018 & -0.010 \\
 & (0.047) & (0.036) & (0.031) & (0.030) & (0.030) \\
\quad $I_{t-1}$ & 1.415$^{***}$ & 1.124$^{***}$ & 0.978$^{***}$ & 0.906$^{***}$ & 0.870$^{***}$ \\
 & (0.179) & (0.131) & (0.102) & (0.078) & (0.076) \\
\midrule
\multicolumn{6}{l}{\textit{Control variables ($\beta_X$)}} \\[0.5ex]
\quad Male & -0.069 & 0.007 & -0.069 & -0.090$^{*}$ & -0.044 \\
 & (0.076) & (0.057) & (0.050) & (0.048) & (0.048) \\
\quad Mother's age & -0.062 & -0.127$^{***}$ & -0.041 & 0.011 & -0.090$^{***}$ \\
 & (0.054) & (0.038) & (0.030) & (0.029) & (0.029) \\
\quad Mother's educ. & -0.007 & 0.095$^{**}$ & -0.008 & 0.004 & 0.047 \\
 & (0.052) & (0.041) & (0.034) & (0.033) & (0.033) \\
\quad Ch.\ in HH & -0.138$^{*}$ & 0.010 & -0.044 & 0.005 & 0.069$^{*}$ \\
 & (0.070) & (0.049) & (0.040) & (0.037) & (0.040) \\
\quad Mom married & -0.165 & -0.082 & -0.326$^{***}$ & 0.164 & 0.221$^{*}$ \\
 & (0.172) & (0.130) & (0.118) & (0.108) & (0.115) \\
\quad Children $\times$ Mar. & 0.026 & 0.001 & 0.105$^{**}$ & -0.002 & -0.074 \\
 & (0.078) & (0.058) & (0.049) & (0.045) & (0.048) \\
\quad Log perm.\ inc. & 0.047 & 0.002 & 0.017 & 0.036 & -0.065$^{*}$ \\
 & (0.058) & (0.043) & (0.038) & (0.037) & (0.037) \\
\quad Log curr.\ inc. & 0.036 & -0.056$^{**}$ & 0.059$^{**}$ & -0.042 & 0.019 \\
 & (0.042) & (0.028) & (0.024) & (0.026) & (0.025) \\
\midrule
$\sigma_{\eta}$ & 0.559$^{***}$ & 0.531$^{***}$ & 0.356$^{***}$ & 0.309$^{***}$ & 0.148 \\
 & (0.070) & (0.049) & (0.040) & (0.041) & (0.110) \\
\midrule
$N$ & \multicolumn{5}{c}{2,402} \\
\bottomrule
\end{tabular}
\begin{minipage}{0.95\textwidth}
\vspace{0.5em}\footnotesize
\textit{Notes:} Column headings are the child's ages. Factor loadings are multiplied by the standard deviation of the log factor at that age in the simulated cohort, and their standard errors by the same standard deviation, ignoring sampling uncertainty in the standard deviation.
Standard errors in parentheses (delta method; provisional pending the bootstrap); n.d.: not defined; $^{***}p<0.01$, $^{**}p<0.05$, $^{*}p<0.10$.
\textit{AI disclosure:} table produced by code written with Claude Opus 5 (claude-opus-5) and Claude Sonnet 5 (claude-sonnet-5) and Claude Opus 5.5 (claude-opus-5-5), Anthropic, 2026-10-08 (CES-share kernel reused from MLE/estimator/plot\_elasticities.R, mixed model lineage; see git history); MLE/estimator/intertemporal/make\_it\_paper\_outputs.R::write\_it\_investment().
\end{minipage}
\end{table}

The other feature of the table worth highlighting is the log current household income coefficients. These are significant but opposite signed at ages 7-8 and 9-10 and are insignificant and inconsistently signed in other periods. This suggests that income is a poor exogenous shifter after controlling for lagged investment. This indicates that researchers should exercise caution when using current income as an exogenous shifter. In practice, the investment lag is by far the strongest exogenous shifter for current investment.

\subsection{Short Run AMRs}\label{sec:results:shortrun}

Table \ref{tab:tamr_shortrun} presents the short run AMRs skills and investment period-by-period. Short run cognitive skill persistance increases substantially as children age, and the short run returns to parental investment fall by around a factor of five. Parental skills have a higher impact on cognitive skill at ages 7-8, but decline in impact after that. Noncognitive skill seems to play a small role in fostering cognitive skill. 

\begin{table}[htbp]
\centering
\caption{Short-Run Total AMRs for All Factors}
\label{tab:tamr_shortrun}
\begin{tabular}{lcccc}
\toprule
 & 5--6 $\to$ 7--8 & 7--8 $\to$ 9--10 & 9--10 $\to$ 11--12 & 11--12 $\to$ 13--14 \\
\midrule
\multicolumn{5}{l}{\textit{Panel A: Cognitive ($\theta_{C,t+1}$)}} \\[0.5ex]
\quad $\theta_C$ & 0.495 & 0.452 & 0.711 & 0.851 \\
\quad $\theta_N$ & 0.033 & 0.030 & 0.010 & 0.024 \\
\quad $M_I$ & 0.108 & 0.055 & 0.020 & 0.021 \\
\quad $M_C$ & 0.301 & 0.370 & 0.208 & 0.054 \\
\quad $M_N$ & 0.063 & 0.093 & 0.052 & 0.049 \\
\quad Total parental & 0.472 & 0.518 & 0.280 & 0.124 \\
\midrule
\multicolumn{5}{l}{\textit{Panel B: Noncognitive ($\theta_{N,t+1}$)}} \\[0.5ex]
\quad $\theta_C$ & 0.009 & 0.012 & 0.008 & 0.019 \\
\quad $\theta_N$ & 0.949 & 0.960 & 0.967 & 0.900 \\
\quad $M_I$ & 0.009 & 0.004 & 0.003 & 0.021 \\
\quad $M_C$ & 0.005 & 0.007 & 0.006 & 0.016 \\
\quad $M_N$ & 0.028 & 0.017 & 0.016 & 0.045 \\
\quad Total parental & 0.043 & 0.028 & 0.025 & 0.081 \\
\bottomrule
\end{tabular}
\begin{minipage}{0.95\textwidth}
\vspace{0.5em}\footnotesize
\textit{Notes:} Total AMRs of each input (rows) on next-period cognitive (Panel A) and noncognitive (Panel B) skill; column headings give the child's ages in the input and output periods. Total parental $= M_I + M_C + M_N$.
\textit{AI disclosure:} table produced by code written with Claude Opus 5 (claude-opus-5) and Claude Sonnet 5 (claude-sonnet-5) and Claude Opus 5.5 (claude-opus-5-5), Anthropic, 2026-10-08 (CES-share kernel reused from MLE/estimator/plot\_elasticities.R, mixed model lineage; see git history); MLE/estimator/intertemporal/make\_it\_paper\_outputs.R::write\_it\_avg\_elasticities().
\end{minipage}
\end{table}

Noncognitive skill persistance is higher than cognitive skill persistance across childhood, and increases until adolescence, when the reutrns to parental skills and investment increase substantially. This is consistent with the idea that adolescence is a key period for forming a child's socio-emotional skills. Cognitive skill seems to play a consistently minor role in fostering cognitive skills.

\begin{figure}
    \centering
    \includegraphics[width=\linewidth]{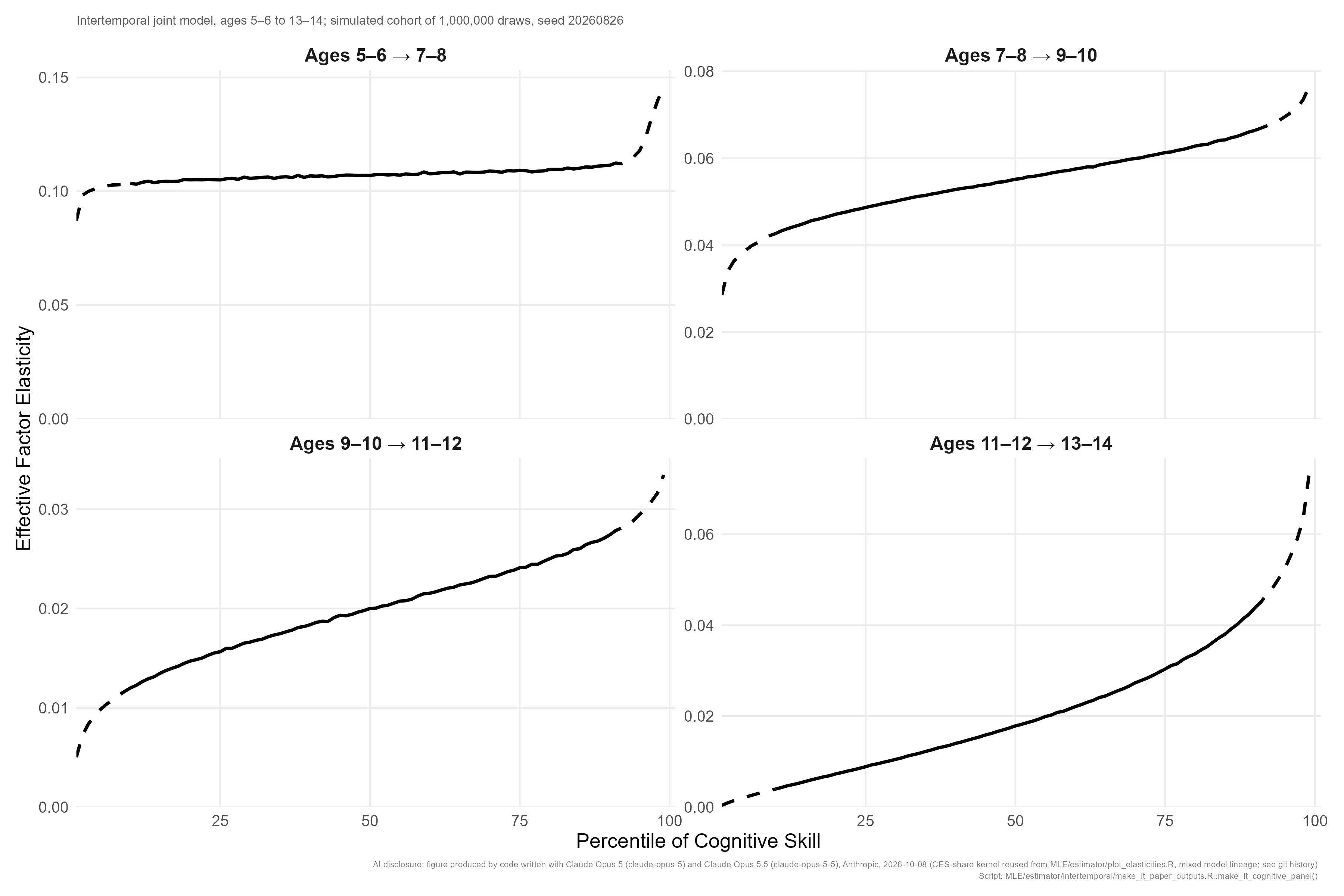}
    \caption{Cognitive AMR to Investment in the Short Run}
    \label{fig:amr_cog_short}
\end{figure}

Figure \ref{fig:amr_cog_short} shows the heterogeneity in short run AMR to investment on cognitive skill by input period cognitive skill percenile at different ages. The tails of the lines are dashed since the Monte Carlo simulations suggest that the estimator has trouble consistently estimating heterogeneity in the tails (see Figures \ref{fig:mc_shape_cog} and \ref{fig:mc_shape_noncog}). The short run AMR to investment at ages 5-6 is roughly flat, rough increasing in the tails. This suggests that low- and high-skilled children benefit roughly equally from investment in the short run in early childhood. In other periods, the short run AMRs to investment have a clear upward slope, indicating that the short run returns to investment are higher for higher-skilled children outside of early childhood. The noncognitive counterpart to Figure \ref{fig:amr_cog_short} can be found in Appendix Figure \ref{fig:amr_noncog_short}. 

It is worth noting that the cognitive substitution parameters in Appendix Table \ref{tab:ces_params_all_it} are all less than 0 and strongly statistically significant, ranging from $-0.526$ at ages 5-6 to $-1.677$ at ages 11-12. It is also worth noting that the ages 7-8 AMRs appear to be much better behaved than the production function coefficients at ages 7-8 in Table \ref{tab:ces_params_all_it} might otherwise suggest.

\subsection{Long Run AMRs}\label{sec:results:longrun}

Table \ref{tab:tamr_by_outcome_age_it} shows the AMRs of factor inputs at ages 5-6 on child skills at ages 7-8, 9-10, 11-12, and 13-14. For cognitive skill, the returns to all early childhood factors grow over time. This is driven by the effective factor unit conversion for output factor $i$ and input factor $j$ from period $0$ (ages 5-6) to period $t$: $\frac{|\rho_{i,t-1}|\phi_{i,j,0}}{|\rho_{i,0}|\psi_{i,t-1}}$ along with changes to the log factor variances in Table \ref{tab:signal_table_it}. Separating the effects of changes in input/output factor scale, returns to scale of the production function, and changes in the complementarity parameter are impossible, even if \cite{freyberger_2025} shows that the elasticities are well defined and invariant to affine normalizations of the (log) measurement system. So, this increase in input elasticities over time is due to a combination of increasing complementarity, decreasing returns to scale in later periods, and increasing variance in log skills over time. 

\begin{table}[htbp]
\centering
\caption{Total AMRs of Inputs at Ages 5--6, by Outcome Age}
\label{tab:tamr_by_outcome_age_it}
\begin{tabular}{lcccc}
\toprule
 & \multicolumn{4}{c}{Outcome age (inputs at ages 5--6)} \\
\cmidrule(lr){2-5}
 & \makebox[3.6em]{7--8} & \makebox[3.6em]{9--10} & \makebox[3.6em]{11--12} & \makebox[3.6em]{13--14} \\
\midrule
\multicolumn{5}{l}{\textit{Panel A: Cognitive skill ($\theta_{C,h}$)}} \\[0.5ex]
\quad $\theta_C$ & 0.495 & 0.422 & 0.791 & 1.838 \\
\quad $\theta_N$ & 0.033 & 0.033 & 0.063 & 0.161 \\
\quad $M_I$ & 0.108 & 0.091 & 0.169 & 0.392 \\
\quad $M_C$ & 0.301 & 0.365 & 0.963 & 2.397 \\
\quad $M_N$ & 0.063 & 0.106 & 0.297 & 0.893 \\
\quad Total parental & 0.472 & 0.562 & 1.429 & 3.682 \\
\midrule
\multicolumn{5}{l}{\textit{Panel B: Noncognitive skill ($\theta_{N,h}$)}} \\[0.5ex]
\quad $\theta_C$ & 0.009 & 0.030 & 0.049 & 0.015 \\
\quad $\theta_N$ & 0.949 & 1.383 & 1.870 & 0.407 \\
\quad $M_I$ & 0.009 & 0.015 & 0.020 & 0.005 \\
\quad $M_C$ & 0.005 & 0.031 & 0.063 & 0.032 \\
\quad $M_N$ & 0.028 & 0.049 & 0.075 & 0.044 \\
\quad Total parental & 0.043 & 0.095 & 0.158 & 0.080 \\
\bottomrule
\end{tabular}
\begin{minipage}{0.95\textwidth}
\vspace{0.5em}\footnotesize
\textit{Notes:} Total AMRs of each input at ages 5--6 (rows) on cognitive (Panel A) and noncognitive (Panel B) skill at the outcome age $h$ in the column heading. Each input is in the effective units of the panel skill's production function at ages 5--6; each outcome in those of the production function that produced it. An increase in investment at ages 5--6 does not enter the next period's investment through lagged investment. Total parental $= M_I + M_C + M_N$.
\textit{AI disclosure:} table produced by code written with Claude Opus 5 (claude-opus-5) and Claude Sonnet 5 (claude-sonnet-5) and Claude Opus 5.5 (claude-opus-5-5), Anthropic, 2026-10-08 (values from MLE/estimator/intertemporal/it\_propagation.R::it\_prop\_tamr\_tables() (Claude Opus 5.5); shared formatting helpers by Claude Opus 5 and Claude Sonnet 5; see git history); MLE/estimator/intertemporal/make\_it\_paper\_outputs.R::write\_it\_tamr\_by\_outcome\_age().
\end{minipage}
\end{table}

The noncognitive skill long0run AMR pattern is more interesting. Like cognitive skill, AMRs to early childhood factor inputs increase from ages 7-12. However, they sharply drop in adolescence. This could be because inputs are much more substitutable in adolescence than in other periods, which would be consistent with the adolescence being a critical time for children's socio-emotional development. 

\begin{figure}
    \centering
    \includegraphics[width=\linewidth]{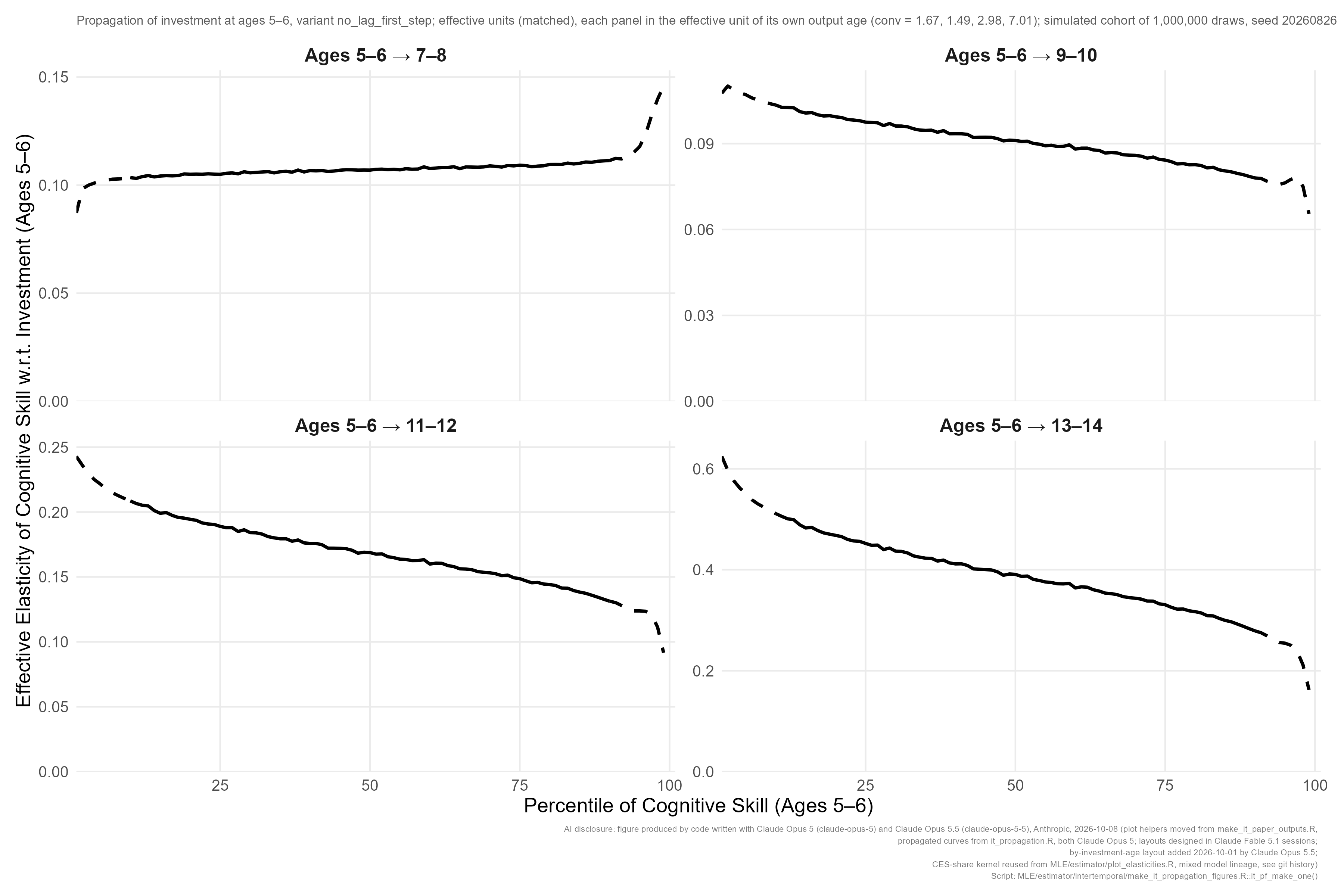}
    \caption{AMR to Investment at Ages 5-6 on Cognitive Skill Across Childhood}
    \label{fig:amr_cog_long}
\end{figure}

Perhaps most interesting is the heterogeneity in AMR to parental investment at age 5-6 on cognitive skill presented in Figure \ref{fig:amr_cog_long}. Consitent with Figure \ref{fig:amr_cog_short}, the short-run AMR is mostly flat, sloping upwards at the tails if they can be trusted. At time goes on, that flat slop becomes negative. The AMR increases substatially for low-skilled children (by a factor of about 5.1 at the 10th percentile), and increases much less for high-skilled children (by a factor of about 2.8 at the 90th percentile). If the tails can be believed, the lowest skilled children see their AMR rise by a factor of 7.2 and the highest skill children essentially see no increase in AMR at all.

This is a remarkable finding. It shows that early childhood investment snowballs for low-skilled children much more strongly than for high-skilled children. This is likely in part due to Dynamic Complementarity (which is present in the ratio of the $\rho$'s in the conversion factor that drives the level of the slope, $\frac{|\rho_{i,t-1}|\phi_{i,j,0}}{|\rho_{i,0}|\psi_{i,t-1}}$), but is also in part due to the joint distribution of inputs in early childhood, in which skill and investment are highly correlated, which causes the short run AMR to flatten substantially.

\subsection{AMRs of Inputs at Different Ages by Early High School}\label{sec:results:invage}

In the last section, I presented results where I held the input period fixed and varied the output period. In this section, I do the opposite: I vary in input period and hold the output period fixed. This allows me to trace at what ages investment is most effective for different children.

Table \ref{tab:tamr_by_input_age_it} shows the Total AMRs to factor inputs at different ages. For the cognitive skill outcomes, later childhood and adolescence have lower AMRs to investment than early childhood, consistent with increasing complementarity over time and investment being a relatively plentiful input. For some input factors, the AMR is higher at ages 7-8 than at 5-6, however, the crucial investment AMR is highest at ages 5-6, where its value is almost twice that of the AMR at ages 7-8. 

\begin{table}[htbp]
\centering
\caption{Total AMRs on Outcomes at Ages 13--14, by Age of Input}
\label{tab:tamr_by_input_age_it}
\begin{tabular}{lcccc}
\toprule
 & \multicolumn{4}{c}{Age of input (outcomes at ages 13--14)} \\
\cmidrule(lr){2-5}
 & \makebox[3.6em]{5--6} & \makebox[3.6em]{7--8} & \makebox[3.6em]{9--10} & \makebox[3.6em]{11--12} \\
\midrule
\multicolumn{5}{l}{\textit{Panel A: Cognitive skill at ages 13--14 ($\theta_{C,4}$)}} \\[0.5ex]
\quad $\theta_C$ & 1.838 & 1.952 & 1.610 & 0.851 \\
\quad $\theta_N$ & 0.161 & 0.245 & 0.123 & 0.024 \\
\quad $M_I$ & 0.392 & 0.217 & 0.043 & 0.021 \\
\quad $M_C$ & 2.397 & 4.316 & 0.614 & 0.054 \\
\quad $M_N$ & 0.893 & 1.176 & 0.233 & 0.049 \\
\quad Total parental & 3.682 & 5.709 & 0.890 & 0.124 \\
\midrule
\multicolumn{5}{l}{\textit{Panel B: Noncognitive skill at ages 13--14 ($\theta_{N,4}$)}} \\[0.5ex]
\quad $\theta_C$ & 0.015 & 0.009 & 0.006 & 0.019 \\
\quad $\theta_N$ & 0.407 & 0.282 & 0.210 & 0.900 \\
\quad $M_I$ & 0.005 & 0.001 & 0.001 & 0.021 \\
\quad $M_C$ & 0.032 & 0.011 & 0.007 & 0.016 \\
\quad $M_N$ & 0.044 & 0.060 & 0.051 & 0.045 \\
\quad Total parental & 0.080 & 0.072 & 0.059 & 0.081 \\
\bottomrule
\end{tabular}
\begin{minipage}{0.95\textwidth}
\vspace{0.5em}\footnotesize
\textit{Notes:} Total AMRs of each input (rows), at the age in the column heading, on cognitive (Panel A) and noncognitive (Panel B) skill at ages 13--14. Each input is in the effective units of the panel skill's production function at the input age; skill at ages 13--14 in those of the ages 11--12 $\to$ 13--14 production function. $M_C$ and $M_N$ are shifted from the input age onward; an increase in investment at the input age does not enter the next period's investment through lagged investment. Total parental $= M_I + M_C + M_N$.
\textit{AI disclosure:} table produced by code written with Claude Opus 5 (claude-opus-5) and Claude Sonnet 5 (claude-sonnet-5) and Claude Opus 5.5 (claude-opus-5-5), Anthropic, 2026-10-08 (values from MLE/estimator/intertemporal/it\_propagation.R::it\_prop\_tamr\_tables() (Claude Opus 5.5); shared formatting helpers by Claude Opus 5 and Claude Sonnet 5; see git history); MLE/estimator/intertemporal/make\_it\_paper\_outputs.R::write\_it\_tamr\_by\_input\_age().
\end{minipage}
\end{table}

For noncognitive skill, the AMRs follow a U shaped patter. The returns to inputs are highest in early childhood (ages 5-6) and adolescence (ages 11-12). This once more emphasizes the important roles that early childhood and adolescence play in forming children's socio-emotional skills.

\begin{figure}
    \centering
    \includegraphics[width=\linewidth]{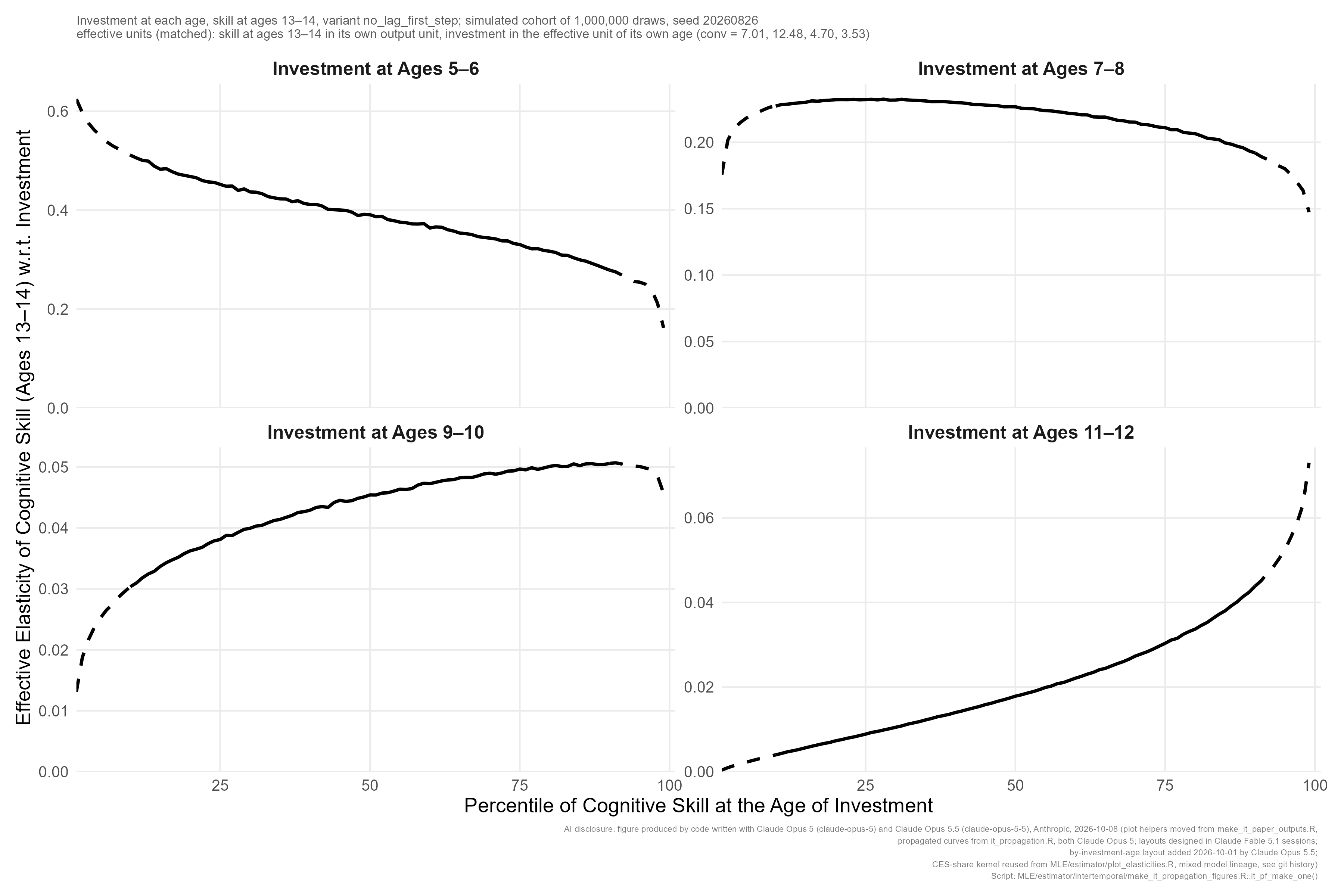}
    \caption{AMR to Investment at Different Ages on Cognitive Skill at Ages 13-14}
    \label{fig:amr_cog_invage}
\end{figure}

Figure \ref{fig:amr_cog_invage} shows that early childhood is unique not just because the levels of AMR to investment are so high, but because it is the only age at which low skill children benefit more from investment in the long run than high skill children do. AMR to investment at ages 7-8 talkes a mostly flat, hump shape, while AMRs at later ages are upward sloping. This highlights the important role that early childhood investment can play in reducing inequality. While all children of all cognitive skill levels benefit substantially from early childhood investment (other than maybe the most skilled children), low-skilled children benefit the most.

\section{Conclusion}\label{sec:conclusion}

In this paper, I introduce the average marginal return (AMR), a marginal analogue of the average treatment effects estimated in reduced-form program evaluations. By averaging marginal returns over the joint distribution of inputs, the AMR connects the properties of a production technology to the investment environments children experience. I use this distinction to clarify the importance of scale invariance in dynamic factor models, explain why interaction coefficients need not identify technological complementarity, and show why comparisons of later interventions across recipients and nonrecipients of earlier interventions may fail to detect dynamic complementarity.

I incorporate ordered probit measurement equations for discrete proxies into the dynamic factor framework of \cite{cunha_heckman_schennach2010}. Monte Carlo simulations calibrated to the empirical application show that treating discrete proxies as continuous can substantially bias both the levels and the shapes of AMR curves. I also introduce an effective factor framework that extends the scale-invariant approach of \cite{freyberger_2025} to settings in which multiple skills are produced each period.

The empirical application estimates AMRs to parental investment in cognitive skill across childhood. Although skill and investment are complements in every estimated cognitive production function, the AMR to investment at ages 5--6 is approximately flat across initial cognitive skill when the outcome is measured two years later and becomes sharply decreasing in initial skill at longer horizons. By ages 13--14, the estimated AMR is about 1.8 times as large at the 10th percentile of early cognitive skill as at the 90th. This pattern is consistent with dynamic complementarity operating alongside unequal investment across children.

The point estimates also suggest that the timing of investment matters. With cognitive skill at ages 13--14 as the common outcome, the estimated AMR to investment at ages 5--6 is nearly twice that of investment at ages 7--8. Early childhood also exhibits the most pronounced downward slope in long-run AMRs across cognitive skill levels. The profile for investment at ages 7--8 is approximately flat and mildly hump-shaped, while the profiles for later investment are generally upward sloping with cognitive skill at the age of investment.

Several limitations qualify these findings. Standard errors for the empirical AMRs are not yet available, and bootstrap inference is under development. The Monte Carlo analysis in Section \ref{sec:mc} provides evidence on finite-sample performance under the model and suggests that relative AMR profiles are recovered more reliably than long-run levels. Those levels depend on weakly estimated curvature parameters. 

These findings are consistent with and help explain a swathe of program analyses that find that early childhood investment in narrowing cognitive skill disparities. More broadly, they show why technological complementarity alone does not determine which children have the highest average marginal returns: that assessment also requires the joint distribution of inputs and the pathways through which early investment affects later development.

\bibliography{bibliography.bib}

\appendix

\section{Supplementary Tables and Figures}\label{section:app:sup_results}

\begin{landscape}
\begingroup\renewcommand{\baselinestretch}{1}%
\footnotesize\setlength{\tabcolsep}{1pt}
\begin{longtable}{@{}lccccccccccccccc@{}}
\caption{Investment Equation: Bias in the Standardized Coefficients, Median Paired Difference from True}
\label{tab:mc_it_investment} \\
\toprule
 & \multicolumn{3}{c}{Ages 5--6} & \multicolumn{3}{c}{Ages 7--8} & \multicolumn{3}{c}{Ages 9--10} & \multicolumn{3}{c}{Ages 11--12} & \multicolumn{3}{c}{Ages 13--14} \\
\cmidrule(lr){2-4} \cmidrule(lr){5-7} \cmidrule(lr){8-10} \cmidrule(lr){11-13} \cmidrule(lr){14-16}
Coefficient & True & M1 & M2 (Cont) & True & M1 & M2 (Cont) & True & M1 & M2 (Cont) & True & M1 & M2 (Cont) & True & M1 & M2 (Cont) \\
\midrule
\endfirsthead
\multicolumn{16}{l}{\textit{Table~\thetable{} (continued)}} \\
\toprule
 & \multicolumn{3}{c}{Ages 5--6} & \multicolumn{3}{c}{Ages 7--8} & \multicolumn{3}{c}{Ages 9--10} & \multicolumn{3}{c}{Ages 11--12} & \multicolumn{3}{c}{Ages 13--14} \\
\cmidrule(lr){2-4} \cmidrule(lr){5-7} \cmidrule(lr){8-10} \cmidrule(lr){11-13} \cmidrule(lr){14-16}
Coefficient & True & M1 & M2 (Cont) & True & M1 & M2 (Cont) & True & M1 & M2 (Cont) & True & M1 & M2 (Cont) & True & M1 & M2 (Cont) \\
\midrule
\endhead
\multicolumn{16}{r}{\textit{Continued on next page}} \\
\endfoot
\bottomrule
\endlastfoot
\quad Lagged investment & 0.902 & -0.009 & -0.042 & 0.926 & +0.023 & -0.078$^{\dagger\dagger\dagger}$ & 0.917 & +0.013 & +0.006 & 0.894 & +0.009 & +0.005 & 0.983 & +0.005 & +0.005 \\*
\quad  &  & (0.043) & (0.039) &  & (0.036) & (0.033) &  & (0.028) & (0.024) &  & (0.023) & (0.023) &  & (0.018) & (0.013) \\[0.3ex]
\multicolumn{16}{l}{\textit{Inputs}} \\*[0.3ex]
\quad $\theta_C$ & -0.021 & -0.009 & +0.008 & -0.077 & -0.013 & +0.003 & 0.073 & +0.002 & +0.004 & 0.010 & +0.001 & -0.001 & -0.020 & -0.001 & -0.006 \\*
\quad  &  & (0.040) & (0.030) &  & (0.035) & (0.027) &  & (0.024) & (0.027) &  & (0.028) & (0.026) &  & (0.032) & (0.030) \\[0.3ex]
\quad $\theta_N$ & 0.046 & -0.005 & +0.005 & 0.049 & -0.001 & +0.011 & -0.009 & +0.004 & +0.005 & 0.012 & -0.001 & +0.006 & 0.063 & -0.009 & -0.005 \\*
\quad  &  & (0.027) & (0.024) &  & (0.026) & (0.020) &  & (0.018) & (0.015) &  & (0.021) & (0.021) &  & (0.021) & (0.020) \\[0.3ex]
\quad $M_C$ & 0.027 & +0.019 & +0.022 & 0.028 & -0.007 & +0.039 & -0.026 & +0.001 & +0.007 & 0.041 & -0.004 & -0.008 & 0.028 & +0.002 & +0.008 \\*
\quad  &  & (0.043) & (0.038) &  & (0.022) & (0.027) &  & (0.025) & (0.027) &  & (0.025) & (0.028) &  & (0.035) & (0.032) \\[0.3ex]
\quad $M_N$ & 0.048 & +0.001 & +0.006 & 0.003 & -0.017 & -0.003 & 0.052 & +0.001 & +0.005 & -0.018 & -0.008 & -0.008 & -0.011 & +0.006 & +0.006 \\*
\quad  &  & (0.024) & (0.020) &  & (0.028) & (0.025) &  & (0.024) & (0.023) &  & (0.031) & (0.033) &  & (0.034) & (0.035) \\[0.3ex]
\multicolumn{16}{l}{\textit{Controls}} \\*[0.3ex]
\quad Male & -0.044 & +0.004 & +0.000 & 0.006 & -0.001 & +0.004 & -0.065 & +0.005 & -0.000 & -0.089 & -0.002 & -0.004 & -0.050 & +0.015 & +0.014 \\*
\quad  &  & (0.042) & (0.043) &  & (0.035) & (0.039) &  & (0.040) & (0.039) &  & (0.039) & (0.038) &  & (0.027) & (0.032) \\[0.3ex]
\quad Ch.\ in HH & -0.088 & +0.010 & -0.011 & 0.008 & +0.015 & +0.025 & -0.042 & -0.021 & -0.026 & 0.005 & +0.009 & +0.008 & 0.078 & -0.007 & -0.003 \\*
\quad  &  & (0.042) & (0.047) &  & (0.039) & (0.042) &  & (0.037) & (0.033) &  & (0.039) & (0.038) &  & (0.038) & (0.039) \\[0.3ex]
\quad Mother's age & -0.040 & +0.001 & +0.020 & -0.105 & -0.001 & +0.027 & -0.039 & -0.014 & -0.019 & 0.011 & +0.009 & +0.008 & -0.102 & -0.013 & -0.016 \\*
\quad  &  & (0.033) & (0.029) &  & (0.038) & (0.032) &  & (0.023) & (0.026) &  & (0.024) & (0.025) &  & (0.030) & (0.030) \\[0.3ex]
\quad Mother's educ. & -0.005 & +0.003 & -0.001 & 0.078 & +0.004 & +0.000 & -0.008 & -0.003 & -0.003 & 0.004 & -0.008 & -0.004 & 0.053 & -0.010 & -0.010 \\*
\quad  &  & (0.030) & (0.024) &  & (0.021) & (0.022) &  & (0.026) & (0.028) &  & (0.013) & (0.015) &  & (0.019) & (0.020) \\[0.3ex]
\quad Mom married & -0.105 & +0.040 & +0.015 & -0.068 & +0.030 & +0.128 & -0.305 & -0.020 & -0.016 & 0.162 & -0.002 & -0.009 & 0.250 & -0.040 & -0.031 \\*
\quad  &  & (0.109) & (0.106) &  & (0.133) & (0.119) &  & (0.117) & (0.104) &  & (0.095) & (0.091) &  & (0.132) & (0.140) \\[0.3ex]
\quad Child $\times$ Married & 0.017 & -0.022 & +0.013 & 0.001 & -0.019 & -0.061 & 0.098 & +0.011 & +0.026 & -0.002 & +0.007 & +0.009 & -0.084 & -0.001 & -0.003 \\*
\quad  &  & (0.050) & (0.054) &  & (0.054) & (0.048) &  & (0.049) & (0.046) &  & (0.046) & (0.044) &  & (0.050) & (0.053) \\[0.3ex]
\quad Log perm.\ inc. & 0.030 & +0.003 & +0.012 & 0.002 & -0.012 & +0.016 & 0.016 & +0.006 & +0.001 & 0.036 & +0.008 & +0.004 & -0.073 & -0.013 & -0.014 \\*
\quad  &  & (0.019) & (0.022) &  & (0.027) & (0.025) &  & (0.033) & (0.029) &  & (0.031) & (0.030) &  & (0.038) & (0.041) \\[0.3ex]
\quad Log curr.\ inc. & 0.023 & -0.012 & -0.011 & -0.046 & -0.005 & -0.003 & 0.056 & -0.002 & +0.002 & -0.041 & -0.001 & -0.008 & 0.021 & +0.006 & +0.007 \\*
\quad  &  & (0.026) & (0.025) &  & (0.024) & (0.023) &  & (0.022) & (0.016) &  & (0.020) & (0.019) &  & (0.023) & (0.024) \\[0.3ex]
\midrule
\quad $\sigma_\eta$ & 0.356 & -0.014 & +0.018 & 0.438 & -0.020 & +0.059$^{\dagger\dagger\dagger}$ & 0.334 & -0.017 & -0.010 & 0.305 & -0.018 & -0.017 & 0.167 & -0.039 & -0.051 \\*
\quad  &  & (0.042) & (0.033) &  & (0.023) & (0.016) &  & (0.027) & (0.020) &  & (0.029) & (0.024) &  & (0.037) & (0.046) \\[0.3ex]
\end{longtable}
\begingroup\footnotesize\setlength{\parindent}{0pt}\setlength{\parskip}{0.4em}%
\setlength{\leftskip}{0.025\textwidth}\setlength{\rightskip}{0.025\textwidth}%
\textit{Notes:} Cells give the median over 20 MC replications of the paired difference (estimated $-$ true) in the standardized coefficient, with the standard deviation of the paired differences across replications (the estimator's standard error) in parentheses; daggers as in Table~\ref{tab:mc_it_total_amr}. True is the value at the data-generating parameters. Each coefficient is multiplied by the standard deviation of its regressor and divided by the standard deviation of log investment at the same ages ($\beta_k\,\mathrm{sd}(\ln f_{k,t})/\mathrm{sd}(\ln I_t)$; lagged investment uses $\mathrm{sd}(\ln I_{t-1})$, controls are only divided, and $\sigma_\eta$ is divided by $\mathrm{sd}(\ln I_t)$), with standard deviations from each fit's own simulated cohort. M2 measures investment, child noncognitive skill and maternal noncognitive skill on a different scale from M1, so raw coefficients are not comparable, whereas these ratios do not depend on the scale or location of any factor; the intercept is omitted for the same reason.

{\scriptsize\setlength{\rightskip}{0.025\textwidth plus 1fil}%
\textit{AI disclosure:} table produced by code written with Claude Sonnet 5.5 (claude-sonnet-5-5), Anthropic, 2026-10-08 (written 2026-10-07 by Claude Sonnet 5.5; statistic, dagger and longtable helpers shared with the total-AMR and per-transition writers, which have mixed model lineage; see git history); MLE/estimator/intertemporal/mc/it\_mc\_analysis.R::it\_mc\_write\_investment\_table().
\par}
\par\endgroup
\endgroup
\end{landscape}

\begingroup
\csname singlespacing\endcsname 
\centering
\begin{longtable}{lcccc}
\caption{CES Production Function Parameters by Period}
\label{tab:ces_params_all_it}\\
\toprule
 & 5--6 & 7--8 & 9--10 & 11--12 \\
\midrule
\endfirsthead
\multicolumn{5}{l}{\textit{(continued from previous page)}}\\
\toprule
 & 5--6 & 7--8 & 9--10 & 11--12 \\
\midrule
\endhead
\midrule
\multicolumn{5}{r}{\textit{continued on next page \ldots}}\\
\endfoot
\bottomrule
\endlastfoot
\multicolumn{5}{l}{\textit{Cognitive ($\theta_{C,t+1}$)}} \\*[0.5ex]
\quad $\gamma_{\theta_C}$ & 0.474$^{***}$ & 0.441$^{***}$ & 0.712$^{***}$ & 0.890 \\*
 & (0.108) & (0.097) & (0.070) & (n.d.) \\
\quad $\gamma_{\theta_N}$ & 0.021 & 0.028 & 0.008 & 0.011 \\*
 & (0.058) & (n.d.) & (0.042) & (n.d.) \\
\quad $\gamma_{M_I}$ & 0.157 & 0.070 & 0.030 & 0.023 \\*
 & (0.102) & (n.d.) & (0.058) & (n.d.) \\
\quad $\gamma_{M_C}$ & 0.288 & 0.368 & 0.201 & 0.041 \\*
 & (n.d.) & (n.d.) & (n.d.) & (n.d.) \\
\quad $\gamma_{M_N}$ & 0.060 & 0.092 & 0.049 & 0.034 \\*
 & (n.d.) & (n.d.) & (0.040) & (n.d.) \\
\quad $\rho$ & -0.526$^{***}$ & -0.668$^{***}$ & -0.817$^{**}$ & -1.677$^{***}$ \\*
 & (0.170) & (0.141) & (0.326) & (0.504) \\
\quad $\psi$ & 1.470$^{***}$ & 2.084$^{***}$ & 1.276$^{***}$ & 1.114$^{***}$ \\*
 & (0.373) & (0.455) & (0.125) & (0.033) \\
\quad $\phi_{\theta_N}$ & 1.995 & 14.657 & 15.076 & 6.956 \\*
 & (3.271) & (n.d.) & (68.771) & (n.d.) \\
\quad $\phi_{M_I}$ & 2.450$^{*}$ & 5.539 & 2.553 & 3.931 \\*
 & (1.302) & (n.d.) & (5.214) & (n.d.) \\
\quad $\phi_{M_C}$ & 2.635 & 10.917 & 2.852 & 1.895 \\*
 & (n.d.) & (n.d.) & (n.d.) & (n.d.) \\
\quad $\phi_{M_N}$ & 23.841 & 54.779 & 21.474 & 17.379 \\*
 & (n.d.) & (n.d.) & (30.428) & (n.d.) \\
\quad $\sigma_u$ & 0.481$^{***}$ & 0.245$^{***}$ & 0.186$^{***}$ & 0.021 \\*
 & (0.016) & (0.013) & (0.014) & (0.030) \\
\midrule
\multicolumn{5}{l}{\textit{Noncognitive ($\theta_{N,t+1}$)}} \\*[0.5ex]
\quad $\gamma_{\theta_C}$ & 0.013 & 0.022 & 0.015 & 0.027 \\*
 & (0.032) & (n.d.) & (0.046) & (0.240) \\
\quad $\gamma_{\theta_N}$ & 0.904 & 0.903 & 0.910 & 0.855$^{**}$ \\*
 & (n.d.) & (n.d.) & (n.d.) & (0.355) \\
\quad $\gamma_{M_I}$ & 0.038 & 0.026 & 0.030 & 0.037 \\*
 & (n.d.) & (n.d.) & (n.d.) & (0.046) \\
\quad $\gamma_{M_C}$ & 0.008 & 0.014 & 0.012 & 0.022 \\*
 & (0.027) & (0.040) & (0.031) & (n.d.) \\
\quad $\gamma_{M_N}$ & 0.036 & 0.036 & 0.033 & 0.059 \\*
 & (0.059) & (0.023) & (0.053) & (0.079) \\
\quad $\rho$ & -0.466 & -0.789 & -1.209 & -0.303 \\*
 & (0.360) & (n.d.) & (0.958) & (0.656) \\
\quad $\psi$ & 0.896$^{***}$ & 0.909$^{***}$ & 0.960$^{***}$ & 1.071$^{***}$ \\*
 & (0.061) & (0.046) & (0.055) & (0.372) \\
\quad $\phi_{\theta_C}$ & 1.204 & 0.997 & 1.917 & 2.013 \\*
 & (n.d.) & (n.d.) & (11.359) & (21.217) \\
\quad $\phi_{M_I}$ & 0.818 & 1.236 & 0.736 & 1.154 \\*
 & (0.528) & (2.461) & (n.d.) & (n.d.) \\
\quad $\phi_{M_C}$ & 2.720 & 2.217 & 2.633 & 2.451 \\*
 & (20.869) & (n.d.) & (n.d.) & (n.d.) \\
\quad $\phi_{M_N}$ & 0.645 & 1.996 & 2.868 & 0.668 \\*
 & (0.833) & (n.d.) & (6.790) & (n.d.) \\
\quad $\sigma_u$ & 0.718$^{***}$ & 0.660$^{***}$ & 0.732$^{***}$ & 0.901$^{***}$ \\*
 & (0.035) & (0.032) & (0.039) & (0.055) \\
\end{longtable}
\noindent
\begin{minipage}{0.95\textwidth}
\vspace{0.5em}\footnotesize
\textit{Notes:} CES production parameters for $\ln\theta_{i,t+1} = (\psi_i/\rho_i)\ln\left(\sum_j \gamma_{ij} y_j^{\rho_i/\phi_{ij}}\right) + \kappa_{0i} + Q_t'\kappa_i + \kappa_{i\eta}\eta_t + \nu_i$, with inputs $y = (\theta_{C,t}, \theta_{N,t}, I_t, M_C, M_N)$.
$\phi_{\text{own}} = 1$ by normalization (omitted); $\gamma$ sums to one. Column headings are the input period's child ages.
Standard errors in parentheses (delta method; provisional pending the bootstrap); n.d.: not defined; $^{***}p<0.01$, $^{**}p<0.05$, $^{*}p<0.10$.
\textit{AI disclosure:} table produced by code written with Claude Opus 5 (claude-opus-5) and Claude Sonnet 5 (claude-sonnet-5) and Claude Opus 5.5 (claude-opus-5-5), Anthropic, 2026-10-08 (CES-share kernel reused from MLE/estimator/plot\_elasticities.R, mixed model lineage; see git history); MLE/estimator/intertemporal/make\_it\_paper\_outputs.R::write\_it\_ces\_params().
\end{minipage}
\par
\endgroup

\begingroup
\csname singlespacing\endcsname 
\centering
\begin{longtable}{lcccc}
\caption{Production Control Coefficients by Period}
\label{tab:controls_all_it}\\
\toprule
 & 5--6 & 7--8 & 9--10 & 11--12 \\
\midrule
\endfirsthead
\multicolumn{5}{l}{\textit{(continued from previous page)}}\\
\toprule
 & 5--6 & 7--8 & 9--10 & 11--12 \\
\midrule
\endhead
\midrule
\multicolumn{5}{r}{\textit{continued on next page \ldots}}\\
\endfoot
\bottomrule
\endlastfoot
\multicolumn{5}{l}{\textit{Cognitive ($\theta_{C,t+1}$)}} \\*[0.5ex]
\quad $\kappa_0$ & -0.066 & 0.001 & -0.023 & -0.049 \\*
 & (0.085) & (0.056) & (0.051) & (0.050) \\
\quad Male & -0.031 & 0.058$^{***}$ & 0.035 & 0.032 \\*
 & (0.025) & (0.021) & (0.022) & (0.024) \\
\quad Mother's age & 0.004 & -0.005 & 0.016 & 0.004 \\*
 & (0.016) & (0.012) & (0.013) & (0.015) \\
\quad Mother's educ. & -0.024 & 0.007 & -0.029$^{**}$ & 0.010 \\*
 & (0.016) & (0.014) & (0.015) & (0.015) \\
\quad Ch.\ in HH & -0.023 & 0.002 & -0.002 & 0.016 \\*
 & (0.028) & (0.022) & (0.019) & (0.019) \\
\quad Mom married & -0.079 & 0.019 & 0.054 & 0.031 \\*
 & (0.065) & (0.053) & (0.053) & (0.057) \\
\quad Children $\times$ Mar. & 0.054 & -0.017 & 0.002 & 0.010 \\*
 & (0.033) & (0.026) & (0.023) & (0.024) \\
\quad $\kappa_{\eta}$ & -0.113 & -0.027 & 0.004 & 0.080 \\*
 & (0.088) & (0.039) & (0.060) & (0.073) \\
\midrule
\multicolumn{5}{l}{\textit{Noncognitive ($\theta_{N,t+1}$)}} \\*[0.5ex]
\quad $\kappa_0$ & -0.428$^{***}$ & -0.239$^{***}$ & -0.304$^{***}$ & -0.487$^{***}$ \\*
 & (0.084) & (0.083) & (0.086) & (0.182) \\
\quad Male & -0.080$^{*}$ & -0.048 & 0.083 & 0.056 \\*
 & (0.047) & (0.048) & (0.054) & (0.071) \\
\quad Mother's age & 0.045$^{*}$ & 0.043 & 0.038 & 0.055 \\*
 & (0.026) & (0.027) & (0.031) & (0.042) \\
\quad Mother's educ. & 0.025 & -0.035 & 0.031 & 0.014 \\*
 & (0.028) & (0.028) & (0.030) & (0.043) \\
\quad Ch.\ in HH & 0.046 & 0.041 & 0.033 & 0.117$^{**}$ \\*
 & (0.052) & (0.049) & (0.046) & (0.059) \\
\quad Mom married & 0.013 & 0.177 & 0.174 & 0.116 \\*
 & (0.119) & (0.121) & (0.131) & (0.167) \\
\quad Children $\times$ Mar. & -0.014 & -0.061 & -0.057 & -0.072 \\*
 & (0.063) & (0.058) & (0.057) & (0.072) \\
\quad $\kappa_{\eta}$ & 0.018 & 0.052 & 0.301$^{*}$ & 0.520$^{*}$ \\*
 & (0.095) & (0.077) & (0.171) & (0.291) \\
\midrule
$N$ & \multicolumn{4}{c}{2,402} \\
\end{longtable}
\noindent
\begin{minipage}{0.95\textwidth}
\vspace{0.5em}\footnotesize
\textit{Notes:} Intercept $\kappa_0$, coefficients on the production controls $Q_t$, and $\kappa_{\eta}$ on the investment shock $\eta_t$ in the production function of Table~\ref{tab:ces_params_all_it}. Column headings are the input period's child ages.
\textit{Male}, children in household, \textit{Mom married} and their interaction enter raw; mother's education and mother's age are standardized.
Standard errors in parentheses (delta method; provisional pending the bootstrap); n.d.: not defined; $^{***}p<0.01$, $^{**}p<0.05$, $^{*}p<0.10$.
\textit{AI disclosure:} table produced by code written with Claude Opus 5 (claude-opus-5) and Claude Sonnet 5 (claude-sonnet-5) and Claude Opus 5.5 (claude-opus-5-5), Anthropic, 2026-10-08 (CES-share kernel reused from MLE/estimator/plot\_elasticities.R, mixed model lineage; see git history); MLE/estimator/intertemporal/make\_it\_paper\_outputs.R::write\_it\_controls().
\end{minipage}
\par
\endgroup

\begin{landscape}
\begingroup
\csname singlespacing\endcsname 
\centering
\footnotesize
\setlength{\tabcolsep}{4pt}
\begin{longtable}{lccccccccc}
\caption{Period-0 Input Mixture (Ages 5--6)}
\label{tab:mixture_it}\\
\toprule
 & $\ln\theta_{C,0}$ & $\ln\theta_{N,0}$ & $\ln I_{-1}$ & $\ln M_C$ & $\ln M_N$ & $E$ & $A_0$ & $P$ & $Y_0$ \\
\midrule
\endfirsthead
\multicolumn{10}{l}{\textit{(continued from previous page)}}\\
\toprule
 & $\ln\theta_{C,0}$ & $\ln\theta_{N,0}$ & $\ln I_{-1}$ & $\ln M_C$ & $\ln M_N$ & $E$ & $A_0$ & $P$ & $Y_0$ \\
\midrule
\endhead
\midrule
\multicolumn{10}{r}{\textit{continued on next page \ldots}}\\
\endfoot
\bottomrule
\endlastfoot
\multicolumn{10}{l}{\textit{Panel A: Mixing weight and component means}} \\*[0.5ex]
\quad Mixing weight $\pi_1$ & \multicolumn{9}{l}{0.902$^{***}$ (0.009); $\pi_2 = 1-\pi_1$} \\
\quad Mean $\mu_1$ & -0.062$^{***}$ & -1.354$^{***}$ & 1.659$^{***}$ & -0.001 & 0.090$^{**}$ & 0.015 & -0.029 & 0.039$^{*}$ & 0.097$^{***}$ \\*
 & (0.020) & (0.051) & (0.108) & (0.022) & (0.043) & (0.022) & (0.022) & (0.021) & (0.015) \\
\quad Mean $\mu_2$ & 0.565$^{***}$ & -1.449$^{***}$ & 1.774$^{***}$ & -0.031 & 0.072 & -0.179$^{*}$ & 0.208$^{**}$ & -0.393$^{***}$ & -1.122$^{***}$ \\*
 & (0.089) & (0.102) & (0.198) & (0.074) & (0.110) & (0.097) & (0.104) & (0.105) & (0.180) \\
\midrule
\multicolumn{10}{l}{\textit{Panel B: Component-1 covariance matrix $\Sigma_1$ (lower triangle)}} \\*[0.5ex]
\quad $\ln\theta_{C,0}$ & 0.153$^{***}$ &  &  &  &  &  &  &  &  \\*
 & (0.013) &  &  &  &  &  &  &  &  \\
\quad $\ln\theta_{N,0}$ & 0.107$^{***}$ & 1.521$^{***}$ &  &  &  &  &  &  &  \\*
 & (0.014) & (0.117) &  &  &  &  &  &  &  \\
\quad $\ln I_{-1}$ & 0.225$^{***}$ & 0.645$^{***}$ & 2.464$^{***}$ &  &  &  &  &  &  \\*
 & (0.028) & (0.078) & (0.408) &  &  &  &  &  &  \\
\quad $\ln M_C$ & 0.114$^{***}$ & 0.224$^{***}$ & 0.832$^{***}$ & 0.705$^{***}$ &  &  &  &  &  \\*
 & (0.010) & (0.027) & (0.079) & (0.029) &  &  &  &  &  \\
\quad $\ln M_N$ & 0.082$^{***}$ & 0.379$^{***}$ & 0.551$^{***}$ & 0.392$^{***}$ & 1.503$^{***}$ &  &  &  &  \\*
 & (0.013) & (0.046) & (0.079) & (0.033) & (0.142) &  &  &  &  \\
\quad $E$ & 0.115$^{***}$ & 0.276$^{***}$ & 0.759$^{***}$ & 0.483$^{***}$ & 0.399$^{***}$ & 0.936$^{***}$ &  &  &  \\*
 & (0.011) & (0.031) & (0.075) & (0.023) & (0.035) & (0.031) &  &  &  \\
\quad $A_0$ & 0.088$^{***}$ & 0.281$^{***}$ & 0.846$^{***}$ & 0.311$^{***}$ & 0.171$^{***}$ & 0.440$^{***}$ & 0.929$^{***}$ &  &  \\*
 & (0.010) & (0.031) & (0.085) & (0.020) & (0.031) & (0.023) & (0.031) &  &  \\
\quad $P$ & 0.117$^{***}$ & 0.321$^{***}$ & 0.872$^{***}$ & 0.485$^{***}$ & 0.327$^{***}$ & 0.447$^{***}$ & 0.380$^{***}$ & 0.874$^{***}$ &  \\*
 & (0.010) & (0.030) & (0.080) & (0.022) & (0.033) & (0.022) & (0.022) & (0.028) &  \\
\quad $Y_0$ & 0.074$^{***}$ & 0.195$^{***}$ & 0.584$^{***}$ & 0.300$^{***}$ & 0.195$^{***}$ & 0.290$^{***}$ & 0.308$^{***}$ & 0.491$^{***}$ & 0.411$^{***}$ \\*
 & (0.007) & (0.021) & (0.055) & (0.015) & (0.022) & (0.016) & (0.016) & (0.018) & (0.014) \\
\midrule
\multicolumn{10}{l}{\textit{Panel C: Component-2 covariance matrix $\Sigma_2$ (lower triangle)}} \\*[0.5ex]
\quad $\ln\theta_{C,0}$ & 1.332$^{***}$ &  &  &  &  &  &  &  &  \\*
 & (0.166) &  &  &  &  &  &  &  &  \\
\quad $\ln\theta_{N,0}$ & 0.217$^{**}$ & 1.106$^{***}$ &  &  &  &  &  &  &  \\*
 & (0.099) & (0.176) &  &  &  &  &  &  &  \\
\quad $\ln I_{-1}$ & 0.894$^{***}$ & 0.727$^{***}$ & 3.082$^{***}$ &  &  &  &  &  &  \\*
 & (0.196) & (0.184) & (0.666) &  &  &  &  &  &  \\
\quad $\ln M_C$ & 0.429$^{***}$ & 0.284$^{***}$ & 0.961$^{***}$ & 0.783$^{***}$ &  &  &  &  &  \\*
 & (0.085) & (0.083) & (0.173) & (0.089) &  &  &  &  &  \\
\quad $\ln M_N$ & 0.170 & 0.243$^{**}$ & 0.652$^{***}$ & 0.510$^{***}$ & 1.311$^{***}$ &  &  &  &  \\*
 & (0.112) & (0.115) & (0.209) & (0.098) & (0.215) &  &  &  &  \\
\quad $E$ & 0.618$^{***}$ & 0.321$^{***}$ & 1.197$^{***}$ & 0.783$^{***}$ & 0.631$^{***}$ & 1.542$^{***}$ &  &  &  \\*
 & (0.115) & (0.107) & (0.233) & (0.108) & (0.130) & (0.179) &  &  &  \\
\quad $A_0$ & 0.224$^{**}$ & 0.111 & 1.151$^{***}$ & 0.308$^{***}$ & 0.463$^{***}$ & 0.581$^{***}$ & 1.597$^{***}$ &  &  \\*
 & (0.103) & (0.108) & (0.228) & (0.090) & (0.129) & (0.122) & (0.169) &  &  \\
\quad $P$ & 0.584$^{***}$ & 0.414$^{***}$ & 1.132$^{***}$ & 0.650$^{***}$ & 0.658$^{***}$ & 0.826$^{***}$ & 0.501$^{***}$ & 1.982$^{***}$ &  \\*
 & (0.124) & (0.123) & (0.230) & (0.102) & (0.140) & (0.135) & (0.131) & (0.198) &  \\
\quad $Y_0$ & 1.035$^{***}$ & 0.205 & 0.866$^{**}$ & 0.870$^{***}$ & 0.670$^{***}$ & 0.667$^{***}$ & 0.425$^{*}$ & 1.443$^{***}$ & 5.641$^{***}$ \\*
 & (0.227) & (0.197) & (0.362) & (0.169) & (0.225) & (0.216) & (0.227) & (0.277) & (0.655) \\
\end{longtable}
\noindent
\begin{minipage}{0.95\textwidth}
\vspace{0.5em}\footnotesize
\textit{Notes:} Initial condition at ages 5--6 (period 0): the two-component mixture $x \sim \pi_1 N(\mu_1, \Sigma_1) + (1-\pi_1) N(\mu_2, \Sigma_2)$ over $x = (\ln\theta_{C,0}, \ln\theta_{N,0}, \ln I_{-1}, \ln M_C, \ln M_N, E, A_0, P, Y_0)$.
$\ln\theta_{C,0}$ and $\ln\theta_{N,0}$ are log child cognitive and noncognitive skill, $\ln I_{-1}$ is log parental investment at ages 3--4 (the lag in the period-0 investment equation), and $\ln M_C$, $\ln M_N$ are log maternal cognitive and noncognitive skill, all on their anchored scales.
$E$ (mother's education at age 30), $A_0$ (mother's age) and $P$ (log permanent income) are standardized over the estimation sample (2,402 children), $Y_0$ (log current income) over the children with a period-0 report.
Panels B and C report the lower triangle of $\Sigma_k = L_kL_k'$, variances on the diagonal.
Standard errors in parentheses (delta method; provisional pending the bootstrap); n.d.: not defined; $^{***}p<0.01$, $^{**}p<0.05$, $^{*}p<0.10$.
\textit{AI disclosure:} table produced by code written with Claude Opus 5 (claude-opus-5) and Claude Sonnet 5 (claude-sonnet-5) and Claude Opus 5.5 (claude-opus-5-5), Anthropic, 2026-10-08 (shared formatting and delta-method helpers by Claude Opus 5 and Claude Sonnet 5; see git history); MLE/estimator/intertemporal/make\_it\_paper\_outputs.R::write\_it\_mixture().
\end{minipage}
\par
\endgroup
\end{landscape}

\begingroup
\csname singlespacing\endcsname 
\centering
\begin{longtable}{lcccc}
\caption{Initial Covariate Equations at Ages 5--6 (Sequential Probits)}
\label{tab:covariates_p4_it}\\
\toprule
 & Male & Mom married ($R_0$) & $K_0\ge 2$ & $K_0\ge 3 \mid K_0\ge 2$ \\
\midrule
\endfirsthead
\multicolumn{5}{l}{\textit{(continued from previous page)}}\\
\toprule
 & Male & Mom married ($R_0$) & $K_0\ge 2$ & $K_0\ge 3 \mid K_0\ge 2$ \\
\midrule
\endhead
\midrule
\multicolumn{5}{r}{\textit{continued on next page \ldots}}\\
\endfoot
\bottomrule
\endlastfoot
Intercept & -0.264$^{***}$ & 0.195 & -0.047 & -0.722$^{***}$ \\*
 & (0.093) & (0.121) & (0.096) & (0.130) \\
$\ln\theta_{C,0}$ & -0.095$^{*}$ & -0.150$^{**}$ & -0.055 & -0.012 \\*
 & (0.053) & (0.070) & (0.056) & (0.070) \\
$\ln\theta_{N,0}$ & -0.162$^{***}$ & 0.006 & 0.003 & 0.021 \\*
 & (0.028) & (0.034) & (0.028) & (0.035) \\
$\ln I_{-1}$ & 0.023 & 0.264$^{***}$ & -0.021 & -0.053 \\*
 & (0.042) & (0.064) & (0.043) & (0.049) \\
$\ln M_C$ & 0.078 & -0.027 & 0.106$^{*}$ & 0.117 \\*
 & (0.056) & (0.073) & (0.059) & (0.074) \\
$\ln M_N$ & 0.072$^{***}$ & -0.193$^{***}$ & -0.038 & -0.002 \\*
 & (0.028) & (0.037) & (0.029) & (0.037) \\
Mother's educ.\ ($E$) & -0.027 & -0.120$^{***}$ & 0.066$^{*}$ & 0.007 \\*
 & (0.036) & (0.045) & (0.037) & (0.050) \\
Mother's age ($A_0$) & 0.023 & -0.067 & -0.181$^{***}$ & -0.112$^{**}$ \\*
 & (0.036) & (0.048) & (0.037) & (0.052) \\
Log perm.\ inc.\ ($P$) & -0.044 & 0.542$^{***}$ & -0.015 & -0.080 \\*
 & (0.042) & (0.054) & (0.043) & (0.059) \\
Log curr.\ inc.\ ($Y_0$) & 0.013 & 0.513$^{***}$ & 0.011 & -0.068 \\*
 & (0.035) & (0.052) & (0.036) & (0.056) \\
Male &  & -0.001 & 0.018 & -0.031 \\*
 &  & (0.066) & (0.055) & (0.070) \\
Mom married ($R_0$) &  &  & 0.744$^{***}$ & 0.281$^{***}$ \\*
 &  &  & (0.072) & (0.105) \\
\end{longtable}
\noindent
\begin{minipage}{0.95\textwidth}
\vspace{0.5em}\footnotesize
\textit{Notes:} Sequential probits generating the discrete covariates at ages 5--6 (period 0), each conditioning on the period-0 mixture vector $x$ (Table~\ref{tab:mixture_it}) and on the covariates generated before it, all with standard normal errors: $\text{Male} = \mathbf{1}\{a^m + b^{m\prime}x + \varepsilon^m > 0\}$; $R_0 = \mathbf{1}\{a^R + b^{R\prime}x + g^R\,\text{Male} + \varepsilon^R > 0\}$; and children in the household $K_0 \in \{\le 1, 2, \ge 3\}$ (a count of 0 is recoded to 1) from two sequential binary probits, $\Pr(K_0 \ge 2) = \Phi(a^{K1} + b^{K1\prime}x + g^{K1}\,\text{Male} + h^{K1}R_0)$ and $\Pr(K_0 \ge 3 \mid K_0 \ge 2) = \Phi(a^{K2} + b^{K2\prime}x + g^{K2}\,\text{Male} + h^{K2}R_0)$. A blank cell means the equation does not contain that term.
The coordinates of $x$ enter on the scales of Table~\ref{tab:mixture_it}; Male and $R_0$ are 0/1 indicators.
Standard errors in parentheses (delta method; provisional pending the bootstrap); n.d.: not defined; $^{***}p<0.01$, $^{**}p<0.05$, $^{*}p<0.10$.
\textit{AI disclosure:} table produced by code written with Claude Opus 5 (claude-opus-5) and Claude Sonnet 5 (claude-sonnet-5) and Claude Opus 5.5 (claude-opus-5-5), Anthropic, 2026-10-08 (shared formatting and delta-method helpers by Claude Opus 5 and Claude Sonnet 5; see git history); MLE/estimator/intertemporal/make\_it\_paper\_outputs.R::write\_it\_covariate\_eqs().
\end{minipage}
\par
\endgroup

\bigskip

\begingroup
\csname singlespacing\endcsname 
\centering
\begin{longtable}{lcccc}
\caption{Covariate Laws of Motion}
\label{tab:covariate_transitions_it}\\
\toprule
 & 7--8 & 9--10 & 11--12 & 13--14 \\
\midrule
\endfirsthead
\multicolumn{5}{l}{\textit{(continued from previous page)}}\\
\toprule
 & 7--8 & 9--10 & 11--12 & 13--14 \\
\midrule
\endhead
\midrule
\multicolumn{5}{r}{\textit{continued on next page \ldots}}\\
\endfoot
\bottomrule
\endlastfoot
\multicolumn{5}{l}{\textit{Panel A: Mom married, $R_t$ (probit)}} \\*[0.5ex]
\quad Intercept & -1.387$^{***}$ & -1.297$^{***}$ & -1.376$^{***}$ & -1.153$^{***}$ \\*
 & (0.194) & (0.228) & (0.203) & (0.212) \\
\quad $\ln\theta_{C,t}$ & 0.020 & 0.027 & -0.022 & 0.110 \\*
 & (0.072) & (0.077) & (0.074) & (0.081) \\
\quad $\ln\theta_{N,t}$ & 0.028 & 0.023 & 0.022 & -0.015 \\*
 & (0.041) & (0.039) & (0.033) & (0.037) \\
\quad $\ln M_C$ & -0.032 & -0.078 & 0.003 & -0.057 \\*
 & (0.081) & (0.117) & (0.085) & (0.096) \\
\quad $\ln M_N$ & -0.048 & -0.067 & -0.068 & 0.048 \\*
 & (0.043) & (0.046) & (0.043) & (0.051) \\
\quad $\ln I_{t-1}$ & 0.086$^{*}$ & 0.039 & 0.074 & -0.027 \\*
 & (0.047) & (0.070) & (0.057) & (0.071) \\
\quad Log perm.\ inc.\ ($P$) & 0.562$^{***}$ & 0.573$^{***}$ & 0.563$^{***}$ & 0.414$^{***}$ \\*
 & (0.063) & (0.068) & (0.071) & (0.080) \\
\quad Log curr.\ inc., lagged ($Y_{t-1}$) & -0.116$^{**}$ & -0.106$^{**}$ & -0.034 & -0.052 \\*
 & (0.049) & (0.044) & (0.050) & (0.068) \\
\quad Mother's educ.\ ($E$) & -0.068 & -0.102$^{*}$ & -0.062 & -0.044 \\*
 & (0.055) & (0.060) & (0.057) & (0.063) \\
\quad Mother's age ($A_t$) & -0.055 & 0.090$^{*}$ & -0.021 & -0.140$^{**}$ \\*
 & (0.049) & (0.050) & (0.048) & (0.059) \\
\quad Mom married, lagged ($R_{t-1}$) & 2.460$^{***}$ & 2.698$^{***}$ & 2.713$^{***}$ & 2.493$^{***}$ \\*
 & (0.203) & (0.207) & (0.209) & (0.237) \\
\quad Ch.\ in HH, lagged ($K_{t-1}$) & 0.237$^{***}$ & 0.188$^{***}$ & 0.149$^{**}$ & 0.030 \\*
 & (0.082) & (0.071) & (0.073) & (0.083) \\
\quad Ch.\ $\times$ married, lagged ($K_{t-1}R_{t-1}$) & -0.115 & -0.176$^{*}$ & -0.200$^{**}$ & 0.042 \\*
 & (0.105) & (0.097) & (0.092) & (0.106) \\
\midrule
\multicolumn{5}{l}{\textit{Panel B: Log current income, $Y_t$ (linear, normal errors)}} \\*[0.5ex]
\quad Intercept & -0.208$^{***}$ & -0.247$^{***}$ & -0.212$^{***}$ & -0.135$^{**}$ \\*
 & (0.056) & (0.057) & (0.052) & (0.066) \\
\quad Log curr.\ inc., lagged ($Y_{t-1}$) & 0.257$^{***}$ & 0.165$^{***}$ & 0.288$^{***}$ & 0.263$^{***}$ \\*
 & (0.025) & (0.028) & (0.023) & (0.032) \\
\quad Log perm.\ inc.\ ($P$) & 0.317$^{***}$ & 0.393$^{***}$ & 0.347$^{***}$ & 0.325$^{***}$ \\*
 & (0.031) & (0.033) & (0.031) & (0.040) \\
\quad Mother's age ($A_t$) & 0.031 & 0.035 & 0.039$^{*}$ & -0.077$^{***}$ \\*
 & (0.022) & (0.023) & (0.022) & (0.028) \\
\quad Mom married ($R_t$) & 0.401$^{***}$ & 0.356$^{***}$ & 0.307$^{***}$ & 0.374$^{***}$ \\*
 & (0.049) & (0.051) & (0.046) & (0.059) \\
\quad $\ln M_C$ & -0.020 & -0.014 & 0.007 & 0.026 \\*
 & (0.035) & (0.036) & (0.034) & (0.045) \\
\quad $\ln M_N$ & 0.004 & 0.026 & 0.013 & 0.005 \\*
 & (0.019) & (0.020) & (0.019) & (0.025) \\
\quad Mother's educ.\ ($E$) & 0.057$^{**}$ & 0.041 & 0.047$^{*}$ & 0.047 \\*
 & (0.026) & (0.027) & (0.025) & (0.033) \\
\quad Ch.\ in HH, lagged ($K_{t-1}$) & -0.049$^{*}$ & -0.008 & 0.005 & -0.058$^{**}$ \\*
 & (0.025) & (0.024) & (0.020) & (0.026) \\
\quad $\sigma_Y$ & 0.795$^{***}$ & 0.838$^{***}$ & 0.752$^{***}$ & 0.888$^{***}$ \\*
 & (0.013) & (0.014) & (0.013) & (0.017) \\
\midrule
\multicolumn{5}{l}{\textit{Panel C: Change in children in household, $\Delta K_t$ (ordered probit)}} \\*[0.5ex]
\quad $\ln\theta_{C,t}$ & 0.046 & -0.066 & -0.048 & -0.089 \\*
 & (0.055) & (0.062) & (0.068) & (0.073) \\
\quad $\ln\theta_{N,t}$ & 0.005 & 0.003 & 0.039 & 0.025 \\*
 & (0.036) & (0.032) & (0.030) & (0.030) \\
\quad $\ln M_C$ & 0.024 & 0.036 & 0.038 & -0.053 \\*
 & (0.062) & (0.070) & (0.074) & (0.080) \\
\quad $\ln M_N$ & -0.057$^{*}$ & 0.019 & 0.013 & -0.056 \\*
 & (0.033) & (0.037) & (0.038) & (0.043) \\
\quad $\ln I_{t-1}$ & -0.040 & -0.081$^{*}$ & -0.072 & -0.036 \\*
 & (0.039) & (0.044) & (0.049) & (0.065) \\
\quad Log perm.\ inc.\ ($P$) & -0.006 & 0.006 & 0.054 & 0.012 \\*
 & (0.053) & (0.059) & (0.057) & (0.066) \\
\quad Log curr.\ inc., lagged ($Y_{t-1}$) & -0.067 & -0.033 & -0.050 & 0.016 \\*
 & (0.042) & (0.041) & (0.043) & (0.045) \\
\quad Mother's educ.\ ($E$) & 0.053 & 0.089$^{*}$ & -0.021 & 0.089 \\*
 & (0.042) & (0.047) & (0.050) & (0.057) \\
\quad Mother's age ($A_t$) & -0.252$^{***}$ & -0.200$^{***}$ & -0.181$^{***}$ & -0.206$^{***}$ \\*
 & (0.039) & (0.041) & (0.044) & (0.049) \\
\quad Mom married, lagged ($R_{t-1}$) & 0.032 & -0.113 & -0.344$^{*}$ & -0.459$^{**}$ \\*
 & (0.181) & (0.210) & (0.204) & (0.217) \\
\quad Mom married ($R_t$) & 0.554$^{***}$ & 0.662$^{***}$ & 0.457$^{***}$ & 0.468$^{***}$ \\*
 & (0.106) & (0.122) & (0.120) & (0.149) \\
\quad Ch.\ in HH, lagged ($K_{t-1}$) & -0.197$^{***}$ & -0.024 & -0.123$^{*}$ & -0.345$^{***}$ \\*
 & (0.073) & (0.080) & (0.066) & (0.059) \\
\quad Ch.\ $\times$ married, lagged ($K_{t-1}R_{t-1}$) & -0.049 & -0.171$^{*}$ & 0.003 & 0.151$^{*}$ \\*
 & (0.085) & (0.093) & (0.082) & (0.080) \\
\quad Threshold $\tau_1$ & -2.450$^{***}$ & -2.352$^{***}$ & -2.388$^{***}$ & -2.638$^{***}$ \\*
 & (0.222) & (0.214) & (0.189) & (0.192) \\
\quad Threshold $\tau_2$ & 0.877$^{***}$ & 1.220$^{***}$ & 1.153$^{***}$ & 1.059$^{***}$ \\*
 & (0.185) & (0.205) & (0.173) & (0.156) \\
\end{longtable}
\noindent
\begin{minipage}{0.95\textwidth}
\vspace{0.5em}\footnotesize
\textit{Notes:} Laws of motion of the time-varying covariates. Within period $t = 1,\dots,4$ the model generates $\theta_t \to R_t \to Y_t \to \Delta K_t \to \ln I_t$ (Table~\ref{tab:investment_all_it} for the last). Column headings are the child ages of period $t$, the period whose covariate is generated.
Panel A: $R_t = \mathbf{1}\{m^R_t + \varepsilon^R_t > 0\}$, $m^R_t = a^R_t + b^{R\prime}_t(\ln\theta_{C,t}, \ln\theta_{N,t}, \ln M_C, \ln M_N) + g^R_t\ln I_{t-1} + d^{R\prime}_t(P, Y_{t-1}, E, A_t, R_{t-1}, K_{t-1}, K_{t-1}R_{t-1})$.
Panel B: $Y_t = a^Y_t + \rho^Y_t Y_{t-1} + b^Y_t P + d^Y_t A_t + g^Y_t R_t + h^Y_{C,t}\ln M_C + h^Y_{N,t}\ln M_N + e^Y_t E + k^Y_t K_{t-1} + \varepsilon^Y_t$, $\varepsilon^Y_t \sim N(0, \sigma^2_{Y,t})$.
Panel C: $\Delta K_t = K_t - K_{t-1}$ in three categories $\{\le -1, 0, \ge +1\}$, $\Pr(\Delta K_t \le -1) = \Phi(\tau_{1,t} - m^K_t)$ and $\Pr(\Delta K_t \le 0) = \Phi(\tau_{2,t} - m^K_t)$, with $m^K_t = b^{K\prime}_t(\ln\theta_{C,t}, \ln\theta_{N,t}, \ln M_C, \ln M_N) + g^K_t\ln I_{t-1} + d^{K\prime}_t(P, Y_{t-1}, E, A_t, R_{t-1}, R_t, K_{t-1}, K_{t-1}R_{t-1})$ and no intercept (the thresholds absorb it). $\varepsilon^R_t$ and the ordered-probit error are standard normal.
Scaling: $R$ is a 0/1 indicator; $K$ (children in the household, a count of 0 recoded to 1) and $K\cdot R$ enter raw. $E$, $A$ and $P$ are standardized over the estimation sample (2,402 children) and $Y$ per period over the children observed in that period. $A_t = A_0 + 2t$ enters as standardized $A_0$ (the age shift is absorbed by the intercepts, and by the thresholds in Panel C). Latent regressors enter as log factors on their anchored scales.
Standard errors in parentheses (delta method; provisional pending the bootstrap); n.d.: not defined; $^{***}p<0.01$, $^{**}p<0.05$, $^{*}p<0.10$.
\textit{AI disclosure:} table produced by code written with Claude Opus 5 (claude-opus-5) and Claude Sonnet 5 (claude-sonnet-5) and Claude Opus 5.5 (claude-opus-5-5), Anthropic, 2026-10-08 (shared formatting and delta-method helpers by Claude Opus 5 and Claude Sonnet 5; see git history); MLE/estimator/intertemporal/make\_it\_paper\_outputs.R::write\_it\_covariate\_eqs().
\end{minipage}
\par
\endgroup

\begin{figure}
    \centering
    \includegraphics[width=\linewidth]{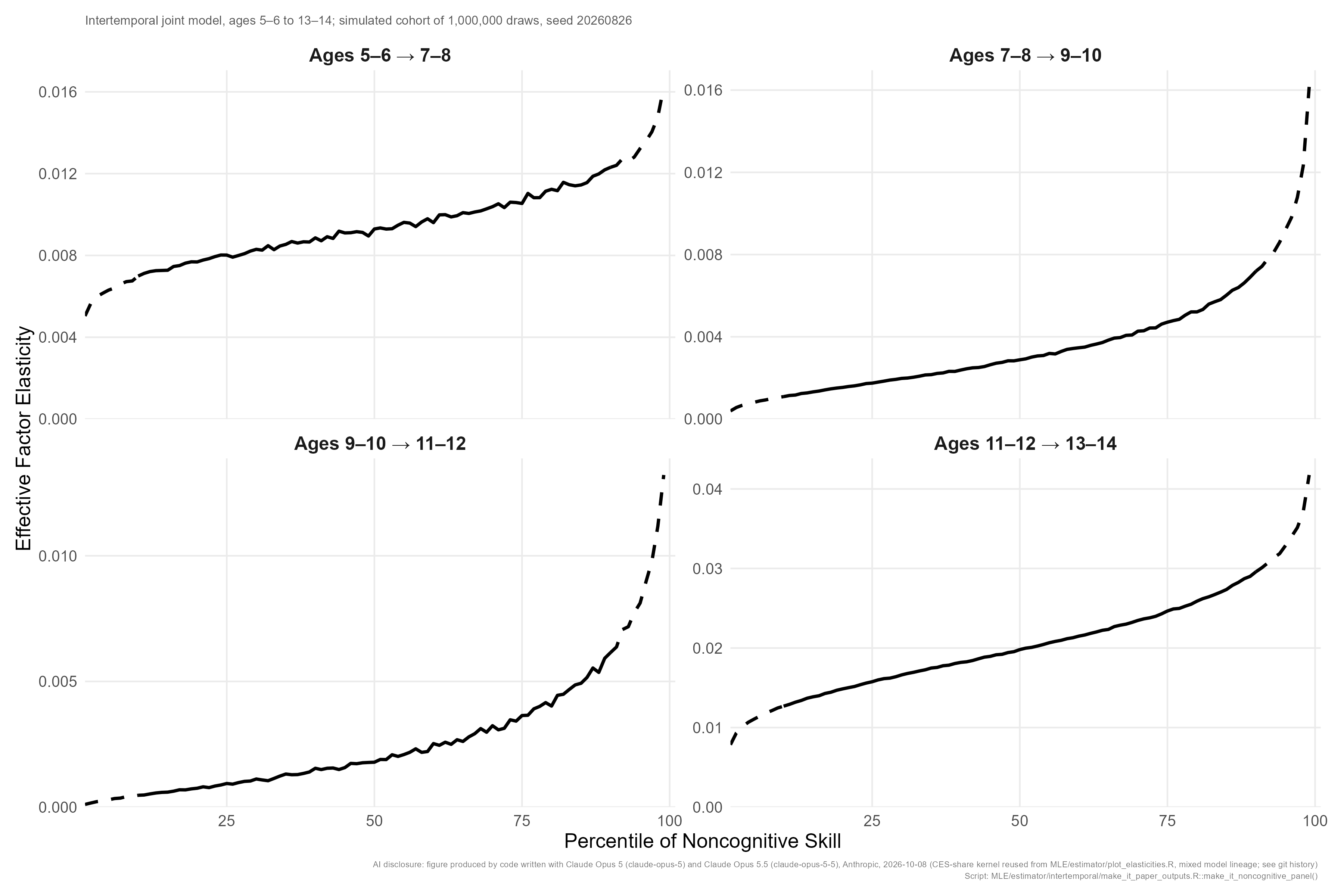}
    \caption{Noncognitive AMR to Investment in the Short Run}
    \label{fig:amr_noncog_short}
\end{figure}

\begin{figure}
    \centering
    \includegraphics[width=\linewidth]{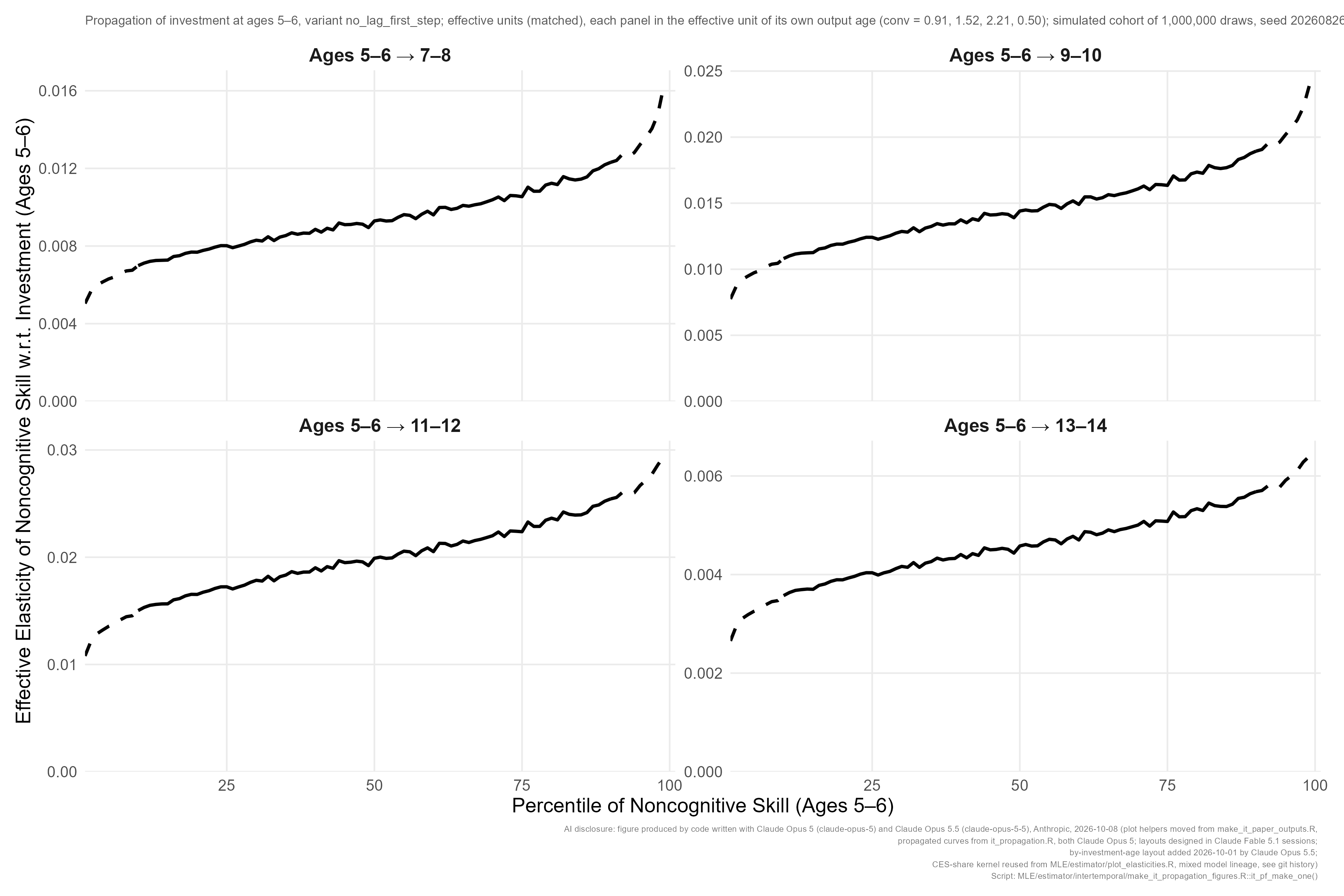}
    \caption{AMR to Investment at Ages 5-6 on Noncognitive Skill Across Childhood}
    \label{fig:amr_noncog_long}
\end{figure}

\begin{figure}
    \centering
    \includegraphics[width=\linewidth]{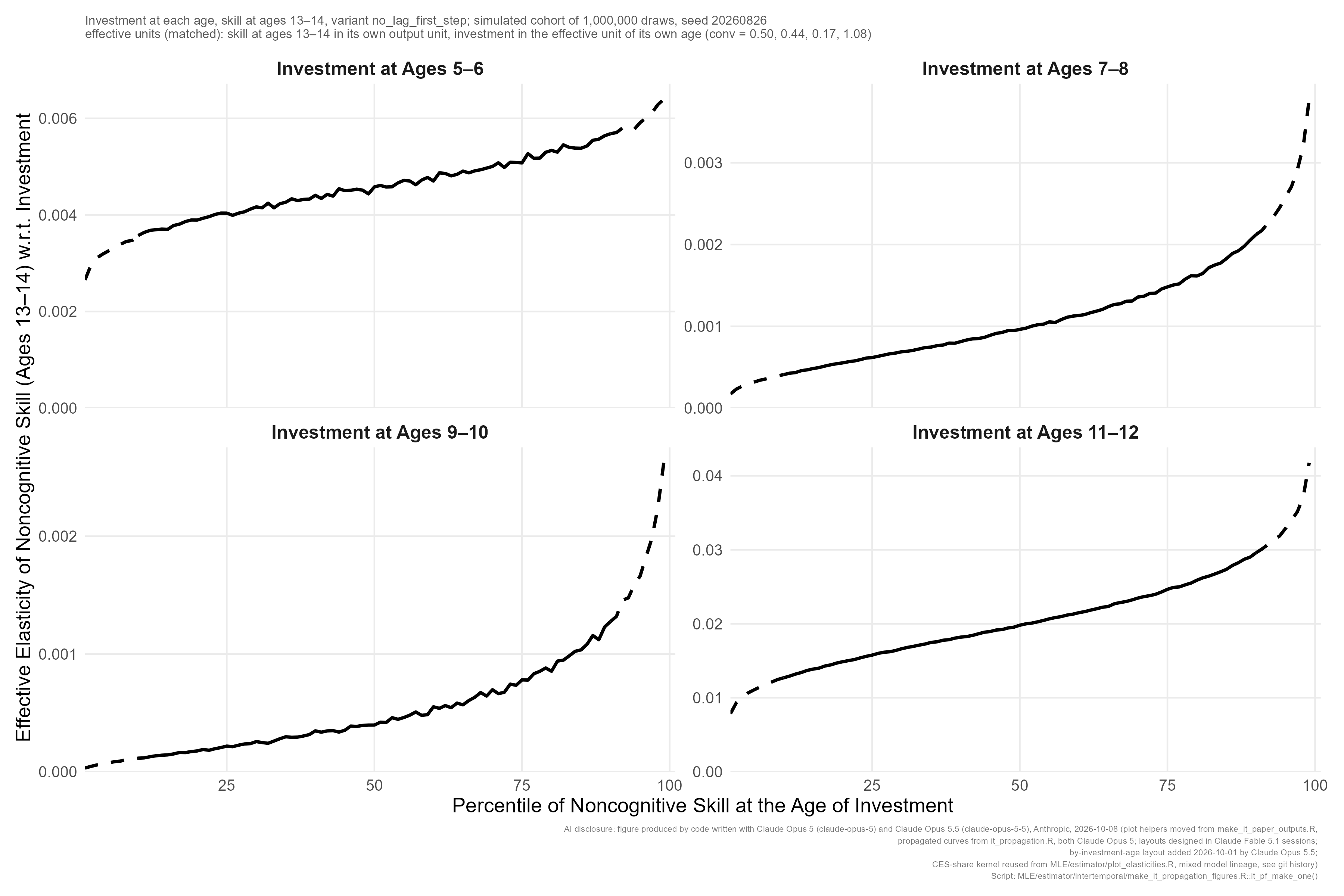}
    \caption{AMR to Investment at Different Ages on Cognitive Skill at Ages 13-14}
    \label{fig:amr_noncog_invage}
\end{figure}

\begin{figure}
    \centering
    \includegraphics[width=\linewidth]{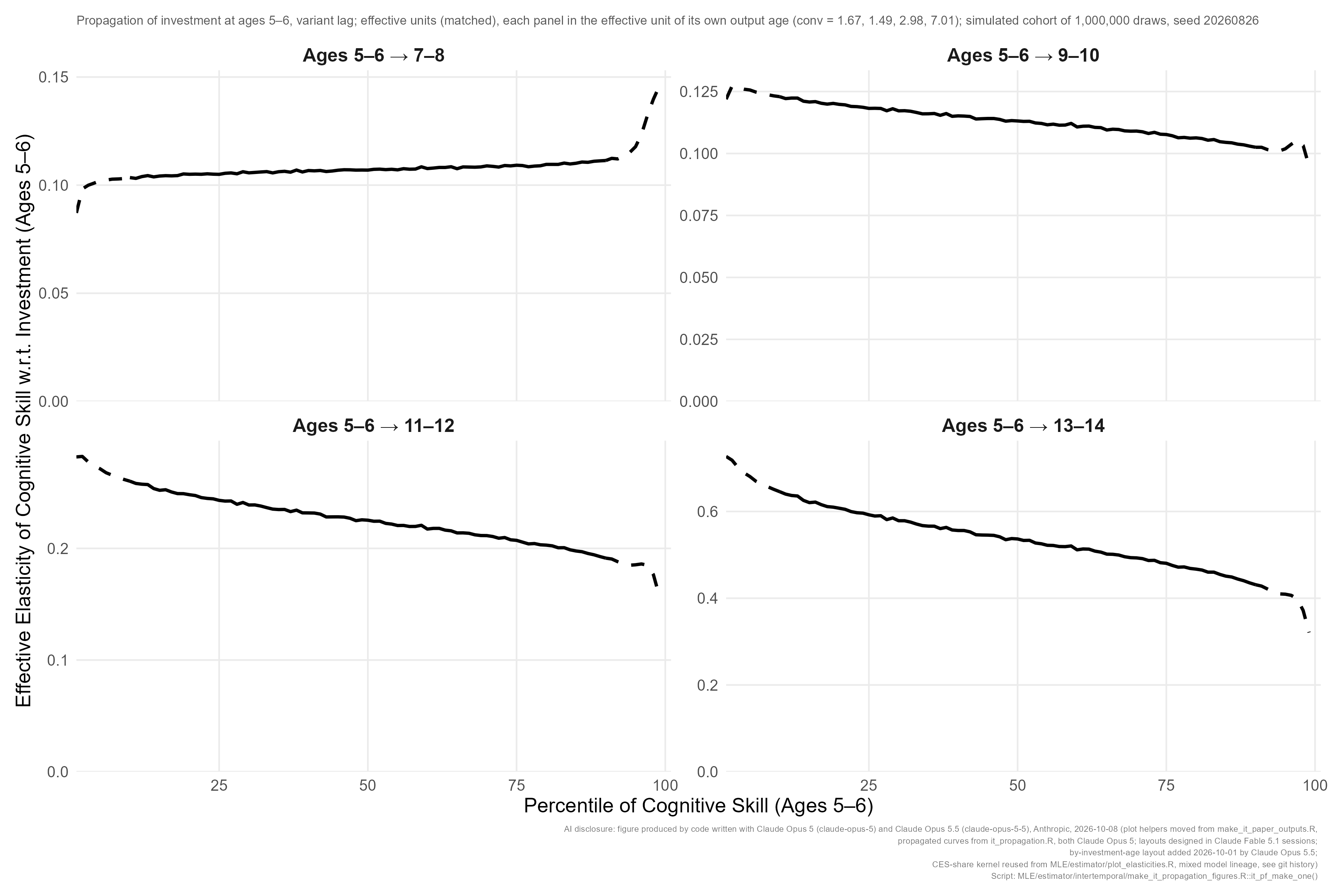}
    \caption{AMR to Investment at Ages 5-6 on Cognitive Skill Across Childhood, Full Investment Propagation}
    \label{fig:amr_cog_long_propI}
\end{figure}

\begin{figure}
    \centering
    \includegraphics[width=\linewidth]{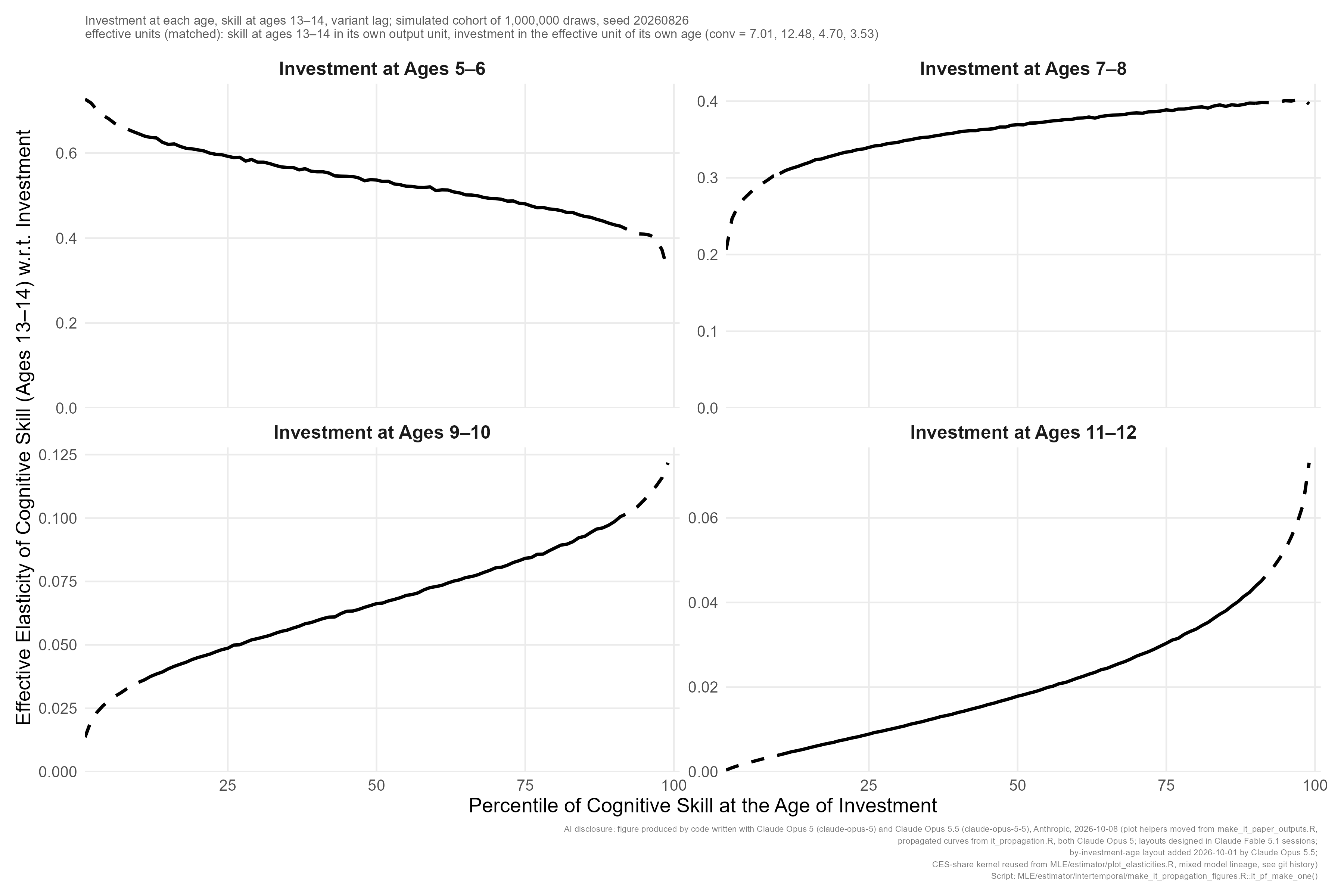}
    \caption{AMR to Investment at Different Ages on Cognitive Skill at Ages 13-14, Full Investment Propagation}
    \label{fig:amr_cog_invage_propI}
\end{figure}

\begin{figure}
    \centering
    \includegraphics[width=\linewidth]{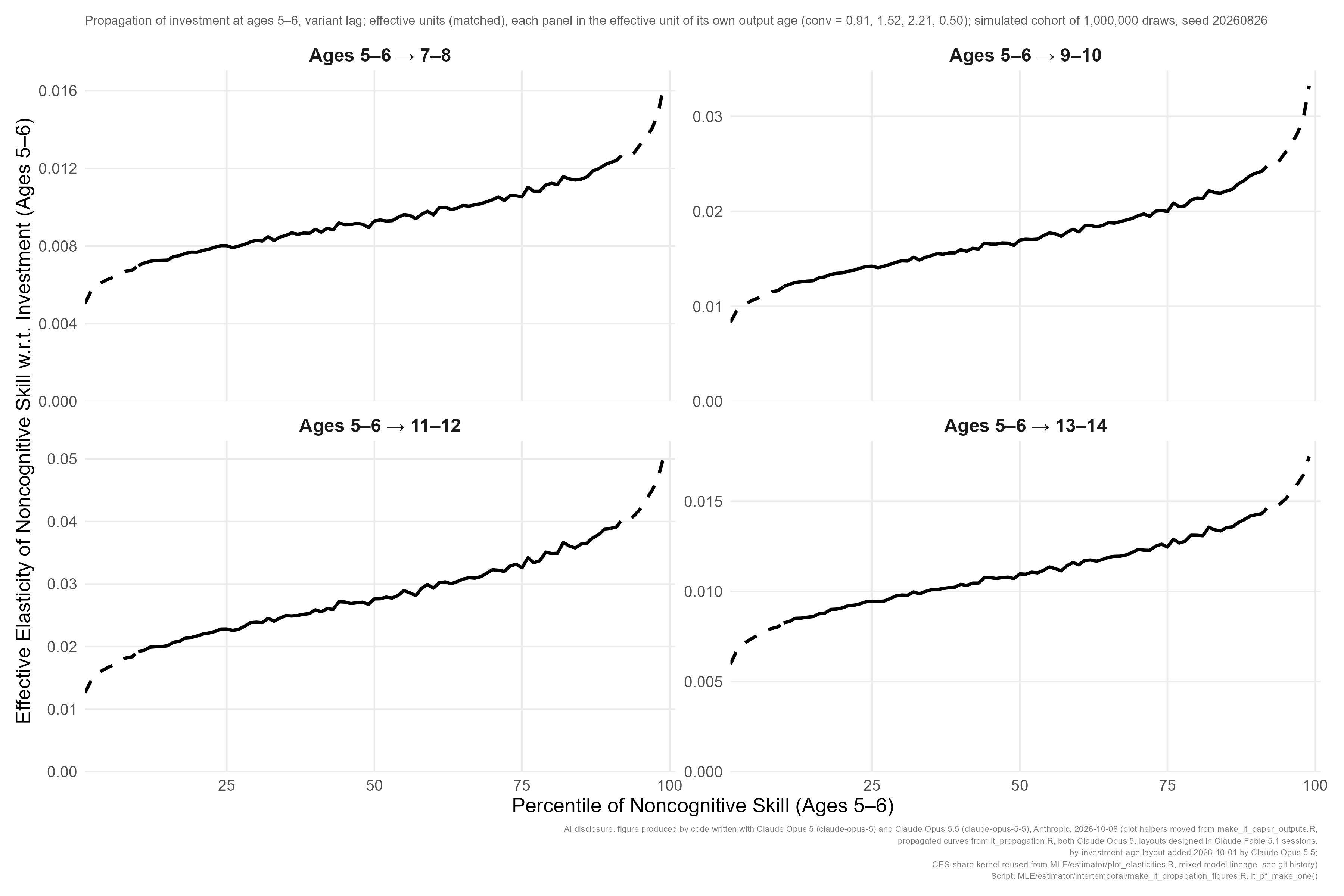}
    \caption{AMR to Investment at Ages 5-6 on Noncognitive Skill Across Childhood, Full Investment Propagation}
    \label{fig:amr_noncog_long_propI}
\end{figure}

\begin{figure}
    \centering
    \includegraphics[width=\linewidth]{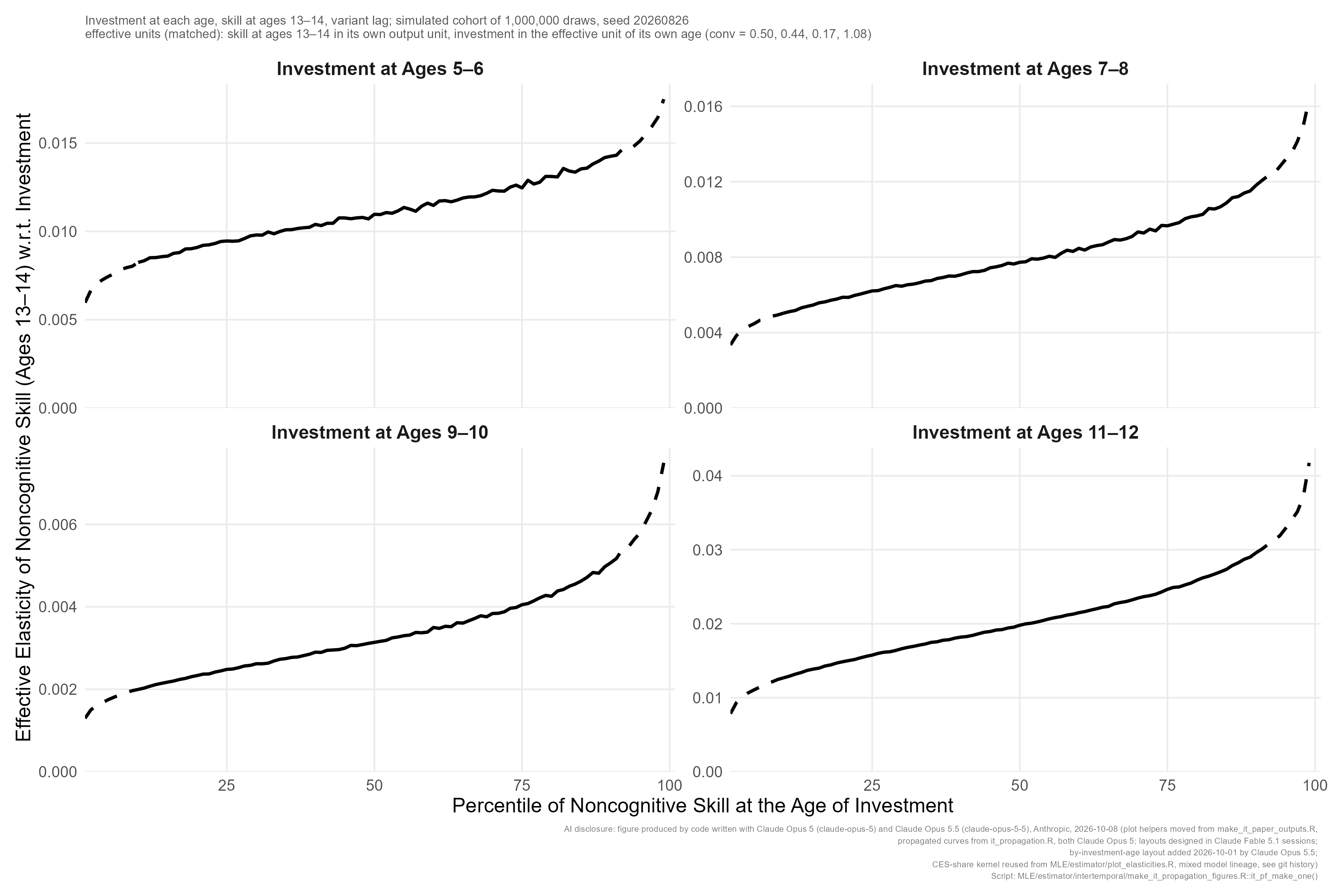}
    \caption{AMR to Investment at Different Ages on Cognitive Skill at Ages 13-14, Full Investment Propagation}
    \label{fig:amr_noncog_invage_propI}
\end{figure}

\section{Calculations and Likelihoods}

\subsection{The Likelihood of the Intertemporal Model}\label{sec:app:it_likelihood}

This appendix writes out the likelihood of the intertemporal model. As in Section \ref{sec:model_skill_form:amr_inv0_cog4}, $t=0$ is ages 5--6 and $t=4$ is ages 13--14, and $t=-1$ (ages 3--4) enters only through lagged investment. The sample is the $N=2{,}402$ firstborn children with at least one proxy for every factor at ages 5--6. For child $n$ the latent log factors are
\[
\mathbf{f}_n=\bigl(\{\ln\theta_{C,n,t},\ln\theta_{N,n,t}\}_{t=0}^{4},\ \{\ln M_{I,n,t}\}_{t=-1}^{4},\ \ln M_{C,n},\ \ln M_{N,n}\bigr),
\]
eighteen in all. The time-invariant covariates are the child's gender $Q_{M,n}$, maternal education at age 30 $Q_{E,n}$, maternal permanent income $W_{P,n}$, and maternal age at the NLSY79 interview data corresponding to when her child was aged 5-6 $Q_{A,n}$. The time-varying covariates are the mother's marital status $Q_{R,n,t}$, the number of the mother's biological, step, and adopted children living in the household $Q_{K,n,t}\geq 1$, the focal child (firstborn) included, and log household income $W_{n,t}$. $Q_E$, $Q_A$, $W_P$, and $W_t$ are standardized within period, while $Q_M$, $Q_R$, and $Q_K$ enter in levels. Gender, education, age, and permanent income are always observed. The other covariates are integrated over when they are missing. 

The initial condition is the nine-vector
\[
\begin{aligned}
\mathbf{x}_n&=(\mathbf{f}_{n,0}',Q_{E,n},Q_{A,n,0},W_{P,n},W_{n,0})',\\
\mathbf{f}_{n,0}&=(\ln\theta_{C,n,0},\ln\theta_{N,n,0},\ln M_{I,n,-1},\ln M_{C,n},\ln M_{N,n})',
\end{aligned}
\]
which follows a two-component mixture of normals as in Section \ref{sec:model_skill_form:covariates}, with weights $\pi_1=\pi$ and $\pi_2=1-\pi$. With the controls and inputs
\[
\begin{aligned}
Q_{n,t}&=(Q_{M,n},Q_{K,n,t},Q_{A,n,t},Q_{E,n},Q_{R,n,t},Q_{K,n,t}Q_{R,n,t})',\\
y_{n,t}&=(\theta_{C,n,t},\theta_{N,n,t},M_{I,n,t},M_{C,n},M_{N,n}),
\end{aligned}
\]
investment ($t=0,\dots,4$) and skills ($t=0,\dots,3$) have conditional means
\begin{equation}\label{eq:it_lik_means}
\begin{aligned}
    \mu^{I}_{n,t}&=\beta_{0,t}+\beta_{\theta_C,t}\ln\theta_{C,n,t}+\beta_{\theta_N,t}\ln\theta_{N,n,t}+\beta_{M_C,t}\ln M_{C,n}+\beta_{M_N,t}\ln M_{N,n}+\beta_{M_I,t}\ln M_{I,n,t-1}\\
    &\quad+Q_{n,t}'\beta_{Q,t}+\beta_{W_P,t}W_{P,n}+\beta_{W,t}W_{n,t},\\
    \mu^{\theta}_{i,n,t}&=\frac{\psi_{it}}{\rho_{it}}\ln\Bigl(\textstyle\sum_{j}\gamma_{i,j,t}\,y_{j,n,t}^{\rho_{it}/\phi_{i,j,t}}\Bigr)+\kappa_{0,i,t}+Q_{n,t}'\kappa_{it}+\kappa_{i,\eta,t}\,\eta_{n,t},\qquad \eta_{n,t}=\ln M_{I,n,t}-\mu^{I}_{n,t},
\end{aligned}
\end{equation}
and normal shocks with standard deviations $\sigma_{\eta,t}$ and $\sigma_{\nu,i,t}$. Marriage and the change in the number of children are probit and ordered probit, and income is linear with a normal error of standard deviation $\sigma_{W,t}$, with indices $\boldsymbol\alpha^{R\prime}_t\mathbf{r}^{R}_{n,t}$, $\boldsymbol\alpha^{K\prime}_t\mathbf{r}^{K}_{n,t}$, and $\boldsymbol\alpha^{W\prime}_t\mathbf{r}^{W}_{n,t}$ ($t=1,\dots,4$). At ages 5--6, gender, marital status, and the number of children are drawn from sequential probits with indices $\boldsymbol\alpha^{M\prime}\mathbf{r}^{M}_{n}$, $\boldsymbol\alpha^{R\prime}_0\mathbf{r}^{R}_{n,0}$, and $\boldsymbol\alpha^{K_2\prime}_0\mathbf{r}^{K}_{n,0}$ (two or more children) and $\boldsymbol\alpha^{K_3\prime}_0\mathbf{r}^{K}_{n,0}$ (three or more, given two or more). The regressors are
\begin{equation}\label{eq:it_lik_regressors}
\begin{aligned}
    \mathbf{r}^{M}_n&=(1,\mathbf{x}_n')',\qquad \mathbf{r}^{R}_{n,0}=(1,\mathbf{x}_n',Q_{M,n})',\qquad \mathbf{r}^{K}_{n,0}=(1,\mathbf{x}_n',Q_{M,n},Q_{R,n,0})',\\
    \mathbf{r}^{R}_{n,t}&=(1,\ln\theta_{C,n,t},\ln\theta_{N,n,t},\ln M_{C,n},\ln M_{N,n},\ln M_{I,n,t-1},\\
    &\qquad W_{P,n},W_{n,t-1},Q_{E,n},Q_{A,n,t},Q_{R,n,t-1},Q_{K,n,t-1},Q_{K,n,t-1}Q_{R,n,t-1})',\\
    \mathbf{r}^{W}_{n,t}&=(1,W_{n,t-1},W_{P,n},Q_{A,n,t},Q_{R,n,t},\ln M_{C,n},\ln M_{N,n},Q_{E,n},Q_{K,n,t-1})',\\
    \mathbf{r}^{K}_{n,t}&=(\ln\theta_{C,n,t},\ln\theta_{N,n,t},\ln M_{C,n},\ln M_{N,n},\ln M_{I,n,t-1},\\
    &\qquad W_{P,n},W_{n,t-1},Q_{E,n},Q_{A,n,t},Q_{R,n,t-1},Q_{R,n,t},Q_{K,n,t-1},Q_{K,n,t-1}Q_{R,n,t-1})'.
\end{aligned}
\end{equation}
The number of children enters the likelihood terms through categories. At ages 5--6 the category is $b_{n,0}=\min(Q_{K,n,0},3)$, so the two sequential probits describe one, two, or three or more children. At later ages, I model the change in the number of children in the house $Q_{K,n,t}-Q_{K,n,t-1}\in\{-1, 0, 1\}$ via an ordered probit model. A reported number of children enters every regressor at its actual count. An unreported one is drawn as a category, and the regressors then need a count for it: one, two, or three children at ages 5--6, and at later ages the previous period's count, reported or drawn, plus $-1$, $0$, or $+1$ for categories 1, 2, and 3, never falling below one child (this is enforced by post hoc flooring after the draw for change in children is made).

Let $\mathsf{z}_{n,j,t}$ be proxy $j$ for child $n$ at age $t$, $\mathcal{J}_{n,t}$ the proxies observed for child $n$ at $t$, $\mathcal{J}^{m}_{n}$ the observed maternal proxies, and $p(\mathsf{z}_{n,j,t}\mid f_{i(j),n,t})$ the measurement densities implied by equations (\ref{eq:basic_cont_meas}) and (\ref{eq:discrete_meas}), with period-specific age controls on the discrete proxies (the continuous proxies are residualized with respect to age via OLS before the MLE). Let $\mathbf{c}^{\mathrm{mis}}_n$ collect the unreported marital statuses and categories of the number of children and $\mathbf{W}^{\mathrm{mis}}_n$ the unreported incomes, $W_{n,0}$ included. Unreported values are latent variables integrated against their own laws of motion, so the likelihood of child $n$ is
\begingroup\allowdisplaybreaks
\begin{align}
    \mathcal{L}_n(\Theta)={}&\underbrace{\sum_{k=1}^{2}\pi_k}_{\substack{\text{mixture}\\ \text{component}}}
    \;\underbrace{\sum_{\mathbf{c}^{\mathrm{mis}}_n}}_{\substack{\text{unreported marriage}\\ \text{and number of children}}}
    \;\underbrace{\int\!\!\cdots\!\!\int}_{\substack{\text{latent factors and}\\ \text{unreported income}}}
    \;\underbrace{\phi_9\bigl(\mathbf{x}_n;\boldsymbol\mu^{(k)},\Sigma^{(k)}\bigr)}_{\substack{\text{initial condition}\\ \text{at ages 5--6}}}
    \notag\\[2pt]
    &\times\underbrace{\Phi\bigl((2Q_{M,n}-1)\,\boldsymbol\alpha^{M\prime}\mathbf{r}^{M}_n\bigr)\,
    \Phi\bigl((2Q_{R,n,0}-1)\,\boldsymbol\alpha^{R\prime}_0\mathbf{r}^{R}_{n,0}\bigr)}_{\text{gender and marital status at ages 5--6}}
    \notag\\[2pt]
    &\times\underbrace{\Phi\bigl((2\cdot 1\{b_{n,0}\geq 2\}-1)\,\boldsymbol\alpha^{K_2\prime}_0\mathbf{r}^{K}_{n,0}\bigr)\,
    \Phi\bigl((2\cdot 1\{b_{n,0}=3\}-1)\,\boldsymbol\alpha^{K_3\prime}_0\mathbf{r}^{K}_{n,0}\bigr)^{1\{b_{n,0}\geq 2\}}}_{\text{number of children at ages 5--6}}
    \notag\\[2pt]
    &\times\prod_{t=-1}^{4}\;\underbrace{\prod_{j\in\mathcal{J}_{n,t}}p\bigl(\mathsf{z}_{n,j,t}\mid f_{i(j),n,t}\bigr)}_{\substack{\text{skill and investment}\\ \text{proxies}}}
    \;\underbrace{\prod_{j\in\mathcal{J}^{m}_{n}}p\bigl(\mathsf{z}_{n,j}\mid f_{i(j),n}\bigr)}_{\text{maternal proxies}}
    \;\prod_{t=0}^{4}\underbrace{\frac{1}{\sigma_{\eta,t}}\,\phi\!\left(\frac{\ln M_{I,n,t}-\mu^{I}_{n,t}}{\sigma_{\eta,t}}\right)}_{\text{investment}}
    \notag\\[2pt]
    &\times\prod_{t=0}^{3}\prod_{i\in\{C,N\}}\underbrace{\frac{1}{\sigma_{\nu,i,t}}\,\phi\!\left(\frac{\ln\theta_{i,n,t+1}-\mu^{\theta}_{i,n,t}}{\sigma_{\nu,i,t}}\right)}_{\text{skill production}}
    \notag\\[2pt]
    &\times\prod_{t=1}^{4}\Biggl\{\underbrace{\Phi\bigl((2Q_{R,n,t}-1)\,\boldsymbol\alpha^{R\prime}_t\mathbf{r}^{R}_{n,t}\bigr)}_{\text{marital status}}
    \;\underbrace{\frac{1}{\sigma_{W,t}}\,\phi\!\left(\frac{W_{n,t}-\boldsymbol\alpha^{W\prime}_t\mathbf{r}^{W}_{n,t}}{\sigma_{W,t}}\right)}_{\text{income}}
    \notag\\[2pt]
    &\hphantom{{}\times\prod_{t=1}^{4}\Biggl\{}\times\underbrace{\Bigl[\Phi\bigl(\nabla^{K}_{c_{n,t},t}-\boldsymbol\alpha^{K\prime}_t\mathbf{r}^{K}_{n,t}\bigr)-\Phi\bigl(\nabla^{K}_{c_{n,t}-1,t}-\boldsymbol\alpha^{K\prime}_t\mathbf{r}^{K}_{n,t}\bigr)\Bigr]}_{\text{change in the number of children}}\Biggr\}
    \;d\mathbf{f}_n\,d\mathbf{W}^{\mathrm{mis}}_n,\label{eq:it_lik}
\end{align}\endgroup
and the log-likelihood is $\mathcal{LL}(\Theta)=\sum_{n=1}^{N}\ln\mathcal{L}_n(\Theta)$. The integral has 18 dimensions for a child with every income reported and up to 23 otherwise, and it has no closed form, so I evaluate it by importance sampling.

I evaluate the integral in equation (\ref{eq:it_lik}) by importance sampling, with a proposal that builds each draw in the order in which the model generates a child's history. Draw $s$ begins with the mixture component: $k_s=1$ if its first Sobol coordinate is below $q_1=q$ and $k_s=2$ otherwise, with $q_2=1-q$. The five coordinates of $\mathbf{f}_{n,0}$ are then drawn one at a time, in the order $\ln\theta_{C,n,0}$, $\ln\theta_{N,n,0}$, $\ln M_{I,n,-1}$, $\ln M_{C,n}$, $\ln M_{N,n}$, each given the earlier coordinates under component $k_s$. The observed initial covariates are not used in these draws and enter only through the weight term in equation (\ref{eq:it_lik_sim}). Next come an unreported $W_{n,0}$, drawn from its mixture conditional given the latent block and $(Q_{E,n},Q_{A,n},W_{P,n})$, an unreported $Q_{R,n,0}$, an unreported category $b_{n,0}$, and $\ln M_{I,n,0}$. Each later period $t=1,\dots,4$ follows the model's order within a period: $\ln\theta_{C,n,t}$ and $\ln\theta_{N,n,t}$, then an unreported $Q_{R,n,t}$, an unreported $W_{n,t}$, an unreported change in the number of children, and finally $\ln M_{I,n,t}$.

Unreported covariates are drawn from their laws of motion. Each latent factor is instead drawn from its distribution given everything drawn before it and given its own guiding proxies, so that fewer draws land where the child's likelihood is negligible. The guides of $\ln\theta_{C,n,t}$ are the three PIAT subtests at age $t$, and those of $\ln M_{C,n}$ are the six ASVAB subtests. The guides of $\ln\theta_{N,n,t}$, $\ln M_{I,n,t}$, and $\ln M_{N,n}$ are their anchor proxies: the BPI headstrong subscale, the number of books, and the Rosenberg item ``I feel I do not have much to be proud of.'' Because the draws follow the model's order, the distribution of each draw given everything before it is its own equation in the model. The identity $p(\vartheta\mid\text{past})\,p(\mathbf{z}^{g}\mid\vartheta)=p(\mathbf{z}^{g}\mid\text{past})\,p(\vartheta\mid\text{past},\mathbf{z}^{g})$ therefore cancels every transition density in equation (\ref{eq:it_lik}) against the proposal and replaces each guide's measurement density with its predictive density given past draws and observed proxies. Guides at age $t$ inform only the draws at age $t$, so later proxies affect earlier latents only through the weights. 

The observed initial covariates $\mathbf{x}^{o}_n$ are $Q_{E,n}$, $Q_{A,n,0}$, $W_{P,n}$, and $W_{n,0}$, without $W_{n,0}$ when it is unreported. Partition the mean and covariance of mixture component $k$ into the blocks of these observed coordinates ($o$) and of the latent block $\mathbf{f}_{n,0}$ (for brevity, denoted $f$). The mean and covariance of the observed coordinates conditional on the latent block are then
\[
\boldsymbol\mu^{(k)}_{o|f}(\mathbf{f})=\boldsymbol\mu^{(k)}_{o}+\Sigma^{(k)}_{of}(\Sigma^{(k)}_{ff})^{-1}(\mathbf{f}-\boldsymbol\mu^{(k)}_{f}),\qquad
\Sigma^{(k)}_{o|f}=\Sigma^{(k)}_{oo}-\Sigma^{(k)}_{of}(\Sigma^{(k)}_{ff})^{-1}\Sigma^{(k)}_{fo}.
\]

The objective the optimizer maximizes is then
\begingroup\allowdisplaybreaks
\begin{align}
    \widehat{\mathcal{LL}}(\Theta)={}&\sum_{n=1}^{N}\ln\frac{1}{S}\sum_{s=1}^{S}\exp\bigl(\ell_{n,s}(\Theta)\bigr),\notag\\[2pt]
    \ell_{n,s}(\Theta)={}&\underbrace{\ln\frac{\pi_{k_s}}{q_{k_s}}}_{\substack{\text{component}\\ \text{weight}}}
    +\underbrace{\ln\phi\bigl(\mathbf{x}^{o}_n;\,\boldsymbol\mu^{(k_s)}_{o|f}(\mathbf{f}_{n,0,s}),\,\Sigma^{(k_s)}_{o|f}\bigr)}_{\substack{\text{observed initial covariates}\\ \text{given the latent draw}}}
    \notag\\[2pt]
    &+\underbrace{\ln\Phi\bigl((2Q_{M,n}-1)\,\boldsymbol\alpha^{M\prime}\mathbf{r}^{M}_{n,s}\bigr)
    +d^{R}_{n,0}\ln\Phi\bigl((2Q_{R,n,0}-1)\,\boldsymbol\alpha^{R\prime}_0\mathbf{r}^{R}_{n,0,s}\bigr)}_{\text{gender and reported marital status at ages 5--6}}
    \notag\\[2pt]
    &+\underbrace{d^{K}_{n,0}\ln\Phi\bigl((2\cdot 1\{b_{n,0}\geq 2\}-1)\,\boldsymbol\alpha^{K_2\prime}_0\mathbf{r}^{K}_{n,0,s}\bigr)}_{\substack{\text{reported number of children at ages 5--6:}\\ \text{two or more children}}}
    \notag\\[2pt]
    &+\underbrace{d^{K}_{n,0}\,1\{b_{n,0}\geq 2\}\ln\Phi\bigl((2\cdot 1\{b_{n,0}=3\}-1)\,\boldsymbol\alpha^{K_3\prime}_0\mathbf{r}^{K}_{n,0,s}\bigr)}_{\text{reported three or more, given two or more}}
    \notag\\[2pt]
    &+\sum_{t=1}^{4}\Bigl[\underbrace{d^{R}_{n,t}\ln\Phi\bigl((2Q_{R,n,t}-1)\,\boldsymbol\alpha^{R\prime}_t\mathbf{r}^{R}_{n,t,s}\bigr)}_{\text{reported marital status}}
    +\underbrace{d^{W}_{n,t}\ln\Bigl(\frac{1}{\sigma_{W,t}}\phi\Bigl(\frac{W_{n,t}-\boldsymbol\alpha^{W\prime}_t\mathbf{r}^{W}_{n,t,s}}{\sigma_{W,t}}\Bigr)\Bigr)}_{\text{reported income}}
    \notag\\[2pt]
    &\hphantom{+\sum_{t=1}^{4}\Bigl[}
    +\underbrace{d^{K}_{n,t}\ln\Bigl(\Phi\bigl(\nabla^{K}_{c_{n,t,s},t}-\boldsymbol\alpha^{K\prime}_t\mathbf{r}^{K}_{n,t,s}\bigr)-\Phi\bigl(\nabla^{K}_{c_{n,t,s}-1,t}-\boldsymbol\alpha^{K\prime}_t\mathbf{r}^{K}_{n,t,s}\bigr)\Bigr)}_{\text{reported number of children}}\Bigr]
    \notag\\[2pt]
    &+\underbrace{\sum_{t=-1}^{4}\sum_{j\in\mathcal{J}_{n,t}\setminus\mathcal{J}^{g}_{n,t}}\ln p\bigl(\mathsf{z}_{n,j,t}\mid f_{i(j),n,t,s}\bigr)+\sum_{j\in\mathcal{J}^{m}_{n}\setminus\mathcal{J}^{g}_{n}}\ln p\bigl(\mathsf{z}_{n,j}\mid f_{i(j),n,s}\bigr)}_{\text{proxies not used as guides}}
    \notag\\[2pt]
    &+\underbrace{\sum_{\vartheta\in\mathcal{G}_n}\ln p\bigl(\mathbf{z}^{g}_{n,\vartheta}\mid m_{\vartheta,n,s},V_{\vartheta}\bigr).}_{\text{guides, at their predictive density given the past}}\label{eq:it_lik_sim}
\end{align}\endgroup
A subscript $s$ marks a quantity evaluated at draw $s$. The regressors, $\mu^{I}_{n,t}$, and $\mu^{\theta}_{i,n,t}$ use reported covariates where they exist and drawn values otherwise, so a reported outcome whose lag was unreported reweights the draws of that lag, and the category $c_{n,t,s}$ of a reported $Q_{K,n,t}$ varies across draws when $Q_{K,n,t-1}$ was drawn. The observed initial covariates $\mathbf{x}^{o}_n$ are $Q_{E,n}$, $Q_{A,n,0}$, $W_{P,n}$, and $W_{n,0}$, without $W_{n,0}$ when it is unreported, and their conditional moments given the latent block under mixture component $k$ are
\[
\boldsymbol\mu^{(k)}_{o|f}(\mathbf{f})=\boldsymbol\mu^{(k)}_{o}+\Sigma^{(k)}_{of}(\Sigma^{(k)}_{ff})^{-1}(\mathbf{f}-\boldsymbol\mu^{(k)}_{f}),\qquad \Sigma^{(k)}_{o|f}=\Sigma^{(k)}_{oo}-\Sigma^{(k)}_{of}(\Sigma^{(k)}_{ff})^{-1}\Sigma^{(k)}_{fo}.
\]
$\mathcal{G}_n$ is the set of latents whose guiding proxies child $n$ has; a latent whose guide is missing is drawn from its unconditional transition and keeps its ordinary proxy terms. $\mathcal{J}^{g}_{n,t}$ and $\mathcal{J}^{g}_{n}$ are the guides themselves. $m_{\vartheta,n,s}$ and $V_\vartheta$ are the mean and variance of latent $\vartheta$ given the draw's past: the conditional moments of each coordinate of $\mathbf{f}_{n,0}$ given the earlier coordinates under component $k_s$, $(\mu^{I}_{n,t,s},\sigma^2_{\eta,t})$ for investment, and $(\mu^{\theta}_{i,n,t,s},\sigma^2_{\nu,i,t})$ for skills. With continuous guides $\mathbf{z}^{g}=\mathbf{a}+\boldsymbol\lambda\vartheta+\boldsymbol\varepsilon$, $\boldsymbol\varepsilon\sim N(\mathbf{0},\operatorname{diag}(\boldsymbol\sigma^2))$, and a discrete guide $j$ observed in category $c$, the predictive densities are
\begin{equation}\label{eq:it_lik_guides}
\begin{aligned}
    p\bigl(\mathbf{z}^{g}\mid m,V\bigr)&=\phi_{J}\bigl(\mathbf{z}^{g};\,\mathbf{a}+\boldsymbol\lambda m,\,\operatorname{diag}(\boldsymbol\sigma^2)+V\boldsymbol\lambda\boldsymbol\lambda'\bigr),\\
    p\bigl(\dot z_{j}=c\mid m,V\bigr)&=\Phi\!\left(\frac{\nabla_{j,c}-X_{j}'\beta_{j}-\lambda_{j}m}{\sqrt{1+\lambda_j^2V}}\right)-\Phi\!\left(\frac{\nabla_{j,c-1}-X_{j}'\beta_{j}-\lambda_{j}m}{\sqrt{1+\lambda_j^2V}}\right).
\end{aligned}
\end{equation}
The draws are one unscrambled Sobol sequence of $S=2^{15}=32{,}768$ points in 54 dimensions, shared by every child and every parameter value, so equation (\ref{eq:it_lik_sim}) is a smooth function of $\Theta$ except through the drawn values of unreported discrete covariates. The allocation $q=29{,}392/32{,}768\approx 0.897$ is the starting value of $\pi$ rounded to the Sobol grid and stays fixed, so that $\pi$ enters only through the weight $\pi_{k_s}/q_{k_s}$ and the objective is differentiable in it. I maximize equation (\ref{eq:it_lik_sim}) with L-BFGS-B, using its analytic gradient.

\section{Proofs}\label{sec:proofs}

\subsection{Mechanism Behind Negative Slope}\label{sec:proofs:cov_var_explain}

\begin{theorem}[Log-skill derivative of the expected investment elasticity]
Let $\theta_{t+1}=\bigl(\theta_t^{-1}+I_t^{-1}\bigr)^{-1}$. Then the investment and skill elasticities of next-period skill are
$$
\varepsilon_I \equiv \frac{\partial \ln\theta_{t+1}}{\partial \ln I_t}=\frac{\theta_t}{I_t+\theta_t},
\qquad
\varepsilon_\theta \equiv \frac{\partial \ln\theta_{t+1}}{\partial \ln\theta_t}=\frac{I_t}{I_t+\theta_t},
$$
so that $\varepsilon_I+\varepsilon_\theta=1$. Suppose $(\ln I_t,\ln\theta_t)$ are jointly normal with $\Var(\ln\theta_t)\in(0,\infty)$, and define
$$
\beta_t \equiv \frac{\Cov(\ln I_t,\ln\theta_t)}{\Var(\ln\theta_t)}.
$$
Then $\E[\varepsilon_I\mid\theta_t]$ is everywhere differentiable in $\ln\theta_t$, with
$$
\frac{d}{d\ln\theta_t}\,\E\bigl[\varepsilon_I \bigm| \theta_t\bigr]
=
(1-\beta_t)\,\E\bigl[\varepsilon_I\,\varepsilon_\theta \bigm| \theta_t\bigr].
$$
\end{theorem}
 
\begin{proof}
Write $L_I\equiv\ln I_t$, $L_\theta\equiv\ln\theta_t$, and define the log skill--investment ratio
$$
D \equiv L_\theta - L_I .
$$
The key structural observation is that the investment elasticity is a logistic function of $D$:
$$
\varepsilon_I
=
\frac{\theta_t}{I_t+\theta_t}
=
\frac{1}{1+e^{L_I-L_\theta}}
=
\sigma(D),
\qquad
\sigma(x)\equiv\frac{1}{1+e^{-x}},
$$
and the logistic function satisfies the identity $\sigma'=\sigma(1-\sigma)$ with $0<\sigma'\le \tfrac14$.
 
By joint normality and $\Var(L_\theta)>0$, the conditional distribution of $L_I$ given $L_\theta$ admits the linear representation
$$
L_I
=
\E[L_I]+\beta_t\bigl(L_\theta-\E[L_\theta]\bigr)+s\,Z,
\qquad
Z\sim N(0,1),\ Z\perp L_\theta,
$$
where $s^2\equiv\Var(L_I)-\beta_t^2\Var(L_\theta)\ge 0$. Hence, conditional on $L_\theta=u$,
$$
D
=
u-L_I
\;\overset{d}{=}\;
(1-\beta_t)\,u+c+sZ,
\qquad
c\equiv \beta_t\,\E[L_\theta]-\E[L_I],
$$
using the symmetry $-Z\overset{d}{=}Z$. Conditioning on $\theta_t$ is equivalent to conditioning on $L_\theta=\ln\theta_t$, so the object of interest is
$$
M(u)
\equiv
\E\bigl[\varepsilon_I \bigm| L_\theta=u\bigr]
=
\E\bigl[\sigma\bigl((1-\beta_t)u+c+sZ\bigr)\bigr].
$$
The conditional mean of $D$ moves one-for-one with $u$ up to the passthrough coefficient $(1-\beta_t)$; the theorem is now a one-line differentiation.
 
\emph{Interchange of derivative and expectation.} Fix $u\in\mathbb{R}$ and $h\neq 0$. By the mean value theorem applied to $\sigma$, for every realization of $Z$,
$$
\left|
\frac{\sigma\bigl((1-\beta_t)(u+h)+c+sZ\bigr)-\sigma\bigl((1-\beta_t)u+c+sZ\bigr)}{h}
\right|
\le
|1-\beta_t|\,\sup_{x\in\mathbb{R}}\sigma'(x)
=
\frac{|1-\beta_t|}{4}.
$$
The bound is a finite constant, free of $Z$, $u$, and $h$, hence trivially integrable; dominated convergence therefore yields
$$
M'(u)
=
(1-\beta_t)\,
\E\bigl[\sigma'\bigl((1-\beta_t)u+c+sZ\bigr)\bigr].
$$
 
\emph{Conclusion.} By the logistic identity, $\sigma'(D)=\sigma(D)\bigl(1-\sigma(D)\bigr)=\varepsilon_I\,\varepsilon_\theta$, and since $(1-\beta_t)u+c+sZ$ has the conditional distribution of $D$ given $L_\theta=u$,
$$
\frac{d}{d\ln\theta_t}\,\E\bigl[\varepsilon_I\bigm|\theta_t\bigr]
=
M'(\ln\theta_t)
=
(1-\beta_t)\,\E\bigl[\varepsilon_I\,\varepsilon_\theta\bigm|\theta_t\bigr].
$$
The second displayed form follows from $\varepsilon_I\varepsilon_\theta=\theta_t I_t/(I_t+\theta_t)^2$.
\end{proof}

\subsection{Proposition \ref{prop:translog_bias}}\label{sec:proofs:translog_bias}


\begin{proposition}[Omitted-quadratic bias in the translog-like interaction coefficient]\label{prop:translog_bias}
Let $x=\ln\theta_t$ and $y=\ln I_t$, collect the quadratic regressors as
$$
R=(1,\,x,\,y,\,xy,\,x^{2},\,y^{2})',
$$
and suppose the next-period outcome satisfies
$$
\ln\theta_{t+1}=g(x,y)+\eta .
$$
Maintain the following assumptions.
\begin{enumerate}
\item[(i)] \emph{Moments.} $\E[x^{4}]<\infty$ and $\E[y^{4}]<\infty$; consequently every component of $R$ has finite second moment.
\item[(ii)] \emph{Linear independence.} The regressors $1$, $x$, $y$, $xy$, $x^{2}$, $y^{2}$ are linearly independent as random variables: the only $a\in\mathbb{R}^{6}$ with $a'R=0$ almost surely is $a=0$. Equivalently, $\E[RR']$ has full rank, so there is no perfect multicollinearity among the quadratic regressors. 
\item[(iii)] \emph{Outcome and shock.} $\E\big[(\ln\theta_{t+1})^{2}\big]<\infty$; the shock $\eta$ is integrable and mean-independent of current inputs, $\E[\eta\mid x,y]=0$.
\end{enumerate}

Define the objects of interest as follows. Let $\beta_{\theta,I}^{\mathrm{full}}$ denote the coefficient on $xy$ in the population least-squares regression of $\ln\theta_{t+1}$ on $(1,x,y,xy,x^{2},y^{2})$ (the \emph{full translog} regression), and $\beta_{\theta,I}^{\mathrm{like}}$ the coefficient on $xy$ in the population regression of $\ln\theta_{t+1}$ on $(1,x,y,xy)$ (the \emph{translog-like} regression). Let
$$
P_{2}g(x,y)=c_{0}+c_{x}x+c_{y}y+c_{xy}\,xy+c_{xx}\,x^{2}+c_{yy}\,y^{2}
$$
be the best quadratic predictor of the technology, i.e.\ the population least-squares projection of $g$ onto polynomials of degree at most two in $(x,y)$; if $g$ is itself quadratic (an exact translog), then $P_{2}g=g$ and the $c$'s are its structural coefficients. For $j\in\{x,y\}$, let $\pi_{jj}$ denote the coefficient on $xy$ in the population regression of the omitted regressor $j^{2}$ on the included regressors $(1,x,y,xy)$. Finally, for any square-integrable random variable $w$, write $\widetilde{w}$ for the residual from the population regression of $w$ on $(1,x,y)$. Assumption (ii) guarantees that all of these coefficients are uniquely determined and that $\operatorname{Var}(\widetilde{xy})>0$.

Then:
\begin{enumerate}
\item[(a)] \emph{The full regression is clean:} $\beta_{\theta,I}^{\mathrm{full}}=c_{xy}$.
\item[(b)] \emph{The truncated regression is contaminated:}
$$
\beta_{\theta,I}^{\mathrm{like}}=c_{xy}+\pi_{xx}\,c_{xx}+\pi_{yy}\,c_{yy},
\qquad\text{hence}\qquad
\beta_{\theta,I}^{\mathrm{like}}-\beta_{\theta,I}^{\mathrm{full}}
=\pi_{xx}\,c_{xx}+\pi_{yy}\,c_{yy}.
$$
\item[(c)] \emph{Covariance form of the contamination weights:}
$$
\pi_{xx}=\frac{\operatorname{Cov}\!\big(\widetilde{xy},\,\widetilde{x^{2}}\big)}{\operatorname{Var}\!\big(\widetilde{xy}\big)},
\qquad
\pi_{yy}=\frac{\operatorname{Cov}\!\big(\widetilde{xy},\,\widetilde{y^{2}}\big)}{\operatorname{Var}\!\big(\widetilde{xy}\big)} .
$$
\item[(d)] \emph{Invariance to centering.} Let $\dot x=x-\E[x]$ and $\dot y=y-\E[y]$, and residualize centered quantities on $(1,\dot x,\dot y)$. Then
$$
\widetilde{\dot x\dot y}=\widetilde{xy},
\qquad
\widetilde{\dot x^{2}}=\widetilde{x^{2}},
\qquad
\widetilde{\dot y^{2}}=\widetilde{y^{2}}
\quad\text{almost surely,}
$$
and the coefficients $c_{xy}$, $c_{xx}$, $c_{yy}$, $\beta_{\theta,I}^{\mathrm{full}}$, $\beta_{\theta,I}^{\mathrm{like}}$, $\pi_{xx}$, and $\pi_{yy}$ are identical whether computed from $(x,y)$ or from $(\dot x,\dot y)$. In particular, every moment in (c) may be evaluated after demeaning the logs, which is convenient under distributional assumptions such as joint normality.
\end{enumerate}
\end{proposition}

\begin{proof}
The argument has four steps.

\medskip\noindent\emph{Step 0: integrability, and reduction from the outcome to the technology.}
Since $\ln\theta_{t+1}$ is square-integrable, conditional Jensen's inequality gives
$\E\big[\E[\ln\theta_{t+1}\mid x,y]^{2}\big]\le\E\big[(\ln\theta_{t+1})^{2}\big]<\infty$.
Taking conditional expectations in $\ln\theta_{t+1}=g(x,y)+\eta$ and using $\E[\eta\mid x,y]=0$ yields
$$
g(x,y)=\E[\ln\theta_{t+1}\mid x,y]\quad\text{a.s.},
$$
so $g(x,y)$ is square-integrable, and hence so is $\eta=\ln\theta_{t+1}-g(x,y)$. 

Next, $\eta$ is orthogonal to every quadratic regressor: each component $R_{k}$ of $R$ is square-integrable by (i) and is a function of $(x,y)$, so $R_{k}\eta$ is integrable by the Cauchy--Schwarz inequality and, by iterated expectations,
$$
\E[R_{k}\,\eta]=\E\big[R_{k}\,\E[\eta\mid x,y]\big]=0 .
$$
Hence for any subvector $S$ of $R$ we have $\E[S\,\ln\theta_{t+1}]=\E[S\,g(x,y)]$, while $\E[SS']$ is the same in both regressions; the population regression of $\ln\theta_{t+1}$ on $S$ therefore has exactly the same coefficients as the population regression of $g(x,y)$ on $S$. For the rest of the proof we may replace the outcome by $g$; in particular, the best quadratic predictor of the outcome coincides with $P_{2}g$.

\medskip\noindent\emph{Step 1: the full regression recovers the best quadratic predictor.}
Partition $R_{1}=(1,x,y,xy)'$ and $R_{2}=(x^{2},y^{2})'$, and let $\mathcal{Q}\subset L^{2}$ denote the linear span of the components of $R$: the space of polynomials of degree at most two in $(x,y)$, viewed as random variables. By (ii) the components of $R$ form a basis of $\mathcal{Q}$, and both $\E[RR']$ and $\E[R_{1}R_{1}']$ are invertible, so all projection coefficients below are uniquely determined. The full regression projects $g$ onto $\mathcal{Q}$, and the projection of $g$ onto $\mathcal{Q}$ is, by definition, $P_{2}g$. Reading off the (unique) coordinate on $xy$ gives $\beta_{\theta,I}^{\mathrm{full}}=c_{xy}$, which is (a).

\medskip\noindent\emph{Step 2: omitted-variable algebra for the truncated regression.}
Write $g=P_{2}g+\delta$ with $\delta\perp\mathcal{Q}$; in particular $\delta$ is orthogonal to every component of $R_{1}$, so $g$ and $P_{2}g$ have the same projection on $R_{1}$. Writing $P_{2}g=R_{1}'\gamma_{1}+R_{2}'\gamma_{2}$ with $\gamma_{1}=(c_{0},c_{x},c_{y},c_{xy})'$ and $\gamma_{2}=(c_{xx},c_{yy})'$, the translog-like coefficient vector is
$$
b=\E[R_{1}R_{1}']^{-1}\E\big[R_{1}\,P_{2}g\big]
=\E[R_{1}R_{1}']^{-1}\big(\E[R_{1}R_{1}']\,\gamma_{1}+\E[R_{1}R_{2}']\,\gamma_{2}\big)
=\gamma_{1}+\Pi\,\gamma_{2},
$$
where $\Pi:=\E[R_{1}R_{1}']^{-1}\E[R_{1}R_{2}']$. The two columns of the $4\times 2$ matrix $\Pi$ are precisely the coefficient vectors from the population regressions of $x^{2}$ and of $y^{2}$ on $R_{1}$; its row corresponding to the regressor $xy$ is therefore $(\pi_{xx},\pi_{yy})$ by the definition of $\pi_{jj}$. Reading off the $xy$-coordinate of $b$ gives
$$
\beta_{\theta,I}^{\mathrm{like}}=c_{xy}+\pi_{xx}\,c_{xx}+\pi_{yy}\,c_{yy},
$$
and subtracting (a) yields the bias formula, proving (b).

\medskip\noindent\emph{Step 3: covariance form of the contamination weights.}
Fix $j\in\{x,y\}$ and apply the Frisch--Waugh--Lovell theorem to the regression of $j^{2}$ on $\big((1,x,y),\,xy\big)$: the coefficient on $xy$ equals the slope of the simple regression of the residualized dependent variable $\widetilde{j^{2}}$ on the residualized regressor $\widetilde{xy}$,
$$
\pi_{jj}
=\frac{\E\big[\widetilde{xy}\;\widetilde{j^{2}}\big]}{\E\big[\widetilde{xy}^{\,2}\big]} .
$$
Because the first-stage regression includes a constant, both residuals have mean zero, so the numerator equals $\operatorname{Cov}\big(\widetilde{xy},\widetilde{j^{2}}\big)$ and the denominator equals $\operatorname{Var}\big(\widetilde{xy}\big)$; the latter is strictly positive by (ii), since $xy$ does not lie in the span of $(1,x,y)$. This proves (c). 

\medskip\noindent\emph{Step 4: invariance to centering.}
With $\dot x=x-\E[x]$ and $\dot y=y-\E[y]$ we have the exact identities
$$
\dot x\dot y=xy-\E[y]\,x-\E[x]\,y+\E[x]\E[y],
\qquad
\dot x^{2}=x^{2}-2\E[x]\,x+\E[x]^{2},
\qquad
\dot y^{2}=y^{2}-2\E[y]\,y+\E[y]^{2},
$$
so each centered second-order monomial equals its raw counterpart plus an element of $\operatorname{span}\{1,x,y\}$, and consequently
$$
\operatorname{span}\{1,\dot x,\dot y\}=\operatorname{span}\{1,x,y\},
\qquad
\operatorname{span}\{1,\dot x,\dot y,\dot x\dot y\}=\operatorname{span}\{1,x,y,xy\},
\qquad
\operatorname{span}\{1,\dot x,\dot y,\dot x\dot y,\dot x^{2},\dot y^{2}\}=\mathcal{Q}.
$$
Two consequences follow.

First, the residuals in (c) are literally unchanged. Residualizing on $(1,\dot x,\dot y)$ is residualizing on the same subspace as $(1,x,y)$, and the discrepancies $\dot x\dot y-xy$ and $\dot j^{\,2}-j^{2}$ lie in that subspace, so residualization annihilates them: $\widetilde{\dot x\dot y}=\widetilde{xy}$ and $\widetilde{\dot j^{\,2}}=\widetilde{j^{2}}$ almost surely. By (c), the weights $\pi_{xx}$ and $\pi_{yy}$ are therefore identical in the raw and centered parameterizations.

Second, because the regressor spans coincide, each projection of $g$---translog-like or full---is the same random variable whether computed from raw or centered regressors; only its coordinates can differ between the two bases. Substituting the identities above into the centered representation of the common fitted value and collecting monomials, the only second-order monomial contributed by $\dot x\dot y$ is $xy$ (with unit weight), and likewise $\dot x^{2}$ contributes only $x^{2}$ and $\dot y^{2}$ only $y^{2}$; every remaining correction has degree at most one and lies in $\operatorname{span}\{1,x,y\}$. Uniqueness of coordinates in the basis $(1,x,y,xy)$---respectively $(1,x,y,xy,x^{2},y^{2})$---then forces the centered coefficient on $\dot x\dot y$ to equal the raw coefficient on $xy$, and in the full case likewise for $\dot x^{2}$ and $\dot y^{2}$. Hence $\beta_{\theta,I}^{\mathrm{like}}$, $\beta_{\theta,I}^{\mathrm{full}}$, and $(c_{xy},c_{xx},c_{yy})$ are invariant to centering, and (d) follows.
\end{proof}

\end{document}